\documentclass[12pt,english,aps,manuscript,aps,manuscript]{revtex4}
\usepackage[latin9]{inputenc}
\usepackage{verbatim}
\usepackage{longtable}
\usepackage{float}
\usepackage{calc}
\usepackage{amsthm}
\usepackage{amsmath}
\usepackage{amssymb}
\usepackage{graphicx}
\usepackage{esint}
\PassOptionsToPackage{normalem}{ulem}
\usepackage{ulem}

\makeatletter

\providecommand{\tabularnewline}{\\}

\@ifundefined{textcolor}{}
{%
 \definecolor{BLACK}{gray}{0}
 \definecolor{WHITE}{gray}{1}
 \definecolor{RED}{rgb}{1,0,0}
 \definecolor{GREEN}{rgb}{0,1,0}
 \definecolor{BLUE}{rgb}{0,0,1}
 \definecolor{CYAN}{cmyk}{1,0,0,0}
 \definecolor{MAGENTA}{cmyk}{0,1,0,0}
 \definecolor{YELLOW}{cmyk}{0,0,1,0}
}
 \theoremstyle{definition}
 \newtheorem*{defn*}{\protect\definitionname}
  \theoremstyle{plain}
  \newtheorem{lem}{\protect\lemmaname}
  \theoremstyle{plain}
  \newtheorem*{thm*}{\protect\theoremname}
  \theoremstyle{plain}
  \newtheorem*{lem*}{\protect\lemmaname}

\makeatother

\usepackage{babel}
  \providecommand{\definitionname}{Definition}
  \providecommand{\lemmaname}{Lemma}
  \providecommand{\theoremname}{Theorem}

\begin{document}

\title{{\Large Isotropic Entanglement }}

\author{\textup{Ramis Movassagh}}

\email{ramis@math.mit.edu}

\affiliation{Department of Mathematics, Massachusetts Institute of Technology}

\author{\textup{Alan Edelman}}

\email{edelman@math.mit.edu}

\affiliation{Department of Mathematics, Massachusetts Institute of Technology}

\date{\today}
\begin{abstract}
The method of \textbf{``Isotropic Entanglement''} (IE), inspired
by Free Probability Theory and Random Matrix Theory predicts the eigenvalue
distribution of quantum many-body (spin) systems with generic interactions.
At the heart is a ``Slider'', which interpolates between two extrema
by matching fourth moments. The first extreme treats the non-commuting
terms \textit{classically} and the second treats them \textit{isotropically.
}Isotropic means that the eigenvectors are in generic positions. We
prove Matching Three Moments and Slider Theorems and further prove
that the interpolation is universal, i.e., independent of the choice
of local terms. Our examples show that IE provides an accurate picture
well beyond what one expects from the first four moments alone.\textit{ }
\end{abstract}
\maketitle
\tableofcontents{}

\section{\label{sec:The-Elusive-Spectra} Elusive Spectra of Hamiltonians}

Energy \textit{eigenvalue distributions} or the \textit{density of
states (DOS)} are needed for calculating the partition function\cite[p. 14]{kadanoff}.
The DOS plays an important role in the theory of solids, where it
is used to calculate various physical properties such the internal
energy, the density of particles, specific heat capacity, and thermal
conductivity \cite{mardar,wen-1}. Quantum Many-Body Systems (QMBS)
spectra have been elusive for two reasons: 1. The terms that represent
the interactions are generally non-commuting. This is pronounced for
systems with random interactions (e.g., quantum spin glasses \cite[p. 320]{sachdev}\cite{spinGlass,sachdev2}).
2. Standard numerical diagonalization is limited by memory and computer
speed. Calculation of the spectrum of interacting QMBS has been shown
to be difficult \cite{schuch}. 

An accurate description of tails of distributions are desirable for
condensed matter physics. Though we understand much about the ground
states of interacting QMBS \cite{cirac,vidal,white,wen,frank,hastings,ramis},
eigenvalue distributions are less studied. Isotropic Entanglement
(IE) provides a direct method for obtaining eigenvalue distributions
of quantum spin systems with generic local interactions and does remarkably
well in capturing the tails. Indeed interaction is the very source
of entanglement generation \cite[Section 2.4.1]{NC}\cite{Preskill}
which makes QMBS a resource for quantum computation \cite{gottesman}
but their study a formidable task on a classical computer. 

Suppose we are interested in the eigenvalue distribution of a sum
of Hermitian matrices $M=\sum_{i=1}^{N}M_{i}.$ In general, $M_{i}$
cannot be simultaneously diagonalized, consequently the spectrum of
the sum is not the sum of the spectra. Summation mixes the entries
in a very complicated manner that depends on eigenvectors. Nevertheless,
it seems possible that a one-parameter approximation might suffice. 

Though we are not restricted to one dimensional chains, for sake of
concreteness, we investigate $N$ interacting $d$-dimensional quantum
spins (qudits) on a line with generic interactions. The Hamiltonian
is

\vspace{-0.15in}

\begin{equation}
H=\sum_{l=1}^{N-1}\mathbb{I}_{d^{l-1}}\otimes H_{l,\cdots,l+L-1}\otimes\mathbb{I}_{d^{N-l-\left(L-1\right)}},\label{eq:Hamiltonian}
\end{equation}
where the local terms $H_{l,\cdots,l+L-1}$ are finite $d^{L}\times d^{L}$
random matrices. We take the case of nearest neighbors interactions,
$L=2$, unless otherwise specified. 

The eigenvalue distribution of any commuting subset of $H$ such as
the terms with $l$ odd (the ``odds'') or $l$ even (the ``evens'')
can be obtained using local diagonalization. However, the difficulty
in approximating the full spectrum of $H\equiv H_{\mbox{odd}\vphantom{\mbox{even}}}+H_{\mbox{even}\vphantom{\mbox{odd}}}$
is in summing the odds and the evens because of their overlap at every
site. 

The intuition behind IE is that terms with an overlap, such as $H_{l,l+1}$
and $H_{l+1,l+2}$, introduce randomness and mixing through sharing
of a site. Namely, the process of entanglement generation introduces
an \textit{isotropicity }between the eigenvectors of evens and odds
that can be harnessed to capture the spectrum.

\begin{figure}
\centering{}\includegraphics[scale=1.5]{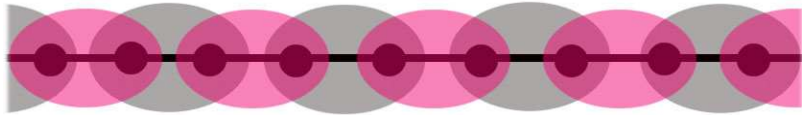}\caption{\label{fig:Odd-and-even}Odd and even summands can separately be locally
diagonalized, but not the sum. The overlap of the two subsets at every
site generally requires a global diagonalization.}
\end{figure}

Random Matrix Theory (RMT) often takes advantage of eigenvectors with
Haar measure, the uniform measure on the orthogonal/unitary group.
However, the eigenvectors of QMBS have a more special structure (see
Eq. \ref{eq:OddEven}).

Therefore we created a \textit{hybrid theory}, where we used a finite
version of Free Probability Theory (FPT) and Classical Probability
Theory to capture the eigenvalue distribution of Eq. \ref{eq:Hamiltonian}.
Though such problems can be QMA-complete, our examples show that IE
provides an accurate picture well beyond what one expects from the
first four moments alone.\textit{ }The \textit{Slider} (bottom of
Figure \ref{fig:RoadMap}) displays the proposed mixture $p$.

\section{\label{sec:The-Method}The Method of Isotropic Entanglement}

\subsection{\label{sub:Road-Map} Overview}

\begin{figure}
\centering{}\includegraphics[scale=0.6]{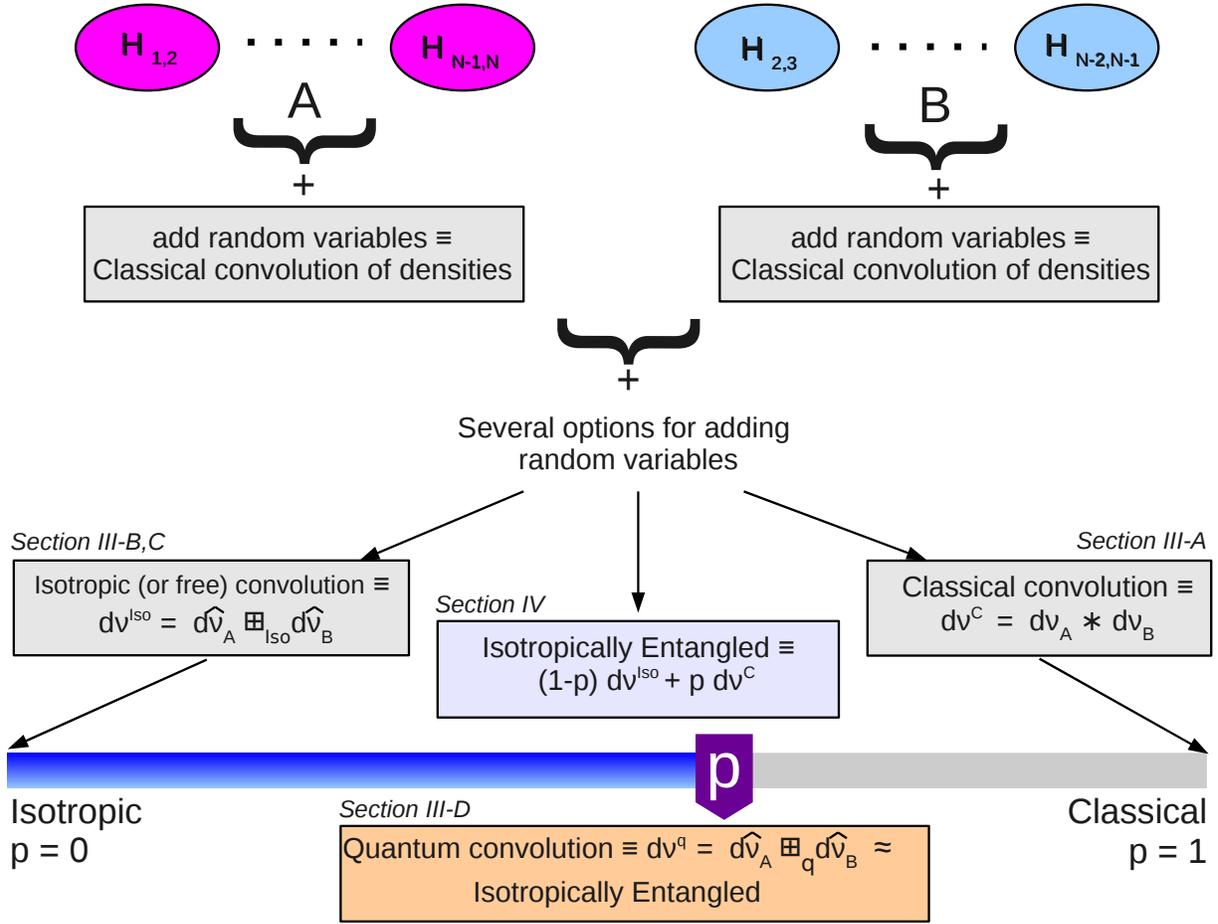}\caption{\label{fig:RoadMap}The method of Isotropic Entanglement: Quantum
spectra as a convex combination of isotropic and classical distributions.
The Slider (bottom) indicates the $p$ that matches the quantum kurtosis
as a function of classical ($p=1$) and isotropic ($p=0$) kurtoses.
To simplify we drop the tensor products (Eq. \ref{eq:HevenHodd})
in the local terms (ellipses on top). Note that isotropic and quantum
convolution depend on multivariate densities for the eigenvalues. }
\end{figure}

We propose a method to compute the ``density of states'' (DOS) or
``eigenvalue density'' of quantum spin systems with generic local
interactions. More generally one wishes to compute the DOS of the
sum of non-commuting random matrices from their, individually known,
DOS's. 

We begin with an example in Figure \ref{fig:The_intro}, where we
compare exact diagonalization against two approximations:

\begin{figure}
\centering{}\includegraphics[width=12cm]{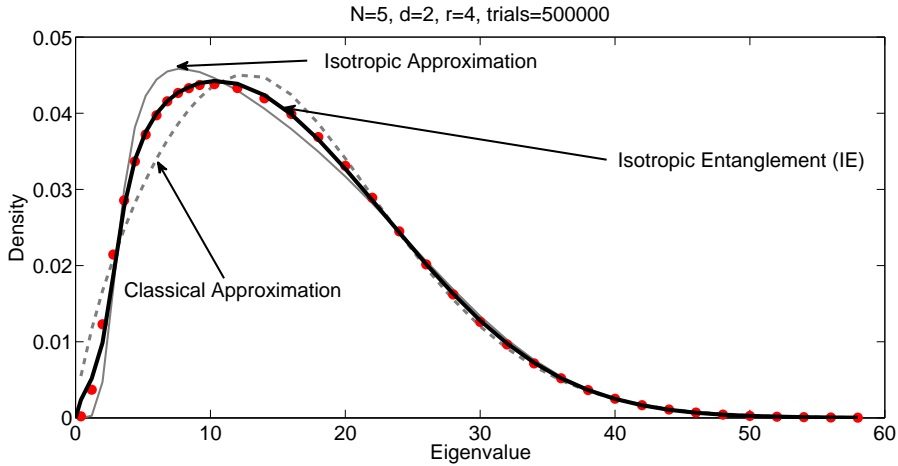}\caption{\label{fig:The_intro}The exact diagonalization in dots and IE compared
to the two approximations. The title parameters are explained in the
section on numerical results.}
\end{figure}

\begin{itemize}
\item Dashed grey curve: \textit{classical} approximation. Notice that it
overshoots to the right.
\item Solid grey curve: \textit{isotropic} approximation. Notice that it
overshoots to the left.
\item Solid black curve: \textit{isotropic entanglement (IE).}
\item Dots: \textit{exact diagonalization} of the quantum problem given
in Eq. \ref{eq:Hamiltonian}.
\end{itemize}
The \textit{classical approximation} ignores eigenvector structure
by summing random eigenvalues uniformly from non-commuting matrices.
The dashed curve is the convolution of the probability densities of
the eigenvalues of each matrix. 

The\textit{ isotropic approximation} assumes that the eigenvectors
are in ``general position''; that is, we add the two matrices with
correct eigenvalue densities but choose the eigenvectors from Haar
measure. As the matrix size goes to infinity, the resulting distribution
is the free convolution of the individual distributions \cite{speicher}. 

The exact diagonalization given by red dots, the dashed and solid
grey curves have exactly the same first three moments, but differing
fourth moments. 

\textit{Isotropic Entanglement (IE)} is a linear combination of the
two approximations that is obtained by matching the fourth moments.
We show that 1) the fit is better than what might be expected by four
moments alone, 2) the combination is always convex for the problems
of interest, given by $0\leq p\leq1$ and 3) this convex combination
is universal depending on the parameter counts of the problem but
not the eigenvalue densities of the local terms.

\textit{Parameter counts: exponential, polynomial and zero.} Because
of the \textit{locality} of generic interactions, the complete set
of eigenstates has parameter count equal to a polynomial in the number
of spins, though the dimensionality is exponential. The classical
and isotropic approximations have zero and exponentially many random
parameters respectively. This suggests that the problem of interest
somehow lies in between the two approximations. 

Our work supports a very general principle that one can obtain an
accurate representation of inherently exponential problems by approximating
them with less complexity. This realization is at the heart of other
recent developments in QMBS research such as Matrix Product States
\cite{cirac,vidal}, and Density Matrix Renormalization Group \cite{white},
where the \textit{state} (usually the ground state of $1D$ chains)
can be adequately represented by a Matrix Product State (MPS) ansatz
whose parameters grow \textit{linearly} with the number of quantum
particles. Future work includes explicit treatment of fermionic systems
and numerical exploration of higher dimensional systems.

\subsection{Inputs and Outputs of the Theory}

In general we consider Hamiltonians $H=H_{\mbox{odd}\mbox{\ensuremath{\vphantom{even}}}}+H_{\mbox{even}\mbox{\ensuremath{\vphantom{odd}}}}$,
where the local terms that add up to $H_{\mbox{odd}\vphantom{\mbox{even}}}$
(or $H_{\mbox{even}\vphantom{\mbox{odd}}}$) form a commuting subset.
All the physically relevant quantities such as the lattice structure,
$N$, dimension of the spin $d$ and the rank $r$ are encoded in
the eigenvalue densities. The output of the theory is a $0\leq p\leq1$
by which the IE distribution is obtained and $d\nu^{IE}$ serves as
an approximation to the spectral measure. The inputs can succinctly
be expressed in terms of the dimension of the quantum spins, and the
nature of the lattice (Figure \ref{tab:Inputs-and-Outputs}).

\begin{figure}
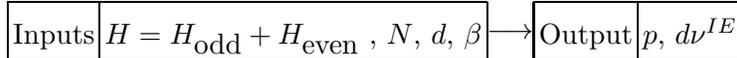

\begin{centering}
\begin{tabular}{|c|c|}
\hline 
Inputs & $H=H_{\mbox{odd}\vphantom{\mbox{even}}}+H_{\mbox{even}\vphantom{\mbox{odd}}}$
, $N$, $d$, $\beta$\tabularnewline
\hline 
\end{tabular}$\longrightarrow$%
\begin{tabular}{|c|c|}
\hline 
Output & $p$, $d\nu^{IE}$\tabularnewline
\hline 
\end{tabular}
\par\end{centering}

\caption{\label{tab:Inputs-and-Outputs}Inputs and outputs of the IE theory.
See section \ref{sec:Spectra-Sums-from-Prob} for the definition of
$d\nu^{IE}$.}
\end{figure}

\subsection{\label{sub:More-Than-Four}More Than Four Moments of Accuracy?}

Alternatives to IE worth considering are 1) Pearson and 2) Gram-Charlier
moment fits. 

We illustrate in Figure \ref{fig:IEvsOthersZERO} how the IE fit is
better than expected when matching four moments. We used the first
four moments to approximate the density using the Pearson fit as implemented
in MATLAB and also the well-known Gram-Charlier fit \cite{Cramer}.
In \cite{Gram} it was demonstrated that the statistical mechanics
methods for obtaining the DOS, when applied to a finite dimensional
vector space, lead to a Gaussian distribution in the lowest order.
Further, they discovered that successive approximations lead naturally
to the Gram-Charlier series \cite{Cramer}. Comparing these against
the accuracy of IE leads us to view IE as more than a moment matching
methodology. 

\begin{figure}
\begin{centering}
\includegraphics[width=8cm]{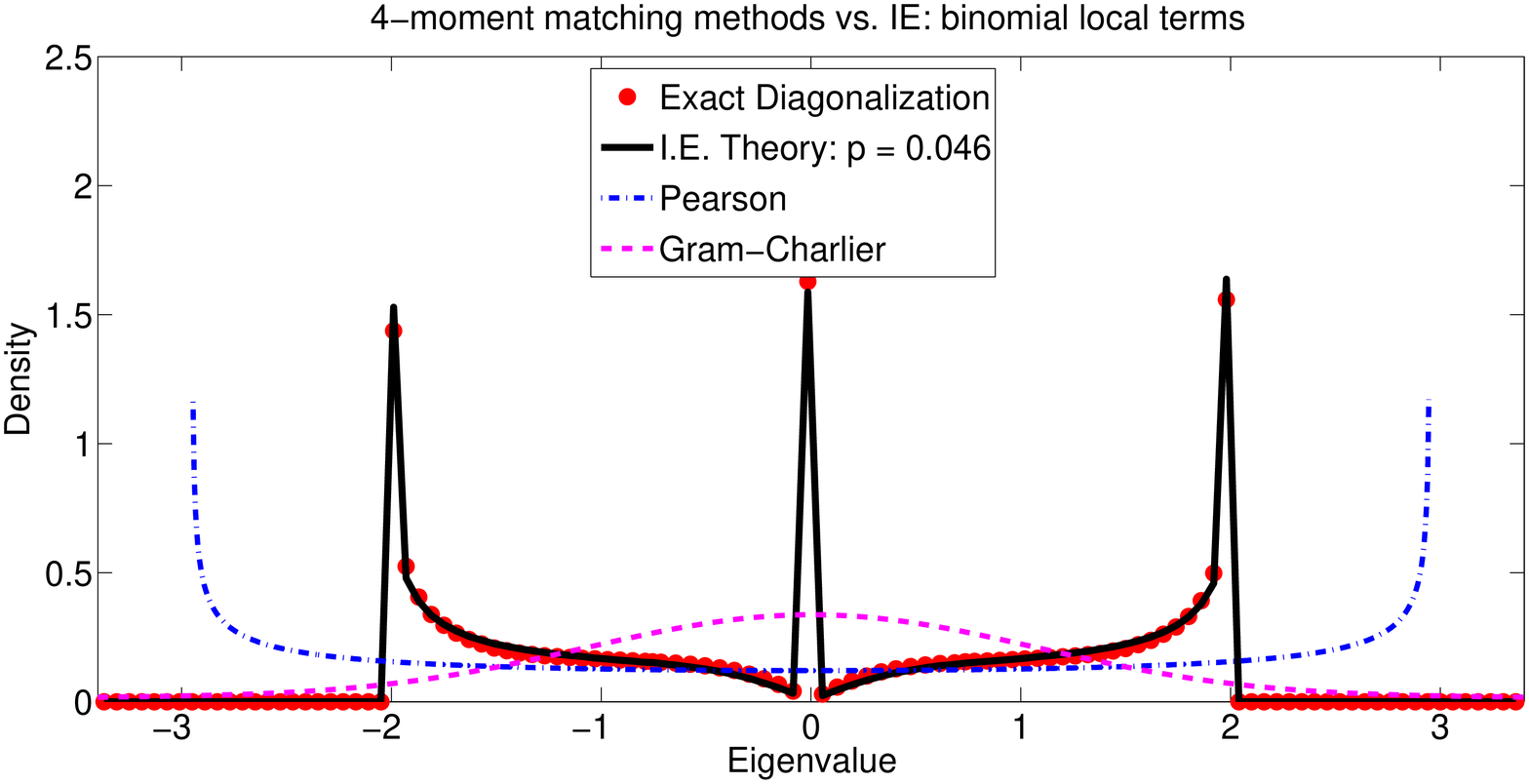}\includegraphics[width=8cm]{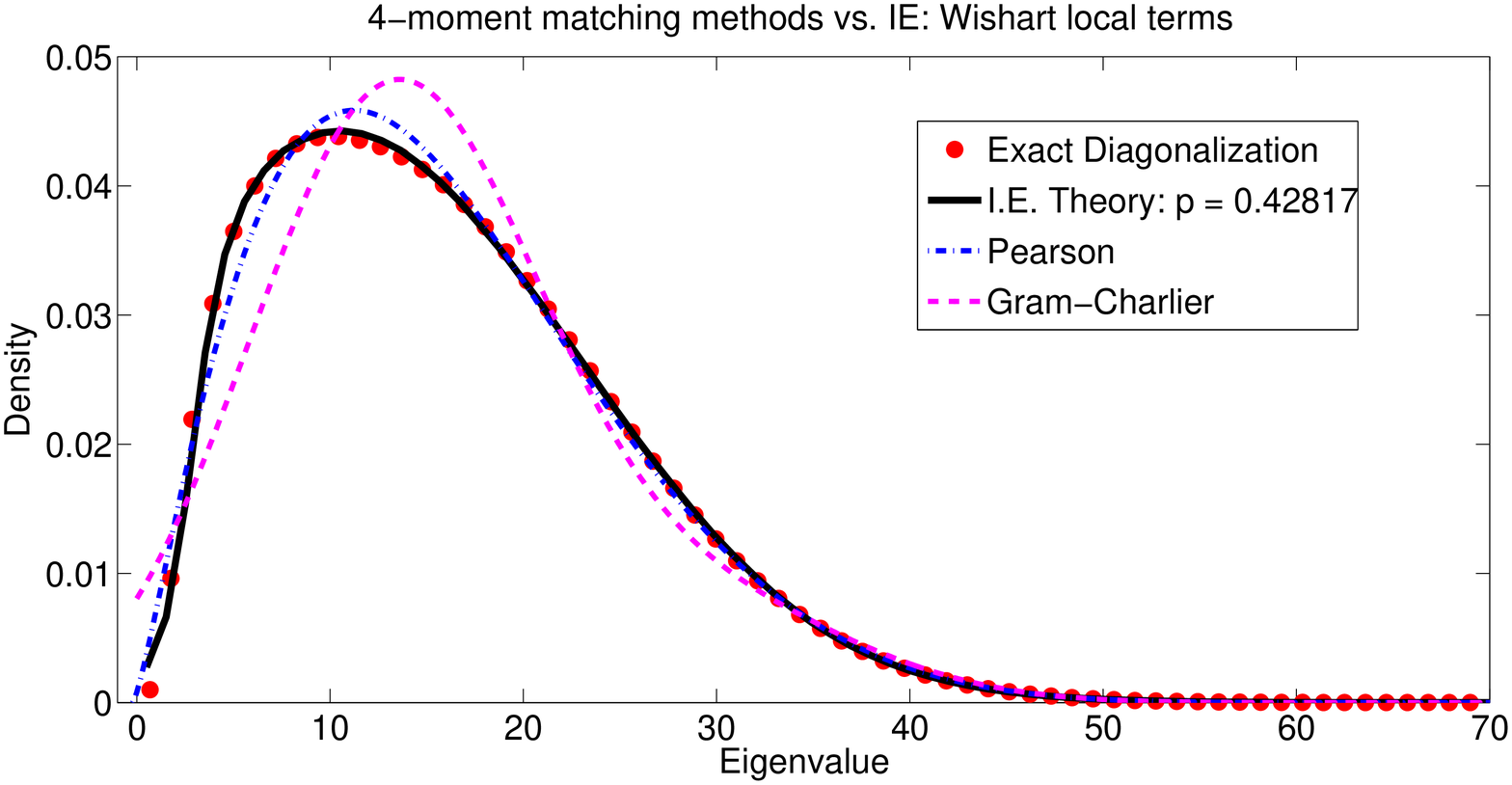}
\par\end{centering}

\centering{}\caption{\label{fig:IEvsOthersZERO}IE vs. Pearson and Gram-Charlier}
\end{figure}

The departure theorem (Section \ref{sub:departure}) shows that in
any of the higher moments ($>4$) there are many terms in the quantum
case that match IE exactly. Further, we conjecture that the effect
of the remaining terms are generally less significant.

\section{\label{sec:Spectra-Sums-from-Prob}Spectra Sums in Terms of Probability
Theory}

The density of eigenvalues may be thought of as a histogram. Formally
for an $m\times m$ matrix $M$ the\textit{ eigenvalue distribution}
is \cite[p. 4]{zeitouni1}\cite[p. 101]{Hiai} 
\begin{equation}
d\nu_{M}(x)=\frac{1}{m}\sum_{i=1}^{m}\delta(x-\lambda_{i}\left(M\right)).
\end{equation}

For a random matrix, there is the expected eigenvalue distribution
\cite{alan}, \cite[p. 362]{speicher}
\begin{equation}
d\nu_{M}(x)=\frac{1}{m}\mathbb{E}\left[\sum_{i=1}^{m}\delta(x-\lambda_{i}\left(M\right))\right],
\end{equation}
 which is typically a smooth curve \cite[p. 101]{zeitouni1}\cite[p. 115]{Hiai}. 

The eigenvalue distributions above are measures on one variable. We
will also need the multivariate measure on the $m$ eigenvalues of
$M$: 
\[
d\hat{\nu}_{M}\left(x\right)\mbox{= The symmetrized joint density of the eigenvalues. }
\]

Given the densities for $M$ and $M',$ the question arises: What
kind of summation of densities might represent the density for $M+M'$?
This question is unanswerable without further information.

One might try to answer this using various assumptions based on probability
theory. The first assumption is the familiar ``classical'' probability
theory where the distribution of the sum is obtained by convolution
of the density of summands. Another assumption is the modern ``free''
probability theory; we introduce a finite version to obtain the ``isotropic''
theory. Our target problem of interest, the ``quantum'' problem,
we will demonstrate, practically falls nicely in between the two.
The ``Slider'' quantifies to what extent the quantum problem falls
in between (Figure \ref{fig:RoadMap} bottom).

\subsection{\label{sub:Classical}Classical}

Consider random diagonal matrices $A$ and $B$ of size $m$, the
only randomness is in a uniform choice among the $m!$ possible orders.
Then there is no difference between the density of eigenvalue sums
and the familiar convolution of densities of random variables,
\begin{equation}
d\nu^{c}=d\nu_{A}*d\nu_{B}.
\end{equation}

Comment: From this point of view, the diagonal elements of $A,$ say,
are identically distributed random variables that need not be independent.
Consider Wishart matrices \cite{wishart}, where there are dependencies
among the eigenvalues. To be precise let $\mathbf{a}\in\mathbb{R}^{m}$
be a symmetric random variable, i.e., $P\mathbf{a}$ has the same
distribution as $\mathbf{a}$ for all permutation matrices $P$. We
write, $A=\left(\begin{array}{ccc}
a_{1}\\
 & \ddots\\
 &  & a_{m}
\end{array}\right)\equiv\textrm{diag}(\mathbf{a})$. Similarly for $B$.

Comment: The classical convolution appears in Figure \ref{fig:RoadMap}
in two different ways. Firstly, in the definition of $A$ (or $B$)
, the eigenvalues of the odd (or even) terms are added classically.
Secondly, $A$ and $B$ are added classically to form one end of the
Slider.

\subsection{\label{sub:Free}Free and Isotropic}

Free probability \cite[is recommended]{speicher} provides a new natural
mathematical ``sum'' of random variables. This sum is computed ``free
convolution'' denoted 

\begin{equation}
d\nu^{f}=d\nu_{A}\boxplus d\nu_{B}.
\end{equation}
Here we assume the random matrices $A$ and $B$, representing the
eigenvalues, have densities $d\nu_{A}$ and $d\nu_{B}$. In the large
$m$ limit, we can compute the DOS of $A+Q^{T}BQ$, where $Q$ is
a $\beta-$Haar distributed matrix (see Table \ref{tab:field}). 

Comment: In this paper we will not explore the free approach strictly
other than observing that it is the infinite limit of the isotropic
approach (i.e., $t\rightarrow\infty$ in Eq. \ref{eq:HIso}). This
infinite limit is independent of the choice of $\beta$ (see Table
\ref{tab:field}).

\begin{table}[H]
\begin{centering}
\begin{tabular}{|c|c|c|c|c|}
\hline 
 & Real $\mathbb{R}$ & Complex $\mathbb{C}$ & Quaternions $\mathbb{H}$ & ``Ghosts'' \tabularnewline
\hline 
\hline 
$\beta$ & $1$ & $2$ & $4$ & general $\beta$\tabularnewline
\hline 
Notation & $Q$ & $U$ & $S$ & $\mathcal{Q_{\beta}}$\tabularnewline
\hline 
Haar matrices & orthogonal & unitary & symplectic & $\beta-$orthogonal\tabularnewline
\hline 
\end{tabular}
\par\end{centering}

\centering{}\caption{\label{tab:field}Various $\beta-$Haar matrices.}
\end{table}

We define an isotropic convolution. The isotropic sum depends on a
copying parameter $t$ and $\beta$ (Table \ref{tab:field}). The
new Hamiltonian is the \textit{isotropic Hamiltonian} (\textit{``iso}''):

\begin{equation}
H_{iso}\equiv\left(A'\otimes\mathbb{I}_{t}\right)+Q_{\beta}^{-1}\left(\mathbb{I}_{t}\otimes B'\right)Q_{\beta},\label{eq:HIso}
\end{equation}
where $Q_{\beta}$ is a $\beta-$Haar distributed matrix, $A=A'\otimes\mathbb{I}_{t}$
and $B=\mathbb{I}_{t}\otimes B'$. For the copying parameter $t=d$,
$H_{iso}$ has the same dimension as $H$ in Eq. \ref{eq:Hamiltonian};
however, $t>d$ allows us to formally treat problems of growing size.
We can recover the free convolution by taking the limit: $\lim_{t\rightarrow\infty}d\nu^{iso\left(\beta,t\right)}=d\nu^{f}$.
The effect of $Q_{\beta}$ is to spin the eigenvectors of $\mathbb{I}_{t}\otimes B$
to point isotropically with respect to the eigenvectors of $A$. We
denote the isotropic eigenvalue distribution by

\begin{equation}
d\nu^{iso\left(\beta,t\right)}=d\hat{\nu}_{A}\boxplus_{iso\left(\beta,t\right)}d\hat{\nu}_{B}
\end{equation}
omitting $t$ and $\beta$ when it is clear from the context. 

Comment: In Eq. \ref{eq:HIso}, the $\mathbb{I}_{t}$ and $B$ in
$\mathbb{I}_{t}\otimes B$, can appear in any order. We chose this
presentation in anticipation of the quantum problem.

Comment: In this paper we primarily consider $t$ to match the dimension
of $H$.

\subsection{\label{sub:Quantum}Quantum}

Let $d\nu^{q}$ denote the eigenvalue distribution for the Hamiltonian
in Eq. \ref{eq:Hamiltonian}. This is the distribution that we will
approximate by $d\nu^{IE}$. In connection to Figure \ref{fig:Odd-and-even}
the Hamiltonian can be written as

\begin{equation}
H\equiv H_{\mbox{odd}\vphantom{\mbox{even}}}+H_{\mbox{even}\mbox{\ensuremath{\vphantom{odd}}}}=\sum_{l=1,3,5,\cdots}\mathbb{I}\otimes H_{l,l+1}\otimes\mathbb{I}+\sum_{l=2,4,6,\cdots}\mathbb{I}\otimes H_{l,l+1}\otimes\mathbb{I}.\label{eq:HevenHodd}
\end{equation}

We proceed to define a ``quantum convolution'' on the distributions
$d\hat{\nu}_{A}$ and $d\hat{\nu}_{B}$, which is $\beta$-dependent

\begin{equation}
d\nu^{q\left(\beta\right)}=d\hat{\nu}_{A}\boxplus_{q}d\hat{\nu}_{B}.
\end{equation}

\begin{flushleft}
In general, without any connection to a Hamiltonian, let $d\hat{\nu}_{A}$
and $d\hat{\nu}_{B}$ be symmetric measures on $\mathbb{R}^{d^{N}}$.
We define $d\nu^{q\left(\beta\right)}$ to be the eigenvalue distribution
of 
\par\end{flushleft}

\begin{equation}
H=A+Q_{q}^{-1}BQ_{q},
\end{equation}

\begin{flushleft}
where $Q_{q}=\left(Q_{q}^{(A)}\right)^{-1}Q_{q}^{(B)}$ with
\par\end{flushleft}

\begin{equation}
\begin{array}{c}
Q_{q}^{(A)}=\left[\bigotimes_{i=1}^{\left(N-1\right)/2}Q_{i}^{(O)}\right]\otimes\mathbb{I}_{d}\;\textrm{and}\; Q_{q}^{(B)}=\mathbb{I}_{d}\otimes\left[\bigotimes_{i=1}^{\left(N-1\right)/2}Q_{i}^{(E)}\right]\qquad N\;\textrm{odd}\\
\\
\qquad Q_{q}^{(A)}=\left[\bigotimes_{i=1}^{N/2}Q_{i}^{(O)}\right]\;\textrm{and}\; Q_{q}^{(B)}=\mathbb{I}_{d}\otimes\left[\bigotimes_{i=1}^{N/2-1}Q_{i}^{(E)}\right]\otimes\mathbb{I}_{d}\qquad N\;\textrm{even}
\end{array}\label{eq:OddEven}
\end{equation}

\begin{flushleft}
and each $Q_{i}^{\left(\bullet\right)}$ is a $\beta-$Haar measure
orthogonal matrix of size $d^{2}$ and $\mathbb{I}_{d}$ is an identity
matrix of size $d$.
\par\end{flushleft}

Comment: $A$, $B$ and $Q_{q}$ are $d^{N}\times d^{N}.$

Comment: In our examples given in this paper, we assume the local
terms are independent and identically distributed (iid) random matrices,
each of which has eigenvectors distributed with $\beta-$Haar measure. 

The tensor product in (\ref{eq:OddEven}) succinctly summarizes the
departure of the quantum case from a generic matrix as well as from
the classical case. First of all the number of parameters in $Q_{q}$
grows linearly with $N$ whereas in $Q$ it grows exponentially with
$N$. Second, the quantum case possesses isotropicity that makes it
different from the classical, whose eigenvectors are a point on the
orthogonal group (i.e., the identity matrix).

Comment: General $\beta$'s can be treated formally \cite{alanGhost}.
In particular, for quantum mechanical problems $\beta$ is taken to
be $1$ or $2$ corresponding to real and complex entries in the local
terms. $\beta=4$ corresponds to quaternions.
\begin{defn*}
\label{def:The-Hadamard-product}The \textbf{Hadamard product} of
two matrices $M_{1}$ and $M_{2}$ of the same size, denoted by $M_{1}\circ M_{2}$,
is the product of the corresponding elements. \end{defn*}
\begin{lem}
\label{lem:var}The elements of $Q_{q}$ defined in Eq. \ref{eq:OddEven}
are (dependent) random variables with mean zero and variance $d^{-N}$.\end{lem}
\begin{proof}
Here expectations are taken with respect to the random matrix $Q_{q}$
which is built from local Haar measure matrices by Eq. \ref{eq:OddEven}.
The fact that $\mathbb{E}\left(Q_{q}^{\left(A\right)}\right)=\mathbb{E}\left(Q_{q}^{\left(B\right)}\right)=0_{d^{N}}$
follows from the Haar distribution of local terms. Thus $\mathbb{E}\left(Q_{q}\right)=0$
by independence of $Q_{q}^{\left(A\right)}$ and $Q_{q}^{\left(B\right)}$.
Further, each element in $Q_{q}$ involves a dot product between columns
of $Q_{q}^{\left(A\right)}$ and $Q_{q}^{\left(B\right)}$. In every
given column of $Q_{q}^{\left(A\right)}$ any nonzero entry is a distinct
product of entries of local $Q's$ (see Eq.\ref{eq:OddEven}). For
example the expectation value of the $1,1$ entry is $\mathbb{E}\left(q_{i,1}^{\left(A\right)}q_{j,1}^{\left(A\right)}q_{i,1}^{\left(B\right)}q_{j,1}^{\left(B\right)}\right)=\mathbb{E}\left(q_{i,1}^{\left(A\right)}q_{j,1}^{\left(A\right)}\right)\mathbb{E}\left(q_{i,1}^{\left(B\right)}q_{j,1}^{\left(B\right)}\right)$.
Because of the Haar measure of the local terms, this expectation is
zero unless $i=j$. We then have that

\begin{equation}
\begin{array}{c}
\mathbb{E}\left(Q_{q}\circ Q_{q}\right)=\mathbb{E}\left(Q_{q}^{\left(A\right)}\circ Q_{q}^{\left(A\right)}\right)^{T}\mathbb{E}\left(Q_{q}^{\left(B\right)}\circ Q_{q}^{\left(B\right)}\right)=\\
\begin{cases}
\begin{array}{c}
\left(\left[\bigotimes_{i=1}^{\left(N-1\right)/2}d^{-2}J_{d^{2}}\right]\otimes\mathbb{I}_{d}\right)\left(\mathbb{I}_{d}\otimes\left[\bigotimes_{i=1}^{\left(N-1\right)/2}d^{-2}J_{d^{2}}\right]\right)\qquad N\;\textrm{odd}\\
\left(\bigotimes_{i=1}^{N/2}d^{-2}J_{d^{2}}\right)\left(\mathbb{I}_{d}\otimes\left[\bigotimes_{i=1}^{N/2-1}d^{-2}J_{d^{2}}\right]\otimes\mathbb{I}_{d}\right)\qquad\qquad\quad N\;\mbox{even}
\end{array}\end{cases}\\
=d^{-N}J_{d^{N}},
\end{array}
\end{equation}
where $J_{i}=i\times i$ matrix of all ones. We use facts such as
$\left(J_{i}/i\right)^{2}=\left(J_{i}/i\right)$, $\left(J_{i}/i\right)\otimes\left(J_{i}/i\right)=\left(J_{i^{2}}/i^{2}\right)$
and the variance of the elements of an $i\times i$ $\beta-$Haar
matrix is $1/i$.
\end{proof}

\section{\label{sec:Theory-of-Isotropic}Theory of Isotropic Entanglement}

\subsection{\label{sub:convex}Isotropic Entanglement as the Combination of Classical
and Isotropic}

We create a ``Slider'' based on the fourth moment. The moment $m_{k}$
of a random variable defined in terms of its density is $m_{k}=\int x^{k}d\nu_{M}.$
For the eigenvalues of an $m\times m$ random matrix, this is $\frac{1}{m}\mathbb{E}\mbox{Tr}M^{k}.$
In general, the moments of the classical sum and the free sum are
different, but the first three moments, $m_{1},\ m_{2},$ and $m_{3}$
are theoretically equal \cite[p. 191]{speicher}. Further, to anticipate
our target problem, the first three moments of the quantum eigenvalues
are also equal to that of the iso and the classical (The Departure
and the Three Moments Matching theorems in Section \ref{sub:departure}).
These moments are usually encoded as the mean, variance, and skewness.

We propose to use the fourth moment (or the excess kurtosis) to choose
a correct $p$ from a sliding hybrid sum:

\begin{equation}
d\nu^{q}\approx d\nu^{IE}=pd\nu^{c}+(1-p)d\nu^{iso}\label{eq:dNuq}
\end{equation}

Therefore, we find $p$ that expresses the kurtosis of the quantum
sum $(\gamma_{2}^{q}$) in terms of the kurtoses of the classical
($\gamma_{2}^{c}$) and isotropic ($\gamma_{2}^{iso}$) sums:

\begin{equation}
\gamma_{2}^{q}=p\gamma_{2}^{c}+\left(1-p\right)\gamma_{2}^{iso}\Rightarrow\qquad p=\frac{\gamma_{2}^{q}-\gamma_{2}^{iso}}{\gamma_{2}^{c}-\gamma_{2}^{iso}}.\label{eq:convex}
\end{equation}
Recall that the kurtosis $\gamma_{2}\equiv\frac{m_{4}}{\sigma^{4}}$,
where $\sigma^{2}$ is the variance. Hence kurtosis is the correct
statistical quantity that encodes the fourth moments:

\begin{equation}
m_{4}^{c}=\frac{1}{d^{N}}\mathbb{E}\textrm{Tr}\left(A+\Pi^{T}B\Pi\right)^{4},\; m_{4}^{iso}=\frac{1}{d^{N}}\mathbb{E}\textrm{Tr}\left(A+Q^{T}BQ\right)^{4},\; m_{4}^{q}=\frac{1}{d^{N}}\mathbb{E}\textrm{Tr}\left(A+Q_{q}^{T}BQ_{q}\right)^{4},\label{eq:moments}
\end{equation}
where $\Pi$ is a random uniformly distributed permutation matrix,
$Q$ is a $\beta-$Haar measure orthogonal matrix of size $d^{N}$,
and $Q_{q}$ is given by Eq. \ref{eq:OddEven}. \textit{Unless stated
otherwise, in the following the expectation values are taken with
respect to random eigenvalues $A$ and $B$ and eigenvectors. The
expectation values over the eigenvectors are taken with respect to
random permutation $\Pi$, $\beta-$Haar $Q$ or $Q_{q}$ matrices
for classical, isotropic or quantum cases respectively. }

\subsection{\label{sub:departure}The Departure and The Matching Three Moments
Theorems}

In general we have the $i^{\textrm{th}}$ moments:

\begin{eqnarray*}
m_{i}^{iso} & = & \frac{1}{m}\mathbb{E}\textrm{Tr}\left(A+Q^{T}BQ\right)^{i}\\
m_{i}^{q} & = & \frac{1}{m}\mathbb{E}\textrm{Tr}\left(A+Q_{q}^{T}BQ_{q}\right)^{i},\mbox{ and }\\
m_{i}^{c} & = & \frac{1}{m}\mathbb{E}\textrm{Tr}\left(A+\Pi^{T}B\Pi\right)^{i}.
\end{eqnarray*}
where $m\equiv d^{N}$. If we expand the moments above we find some
terms can be put in the form $\mathbb{E}\textrm{Tr}\left(A^{m_{1}}Q_{\bullet}^{T}B^{m_{2}}Q_{\bullet}\right)$
and the remaining terms can be put in the form $\mathbb{E}\textrm{Tr}\left\{ \ldots Q_{\bullet}^{T}B^{\ge1}Q_{\bullet}A^{\ge1}Q_{\bullet}^{T}B^{\ge1}Q_{\bullet}\ldots\right\} .$
The former terms we denote \textit{non-departing}; the remaining terms
we denote \textit{departing}.

For example, when $i=4$, 

\begin{eqnarray}
m_{4}^{iso} & = & \frac{1}{m}\mathbb{E}\left\{ \textrm{Tr}\left[A^{4}+4A^{3}Q^{T}BQ+4A^{2}Q^{T}B^{2}Q+4AQ^{T}B^{3}Q+\mathbf{\underline{2\left(\mathbf{AQ^{T}BQ}\right)^{2}}}+B^{4}\right]\right\} \label{eq:fourthMoments}\\
m_{4}^{q} & = & \frac{1}{m}\mathbb{E}\left\{ \textrm{Tr}\left[A^{4}+4A^{3}Q_{q}^{T}BQ_{q}+4A^{2}Q_{q}^{T}B^{2}Q_{q}+4AQ_{q}^{T}B^{3}Q_{q}+\underline{\mathbf{2\left(AQ_{q}^{T}BQ_{q}\right)^{2}}}+B^{4}\right]\right\} \nonumber \\
m_{4}^{c} & = & \frac{1}{m}\mathbb{E}\left\{ \textrm{Tr}\left[A^{4}+4A^{3}\Pi^{T}B\Pi+4A^{2}\Pi^{T}B^{2}\Pi+4A\Pi^{T}B^{3}\Pi+\mathbf{\underline{2\left(A\Pi^{T}B\Pi\right)^{2}}}+B^{4}\right]\right\} ,\nonumber 
\end{eqnarray}
where the only departing terms and the corresponding classical term
are shown as underlined and bold faced.
\begin{thm*}
\textbf{\textup{\label{thm:(The-Departure-Theorem)}(The Departure
Theorem)}}\textbf{ }The moments of the quantum, isotropic and classical
sums differ only in the departing terms: $\mathbb{E}\textrm{Tr}\left\{ \ldots Q_{\bullet}^{T}B^{\ge1}Q_{\bullet}A^{\ge1}Q_{\bullet}^{T}B^{\ge1}Q_{\bullet}\ldots\right\} .$ \end{thm*}
\begin{proof}
Below the repeated indices are summed over. If $A$ and $B$ are any
diagonal matrices, and $Q_{\bullet}$ is $Q$ or $Q_{q}$ or $\Pi$
of size $m\times m$ then $\mathbb{E}\left(q_{ij}^{2}\right)=1/m$
, by symmetry and by Lemma \ref{lem:var} respectively. Since the
$\mathbb{E}\textrm{Tr}\left(AQ_{\bullet}^{T}BQ_{\bullet}\right)=\mathbb{E}\left(q_{ij}^{2}a_{i}b_{j}\right)$,
where expectation is taken over randomly ordered eigenvalues and eigenvectors;
the expected value is $m^{2}\left(\frac{1}{m}\right)\mathbb{E}\left(a_{i}b_{j}\right)$
for any $i$ or $j$. Hence, $\frac{1}{m}\mathbb{E}\textrm{Tr}\left(AQ_{\bullet}^{T}BQ_{\bullet}\right)=\mathbb{E}\left(a_{i}b_{j}\right)=\mathbb{E}\left(a_{i}\right)\mathbb{E}\left(b_{j}\right)$,
which is equal to the classical value. The first equality is implied
by permutation invariance of entries in $A$ and $B$ and the second
equality follows from the independence of $A$ and $B$. 
\end{proof}
\begin{flushleft}
Therefore, the three cases differ only in the terms $\frac{2}{m}\mathbb{E}\textrm{Tr}\left(AQ^{T}BQ\right)^{2}$,
$\frac{2}{m}\mathbb{E}\textrm{Tr}\left(AQ_{q}^{T}BQ_{q}\right)^{2}$
and $\frac{2}{m}\mathbb{E}\textrm{Tr}\left(A\Pi^{T}B\Pi\right)^{2}$
in Eq. \ref{eq:fourthMoments}.
\par\end{flushleft}
\begin{thm*}
\textbf{\textup{\label{thm:(The-Matching-Three}(The Matching Three
Moments Theorem) }}The first three moments of the quantum, iso and
classical sums are equal.\end{thm*}
\begin{proof}
The first three moments are
\begin{equation}
\begin{array}{c}
m_{1}^{\left(\bullet\right)}=\frac{1}{m}\mathbb{E}\textrm{Tr}\left(A+B\right)\\
m_{2}^{\left(\bullet\right)}=\frac{1}{m}\mathbb{E}\textrm{Tr}\left(A+Q_{\bullet}^{T}BQ_{\bullet}\right)^{2}=\frac{1}{m}\mathbb{E}\textrm{Tr}\left(A^{2}+2AQ_{\bullet}^{T}BQ_{\bullet}+B^{2}\right)\\
m_{3}^{\left(\bullet\right)}=\frac{1}{m}\mathbb{E}\textrm{Tr}\left(A+Q_{\bullet}^{T}BQ_{\bullet}\right)^{3}=\frac{1}{m}\mathbb{E}\textrm{Tr}\left(A^{3}+3A^{2}Q_{\bullet}^{T}BQ_{\bullet}+3AQ_{\bullet}^{T}B^{2}Q_{\bullet}+B^{3}\right),
\end{array}
\end{equation}
where $Q_{\bullet}$ is $Q$ and $Q_{q}$ for the iso and the quantum
sums respectively and we used the familiar trace property $\textrm{Tr}(M_{1}M_{2})=\textrm{Tr}(M_{2}M_{1})$.
The equality of the first three moments of the iso and quantum with
the classical follows from The Departure Theorem.
\end{proof}
Furthermore, in the expansion of any of the moments $>4$ all the
non-departing terms are exactly captured by IE. These terms are equal
to the corresponding terms in the classical and the isotropic and
therefore equal to any linear combination of them. The departing terms
in higher moments (i.e.,$>4$) that are approximated by IE, we conjecture
are of little relevance. For example, the fifth moment has only two
terms (shown in bold) in its expansion that are departing: 

\begin{equation}
\begin{array}{c}
m_{5}=\frac{1}{m}\mathbb{E}\textrm{Tr}\left(A^{5}+5A^{4}Q_{\bullet}^{T}BQ_{\bullet}+\mathit{5}A^{3}Q_{\bullet}^{T}B^{2}Q_{\bullet}+\mathit{5}A^{2}Q_{\bullet}^{T}B^{3}Q_{\bullet}+\mathbf{\underline{5A\left(AQ_{\bullet}^{T}BQ_{\bullet}\right)^{2}}+}\right.\\
\left.\mathbf{\underline{5\left(AQ_{\bullet}^{T}BQ_{\bullet}\right)^{2}Q_{\bullet}^{T}BQ_{\bullet}}}+\mathit{5}AQ_{\bullet}^{T}B^{4}Q_{\bullet}+B^{5}\right)
\end{array}\label{eq:fifthmoments}
\end{equation}

\begin{table}
\begin{centering}
\begin{tabular}{|c|c|c|c|c|}
\hline 
\noalign{\vskip2sp}
number of & number of odds  & number of odds  & size of & Number of \tabularnewline
\noalign{\vskip2sp}
 sites & or evens ($N$ odd) & or evens ($N$ even) &  $H$ & copies\tabularnewline
\hline 
$N$ & $k=\frac{N-1}{2}$ & $k_{\mbox{odd}\vphantom{\mbox{even}}}=\frac{N}{2},\; k_{\mbox{even}\mbox{\ensuremath{\vphantom{odd}}}}=\frac{N-2}{2}$ & $m=d^{N}$ & $t$\tabularnewline
\hline 
\end{tabular}
\par\end{centering}

\vspace{0.3in}

\begin{centering}
\begin{tabular}{|c|c|c|c|c|c|c|c|}
\hline 
dimension of qudits & size of local terms & $l^{th}$ moment & $l^{th}$ cumulant & mean & variance & skewness & kurtosis\tabularnewline
\hline 
\hline 
$d$ & $n=d^{2}$ & $m_{l}$ & $\kappa_{l}$ & $\mu$ & $\sigma^{2}$ & $\gamma_{1}$ & $\gamma_{2}$\tabularnewline
\hline 
\end{tabular}
\par\end{centering}

\caption{\label{tab:parameters}Notation}
\end{table}

By the Departure Theorem the numerator in Eq. \ref{eq:convex} becomes,

\begin{equation}
\gamma_{2}^{q}-\gamma_{2}^{iso}=\frac{\kappa_{4}^{q}-\kappa_{4}^{iso}}{\sigma^{4}}=\frac{2}{m}\frac{\mathbb{E}\left\{ \textrm{Tr}\left[\left(AQ_{q}^{T}BQ_{q}\right)^{2}-\left(AQ^{T}BQ\right)^{2}\right]\right\} }{\sigma^{4}}\label{eq:numer}
\end{equation}
and the denominator in Eq. \ref{eq:convex} becomes, 

\begin{equation}
\gamma_{2}^{c}-\gamma_{2}^{iso}=\frac{\kappa_{4}^{c}-\kappa_{4}^{iso}}{\sigma^{4}}=\frac{2}{m}\frac{\mathbb{E}\left\{ \textrm{Tr}\left[\left(A\Pi^{T}B\Pi\right)^{2}-\left(AQ^{T}BQ\right)^{2}\right]\right\} }{\sigma^{4}},\label{eq:denom}
\end{equation}
where as before, $Q$ is a $\beta-$Haar measure orthogonal matrix
of size $m=d^{N}$, $Q_{q}=\left(Q_{q}^{(A)}\right)^{T}Q_{q}^{(B)}$
given by Eq. \ref{eq:OddEven} and $\kappa_{4}^{\bullet}$ denote
the fourth cumulants. Therefore, evaluation of $p$ reduces to the
evaluation of the right hand sides of Eqs. \ref{eq:numer} and \ref{eq:denom}.

Below we do not want to restrict ourselves to only chains with odd
number of sites and we need to take into account the multiplicity
of the eigenvalues as a result of taking the tensor product with identity.
It is convenient to denote the size of the matrices involved by $m=d^{N}=tn^{k}$,
where $n=d^{2}$ and $k=\frac{N-1}{2}$ and $t$ is the number of
copies (Section \ref{sub:Free} and Table \ref{tab:parameters}).

\subsection{Distribution of $A$ and $B$ }

The goal of this section is to express the moments of the entries
of $A$ and $B$ (e.g., $m_{2}^{A}$ and $m_{1,1}^{A}$) in terms
of the moments of the local terms (e.g for odd local terms $m_{2}^{\mbox{odd}},m_{11}^{\mbox{odd}}$).
Note that $A$ and $B$ are independent. The odd summands that make
up $A$ all commute and therefore can be locally diagonalized to give
the diagonal matrix $A$ (similarly for $B$),

\begin{eqnarray}
A & = & \sum_{i=1,3,\cdots}^{N-2}\mathbb{I}\otimes\Lambda_{i}\otimes\mathbb{I}\label{eq:A_and_B}\\
B & = & \sum_{i=2,4,\cdots}^{N-1}\mathbb{I}\otimes\Lambda_{i}\otimes\mathbb{I},\nonumber 
\end{eqnarray}
where $\Lambda_{i}$ are of size $d^{2}$ and are the diagonal matrices
of the local eigenvalues. 

The diagonal matrices $A$ and $B$ are formed by a direct sum of
the local eigenvalues of odds and evens respectively. For open boundary
conditions (OBC) each entry has a multiplicity given by Table \ref{tab:multiplicity}.

\begin{table}[H]
\begin{centering}
\begin{tabular}{|c|c|c|}
\hline 
OBC & $N$ odd & $N$ even\tabularnewline
\hline 
\hline 
$A$ & $d$ & $1$\tabularnewline
\hline 
$B$ & $d$ & $d^{2}$\tabularnewline
\hline 
\end{tabular}
\par\end{centering}

\begin{centering}
\caption{\label{tab:multiplicity}The multiplicity of terms in $A$ and $B$
for OBC. For closed boundary conditions there is no repetition.}

\par\end{centering}

\end{table}

Comment: We emphasize that $A$ and $B$ are independent of the eigenvector
structures. In particular, $A$ and $B$ are the same among the three
cases of isotropic, quantum and classical. 

We calculate the moments of $A$ and $B$. Let us treat the second
moment of $A$ ($B$ is done the same way). By the permutation invariance
of entries in $A$

\begin{eqnarray}
m_{2}^{A}\equiv\mathbb{E}\left(a_{1}^{2}\right) & = & \mathbb{E}\left(\lambda_{i_{1}}^{\left(1\right)}+\cdots+\lambda_{i_{k}}^{\left(k\right)}\right)^{2}\nonumber \\
 & = & \mathbb{E}\left[k\left(\lambda^{2}\right)+k\left(k-1\right)\lambda^{\left(1\right)}\lambda^{\left(2\right)}\right]\label{eq:m_2}\\
 & = & km_{2}^{\mbox{odd}}+k\left(k-1\right)m_{1,1}^{\mbox{odd}}\nonumber 
\end{eqnarray}
where expectation is taken over randomly chosen local eigenvalues,
$m_{2}^{\mbox{odd}}\equiv\mathbb{E}\left(\lambda_{i}^{2}\right)$
and $m_{1,1}^{\mbox{odd}}\equiv\mathbb{E}\left(\lambda_{i}\lambda_{j}\right)$
for some uniformly chosen $i$ and $j$ with $i\ne j$. The permutation
invariance assumption implies $\mathbb{E}\left(a_{i}^{2}\right)=\mathbb{E}\left(a_{1}^{2}\right)$
for all $i=1\cdots m$.

Comment: The key to this argument giving $m_{2}^{A}$ is that the
indices are not sensitive to the copying that results from the tensor
product with $\mathbb{I}_{d}$ at the boundaries.

Next we calculate the correlation between two diagonal terms, namely
$m_{1,1}^{A}\equiv\mathbb{E}\left(a_{i}a_{j}\right)$ for $i\neq j$.
We need to incorporate the multiplicity, denoted by $t$, due to the
tensor product with an identity matrix at the end of the chain,

\begin{eqnarray}
m_{1,1}^{A} & = & \frac{1}{m\left(m-1\right)}\mathbb{E}\left\{ \left(\sum_{i_{1},\cdots,i_{k}=1}^{n}\lambda_{i_{1}}^{\left(1\right)}+\cdots+\lambda_{i_{k}}^{\left(k\right)}\right)^{2}-\sum_{i_{1},\cdots,i_{k}=1}^{n}\left(\lambda_{i_{1}}^{\left(1\right)}+\cdots+\lambda_{i_{k}}^{\left(k\right)}\right)^{2}\right\} \label{eq:elemGeneral}\\
 & = & k\left(k-1\right)\mathbb{E}\left(\lambda\right)^{2}+k\left\{ \textrm{prob}\left(\lambda^{2}\right)\mathbb{E}\left(\lambda^{2}\right)+\textrm{prob}\left(\lambda_{1}\lambda_{2}\right)\mathbb{E}\left(\lambda_{1}\lambda_{2}\right)\right\} \nonumber \\
 & = & k\left(k-1\right)m_{2}^{\mbox{odd}}+\frac{k}{m-1}\left\{ \left(tn^{k-1}-1\right)m_{2}^{\mbox{odd}}+\left(tn^{k-1}\left(n-1\right)\right)m_{1,1}^{\mbox{odd}}\right\} \nonumber 
\end{eqnarray}
where, $\textrm{prob}\left(\lambda^{2}\right)=\frac{tn^{k-1}-1}{m-1}$
and $\textrm{prob}\left(\lambda_{1}\lambda_{2}\right)=\frac{tn^{k-1}\left(n-1\right)}{m-1}$.
Similarly for $B$.

\subsection{\label{sub:Isotropic-theory} Evaluation and Universality of $p=\frac{\gamma_{2}^{q}-\gamma_{2}^{iso}}{\gamma_{2}^{c}-\gamma_{2}^{iso}}$}

Recall the definition of $p$; from Eqs. \ref{eq:convex}, \ref{eq:numer}
and \ref{eq:denom} we have,

\begin{equation}
1-p=\frac{\mathbb{E}\mbox{Tr}\left(A\Pi^{T}B\Pi\right)^{2}-\mathbb{E}\textrm{Tr}\left(AQ_{q}^{T}BQ_{q}\right)^{2}}{\mathbb{E}\mbox{Tr}\left(A\Pi^{T}B\Pi\right)^{2}-\mathbb{E}\textrm{Tr}\left(AQ^{T}BQ\right)^{2}}.\label{eq:1-p}
\end{equation}

The classical case

\begin{equation}
\frac{1}{m}\mathbb{E}\mbox{Tr}\left(A\Pi^{T}B\Pi\right)^{2}=\frac{1}{m}\mathbb{E}\sum_{i=1}^{m}a_{i}^{2}b_{i}^{2}=\mathbb{E}\left(a_{i}^{2}\right)\mathbb{E}\left(b_{i}^{2}\right)=m_{2}^{A}m_{2}^{B}.\label{eq:classical_depart}
\end{equation}
\begin{table}
\begin{centering}
\begin{tabular}{|c|c|c|}
\hline 
moments & expectation values & count\tabularnewline
\hline 
\hline 
$m_{2}^{2}$ & $\mathbb{E}\left(\left|q_{i,j}\right|^{4}\right)=\frac{\beta+2}{m\left(m\beta+2\right)}$ & $m^{2}$\tabularnewline
\hline 
$m_{2}m_{11}$ & $\mathbb{E}\left(\left|q_{1,1}q_{1,2}\right|^{2}\right)=\frac{\beta}{m\left(m\beta+2\right)}$ & $2m^{2}\left(m-1\right)$\tabularnewline
\hline 
$\left(m_{11}\right)^{2}$ & $\mathbb{E}\left(q_{1,1}\overline{q_{1,2}}\overline{q_{2,1}}q_{2,2}\right)=-\frac{\beta}{m\left(m\beta+2\right)\left(m-1\right)}$ & $m^{2}\left(m-1\right)^{2}$\tabularnewline
\hline 
 & $\mathbb{E}\left(q_{13}^{2}q_{24}^{2}\right)=\frac{\beta\left(n-1\right)+2}{n\left(n\beta+2\right)\left(n-1\right)}$ & \tabularnewline
\hline 
\end{tabular}
\par\end{centering}

\centering{}\caption{\label{tab:HaarExp} The expectation values and counts of colliding
terms in $Q$ when it is $\beta-$Haar distributed. In this section
we use the first row; we include the last three rows for the calculations
in the appendix.}
\end{table}

Comment: Strictly speaking after the first equality we must have used
$b_{\pi_{i}}$ instead of $b_{i}$ but we simplified the notation
as they are the same in an expectation sense.

The general form for the denominator of Eq. \ref{eq:1-p} is

\begin{equation}
\frac{1}{m}\mathbb{E}\textrm{Tr}\left[\left(A\Pi^{T}B\Pi\right)^{2}-\left(AQ^{T}BQ\right)^{2}\right]=\frac{1}{m}\mathbb{E}\left\{ a_{l}^{2}b_{l}^{2}-a_{i}a_{k}b_{j}b_{p}\left(q_{ji}q_{jk}q_{pk}q_{pi}\right)\right\} .\label{eq:isotropic}
\end{equation}
It's worth noting that the arguments leading to Eq. \ref{eq:Iso-Classical_FINAL}
hold even if one fixes $A$ and $B$ and takes expectation values
over $\Pi$ and a permutation invariant $Q$ whose entries have the
same expectation value. The right hand side of Eq. \ref{eq:Iso-Classical_FINAL}
is a homogeneous polynomial of order two in the entries of $A$ and
$B$; consequently it necessarily has the form

\[
\frac{1}{m}\mathbb{E}\textrm{Tr}\left[\left(A\Pi^{T}B\Pi\right)^{2}-\left(AQ^{T}BQ\right)^{2}\right]=c_{1}\left(B,Q\right)m_{2}^{A}+c_{2}\left(B,Q\right)m_{1,1}^{A}
\]
but Eq. \ref{eq:isotropic} must be zero for $A=I$, for which $m_{2}^{A}=m_{1,1}^{A}=1$.
This implies that $c_{1}=-c_{2}$, allowing us to factor out $\left(m_{2}^{A}-m_{1,1}^{A}\right)$.
Similarly, the homogeneity and permutation invariance of $B$ implies,

\[
\frac{1}{m}\mathbb{E}\textrm{Tr}\left[\left(A\Pi^{T}B\Pi\right)^{2}-\left(AQ^{T}BQ\right)^{2}\right]=\left(m_{2}^{A}-m_{1,1}^{A}\right)\left(D_{1}\left(Q\right)m_{2}^{B}+D_{2}\left(Q\right)m_{1,1}^{B}\right).
\]
The right hand side should be zero for $B=I$, whereby we can factor
out $\left(m_{2}^{B}-m_{1,1}^{B}\right)$

\begin{equation}
\frac{1}{m}\mathbb{E}\textrm{Tr}\left[\left(A\Pi^{T}B\Pi\right)^{2}-\left(AQ^{T}BQ\right)^{2}\right]=\left(m_{2}^{A}-m_{1,1}^{A}\right)\left(m_{2}^{B}-m_{1,1}^{B}\right)f\left(Q\right),\label{eq:FreeExpect}
\end{equation}
where $m_{2}^{A}=\mathbb{E}\left(a_{i}^{2}\right)$, $ $ $m_{2}^{B}=\mathbb{E}\left(b_{j}^{2}\right)$,
and $m_{1,1}^{A}=\mathbb{E}\left(a_{i},a_{j}\right)$ , $m_{1,1}^{B}=\mathbb{E}\left(b_{i},b_{j}\right)$.
Moreover $f\left(Q\right)$ is a homogeneous function of order four
in the entries of $Q$. To evaluate $f\left(Q\right)$, it suffices
to let $A$ and $B$ be projectors of rank one where $A$ would have
only one nonzero entry on the $i^{\mbox{th }}$ position on its diagonal
and $B$ only one nonzero entry on the $j^{\mbox{th }}$ position
on its diagonal. Further take those nonzero entries to be ones, giving
$m_{1,1}^{A}=m_{1,1}^{B}=0$ and $m_{2}^{A}=m_{2}^{B}=1/m$,

\begin{equation}
\frac{1}{m}\mathbb{E}\textrm{Tr}\left[\left(A\Pi^{T}B\Pi\right)^{2}-\left(AQ^{T}BQ\right)^{2}\right]=\frac{1}{m^{2}}f\left(Q\right)
\end{equation}
But the left hand side is

\begin{eqnarray*}
\frac{1}{m}\mathbb{E}\textrm{Tr}\left[\left(A\Pi^{T}B\Pi\right)^{2}-\left(AQ^{T}BQ\right)^{2}\right] & = & \frac{1}{m}\mathbb{E}\left[\delta_{ij}-q_{ij}^{4}\right]\\
 & = & \frac{1}{m}\left\{ \frac{1}{m^{2}}\sum_{ij}\delta_{ij}-\frac{1}{m^{2}}\sum_{ij}\mathbb{E}\left(q_{ij}^{4}\right)\right\} \\
 & = & \frac{1}{m}\left\{ \frac{1}{m}-\mathbb{E}\left(q_{ij}^{4}\right)\right\} ,
\end{eqnarray*}
where, we used the homogeneity of $Q$. Consequently, by equating
this to $f\left(Q\right)/m^{2}$, we get the desired quantity

\[
f\left(Q\right)=\left\{ 1-m\mathbb{E}\left(q_{ij}^{4}\right)\right\} 
\]
Our final result Eq. \ref{eq:FreeExpect} now reads

\begin{equation}
\frac{1}{m}\mathbb{E}\textrm{Tr}\left[\left(A\Pi^{T}B\Pi\right)^{2}-\left(AQ^{T}BQ\right)^{2}\right]=\left(m_{2}^{A}-m_{1,1}^{A}\right)\left(m_{2}^{B}-m_{1,1}^{B}\right)\left\{ 1-m\mathbb{E}\left(q_{ij}^{4}\right)\right\} .\label{eq:isotropicBook}
\end{equation}
The same calculation where each of the terms is obtained separately
yields the same result (Appendix). In this paper $p$ is formed by
taking $Q$ to have a $\beta-$Haar measure. Expectation values of
the entries of $Q$ are listed in the Table \ref{tab:HaarExp}.

We wish to express everything in terms of the local terms; using Eqs.
\ref{eq:m_2} and \ref{eq:elemGeneral} as well as $tn^{k}=m$,

\begin{eqnarray*}
m_{2}^{A}-m_{1,1}^{A} & = & \frac{tk\left(n-1\right)n^{k-1}}{m-1}\left(m_{2}^{\mbox{odd}}-m_{1,1}^{\mbox{odd}}\right)\\
m_{2}^{B}-m_{1,1}^{B} & = & \frac{tk\left(n-1\right)n^{k-1}}{m-1}\left(m_{2}^{\mbox{even}}-m_{1,1}^{\mbox{even}}\right),
\end{eqnarray*}
giving 

\begin{eqnarray}
\frac{1}{m}\mathbb{E}\left[\mbox{Tr}\left(A\Pi^{T}B\Pi\right)^{2}-\textrm{Tr}\left(AQ^{T}BQ\right)^{2}\right] & = & \left(m_{2}^{\mbox{odd}}-m_{1,1}^{\mbox{odd}}\right)\left(m_{2}^{\mbox{even}}-m_{1,1}^{\mbox{even}}\right)\times\nonumber \\
 &  & \left(\frac{km\left(n-1\right)}{n\left(m-1\right)}\right)^{2}\left\{ 1-m\mathbb{E}\left(q_{ij}^{4}\right)\right\} .\label{eq:Iso-Classical_FINAL}
\end{eqnarray}

\vspace{0.2in}

We now proceed to the quantum case where we need to evaluate $\frac{1}{m}\mathbb{E}\left[\left(A\Pi^{T}B\Pi\right)^{2}-\textrm{Tr}\left(AQ_{q}^{T}BQ_{q}\right)^{2}\right]$.
In this case, we cannot directly use the techniques that we used to
get Eq. \ref{eq:Iso-Classical_FINAL} because $Q_{q}$ is not permutation
invariant despite local eigenvectors being so. Before proceeding further
we like to prove a useful lemma (Lemma \ref{lem:diamonds}). Let us
simplify the notation and denote the local terms that are drawn randomly
from a known distribution by $H_{l,l+1}\equiv H^{\left(l\right)}$
whose eigenvalues are $\Lambda_{l}$ as discussed above. 

\begin{flushleft}
Recall that $A$ represents the \textit{sum} of all the odds and $Q_{q}^{-1}BQ_{q}$
the \textit{sum} of all the evens,
\par\end{flushleft}

\[
H_{\mbox{odd}\vphantom{\mbox{even}}}=\sum_{l=1,3,5,\cdots}\mathbb{I}\otimes H^{\left(l\right)}\otimes\mathbb{I},\mbox{ and}\quad H_{\mbox{even}\vphantom{\mbox{odd}}}=\sum_{l=2,4,6,\cdots}\mathbb{I}\otimes H^{\left(l\right)}\otimes\mathbb{I},
\]

\begin{flushleft}
Hence, the expansion of $\frac{1}{m}\mathbb{E}\left[\textrm{Tr}\left(AQ_{q}^{T}BQ_{q}\right)^{2}\right]$
amounts to picking an odd term, an even term, then another odd term
and another even term, multiplying them together and taking the expectation
value of the trace of the product (Figure \ref{fig:racksMoment}).
Therefore, each term in the expansion can have four, three or two
different local terms, whose expectation values along with their counts
are needed. These expectations are taken with respect to the local
terms (dense $d^{2}\times d^{2}$ random matrices).
\par\end{flushleft}

\begin{figure}
\begin{centering}
\includegraphics[scale=1.5]{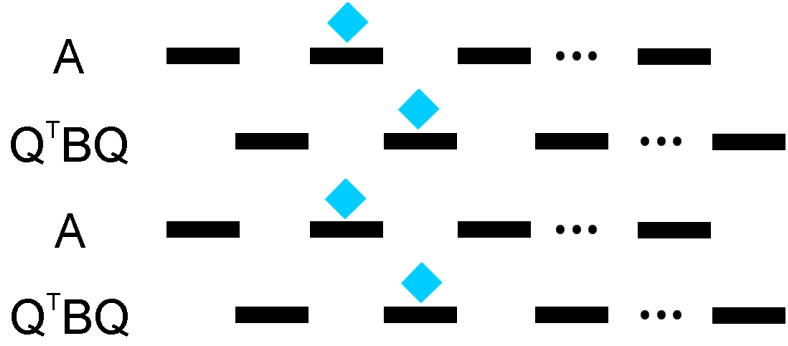}
\par\end{centering}

\caption{\label{fig:racksMoment}The terms in the expansion of $\frac{1}{m}\mathbb{E}\left[\textrm{Tr}\left(AQ_{q}^{T}BQ_{q}\right)^{2}\right]$
can be visualized as picking an element from each row from top to
bottom and multiplying. Each row has $k$ of the local terms corresponding
to a chain with odd number of terms. Among $k^{4}$ terms roughly
$k^{2}$ of them differ among the classical, isotropic and quantum
cases (See Eqs. \ref{eq:QExpCountOdd} and \ref{eq:QExpCountEven}).
An example of such a choice is shown by diamonds. }
\end{figure}

The expectation values depend on the type of random matrix distribution
from which the local terms are drawn. The counting however, depends
on the configuration of the lattice only. We show the counting of
the number of terms, taking care of the boundary terms for an open
chain, along with the type of expectation values by which they need
to be weighted:

\uline{For $N$ odd ($k$ odd terms and $k$ even terms)}

\begin{equation}
\begin{array}{c}
\textrm{Four}\: H^{\left(\centerdot\right)}\textrm{'s}:k^{2}\left(k-1\right)^{2}\Rightarrow d^{N-u_{1}}\mathbb{E}\textrm{Tr}\left(H^{\left(l\right)}\right)^{4},\mbox{ }u_{1}\in\left\{ 5,\cdots,8\right\} \\
\textrm{Three}\: H^{\left(\centerdot\right)}\textrm{'s}:2k^{2}\left(k-1\right)\Rightarrow d^{N-u_{2}}\mathbb{E}\textrm{Tr}\left(\left[H^{\left(l\right)}\right]^{2}\right)\mathbb{E}\textrm{Tr}\left(H^{\left(l\right)}\right)^{2},\mbox{ }u_{2}\in\left\{ 4,5,6\right\} \\
\textrm{Two}\: H^{\left(\centerdot\right)}\textrm{'s:}\left(k-1\right)^{2}\;\textrm{Not\;\ Entangled}\Rightarrow d^{N-4}\left\{ \mathbb{E}\textrm{Tr}\left(\left[H^{\left(l\right)}\right]^{2}\right)\right\} ^{2}\\
\textrm{Two}\: H^{\left(\centerdot\right)}\textrm{'s}:\left(2k-1\right)\;\textrm{ Entangled}\Rightarrow d^{N-3}\mathbb{E}\textrm{Tr}\left[\left(H^{\left(l\right)}\otimes\mathbb{I}\right)\left(\mathbb{I}\otimes H^{\left(l+1\right)}\right)\left(H^{\left(l\right)}\otimes\mathbb{I}\right)\left(\mathbb{I}\otimes H^{\left(l+1\right)}\right)\right]
\end{array}\label{eq:QExpCountOdd}
\end{equation}

\uline{For $N$ even ($k$ odd terms and $k-1$ even terms)}

\begin{equation}
\begin{array}{c}
\textrm{Four}\: H^{\left(\centerdot\right)}\textrm{'s}:k\left(k-1\right)^{2}\left(k-2\right)\Rightarrow d^{N-u_{1}}\mathbb{E}\textrm{Tr}\left(H^{\left(l\right)}\right)^{4},\mbox{ }u_{1}\in\left\{ 5,\cdots,8\right\} \\
\textrm{Three}\: H^{\left(\centerdot\right)}\textrm{'s}:k\left(k-1\right)\left(2k-3\right)\Rightarrow d^{N-u_{2}}\mathbb{E}\textrm{Tr}\left(\left[H^{\left(l\right)}\right]^{2}\right)\mathbb{E}\textrm{Tr}\left(H^{\left(l\right)}\right)^{2},\mbox{ }u_{2}\in\left\{ 4,5,6\right\} \\
\textrm{Two}\: H^{\left(\centerdot\right)}\textrm{'s}:\left(k-1\right)\left(k-2\right)\;\textrm{Not\;\ Entangled}\Rightarrow d^{N-4}\left\{ \mathbb{E}\textrm{Tr}\left(\left[H^{\left(l\right)}\right]^{2}\right)\right\} ^{2}\\
\textrm{Two}\: H^{\left(\centerdot\right)}\textrm{'s}:2\left(k-1\right)\;\textrm{ Entangled}\Rightarrow d^{N-3}\mathbb{E}\textrm{Tr}\left[\left(H^{\left(l\right)}\otimes\mathbb{I}\right)\left(\mathbb{I}\otimes H^{\left(l+1\right)}\right)\left(H^{\left(l\right)}\otimes\mathbb{I}\right)\left(\mathbb{I}\otimes H^{\left(l+1\right)}\right)\right]
\end{array}\label{eq:QExpCountEven}
\end{equation}
Here $u_{1}$ and $u_{2}$ indicate the number of sites that the local
terms act on (i.e., occupy). Therefore, $\frac{1}{m}\mathbb{E}\left[\textrm{Tr}\left(AQ_{q}^{T}BQ_{q}\right)^{2}\right]$
is obtained by multiplying each type of terms, weighted by the counts
and summing. For example for $u_{1}=5$ and $u_{2}=3$, when $N$
is odd,

\begin{equation}
\begin{array}{c}
\frac{1}{m}\mathbb{E}\left[\textrm{Tr}\left(AQ_{q}^{T}BQ_{q}\right)^{2}\right]=\frac{1}{m}\left\{ d^{N-5}k^{2}\left(k-1\right)^{2}\mathbb{E}\textrm{Tr}\left(H^{\left(l\right)}\right)^{4}+\right.\\
2k^{2}\left(k-1\right)d^{N-4}\mathbb{E}\textrm{Tr}\left(\left[H^{\left(l\right)}\right]^{2}\right)\mathbb{E}\textrm{Tr}\left(H^{\left(l\right)}\right)^{2}+\left(k-1\right)^{2}d^{N-4}\left\{ \mathbb{E}\textrm{Tr}\left(\left[H^{\left(l\right)}\right]^{2}\right)\right\} ^{2}+\\
\left.\left(2k-1\right)d^{N-3}\mathbb{E}\textrm{Tr}\left[\left(H^{\left(l\right)}\otimes\mathbb{I}\right)\left(\mathbb{I}\otimes H^{\left(l+1\right)}\right)\left(H^{\left(l\right)}\otimes\mathbb{I}\right)\left(\mathbb{I}\otimes H^{\left(l+1\right)}\right)\right]\right\} 
\end{array}\label{eq:quantumOdd}
\end{equation}
and similarly for $N$ even,

\begin{equation}
\begin{array}{c}
\frac{1}{m}\mathbb{E}\left[\textrm{Tr}\left(AQ_{q}^{T}BQ_{q}\right)^{2}\right]=\frac{\left(k-1\right)}{m}\left\{ k\left(k-1\right)\left(k-2\right)d^{N-5}\mathbb{E}\textrm{Tr}\left(H^{\left(l\right)}\right)^{4}+\right.\\
k\left(2k-3\right)d^{N-4}\mathbb{E}\textrm{Tr}\left(\left[H^{\left(l\right)}\right]^{2}\right)\mathbb{E}\textrm{Tr}\left(H^{\left(l\right)}\right)^{2}+\left(k-2\right)d^{N-4}\left\{ \mathbb{E}\textrm{Tr}\left(\left[H^{\left(l\right)}\right]^{2}\right)\right\} ^{2}+\\
\left.2d^{N-3}\mathbb{E}\textrm{Tr}\left[\left(H^{\left(l\right)}\otimes\mathbb{I}\right)\left(\mathbb{I}\otimes H^{\left(l+1\right)}\right)\left(H^{\left(l\right)}\otimes\mathbb{I}\right)\left(\mathbb{I}\otimes H^{\left(l+1\right)}\right)\right]\right\} .
\end{array}\label{eq:quantumEven}
\end{equation}

The expectation values depend on the type of random matrix distribution
from which the local terms are drawn. We will give explicit examples
in the following sections. In the following lemma, we use $\mathbb{E\left(\mathit{H^{\left(l\right)}}\right)=}\mu\mathbb{I}_{d^{2}}$
and $\mathbb{E\left(\mathit{H^{\left(l\right)}}\right)^{\mathit{2}}=}m_{2}\mathbb{I}_{d^{2}}$.
\begin{lem}
\label{lem:diamonds}In calculating the $\mathbb{E}\mbox{Tr}\left(AQ_{q}^{T}BQ_{q}\right)^{2}$
if at least one of the odds (evens) commutes with one of the evens
(odds) then the expectation value is the same as the classical expectation
value. Further if the local terms have permutation invariance of eigenvalues
then the only quantum expectation value that differs from classical
is of Type II (see the proof and the diamonds in figure \ref{fig:racksMoment}).\end{lem}
\begin{proof}
This can be shown using the trace property $\mbox{Tr}\left(MP\right)=\mbox{Tr}\left(PM\right)$.
In calculating $\mathbb{E}\mbox{Tr}\left(H_{l}^{odd}H_{p}^{even}H_{j}^{odd}H_{k}^{even}\right)$;
if any of the odd (even) terms commutes with any of the even (odd)
terms to its left or right then they can be swapped. For example one
gets $\mathbb{E}\mbox{Tr}\left(H_{l}^{odd}H_{p}^{even}H_{k}^{even}H_{j}^{odd}\right)=\mathbb{E}\mbox{Tr}\left(H_{j}^{odd}H_{l}^{odd}H_{p}^{even}H_{k}^{even}\right)$
which is just the classical value. Hence the only types of expectations
that we need to worry about are

\[
\begin{array}{ccc}
 & \underset{\_\_}{H^{\left(l\right)}}\\
 &  & \underset{\_\_}{H^{\left(l+1\right)}}\\
 & \underset{\_\_}{H^{\left(l\right)}}\\
\underset{\_\_}{H^{\left(l-1\right)}}\\
 & \mbox{Type I}
\end{array}\qquad\mbox{ and }\qquad\begin{array}{cc}
\underset{\_\_}{H^{\left(l\right)}}\\
 & \underset{\_\_}{H^{\left(l+1\right)}}\\
\underset{\_\_}{H^{\left(l\right)}}\\
 & \underset{\_\_}{H^{\left(l+1\right)}}\\
\mbox{Type II}
\end{array}
\]
now we show that with permutation invariance of the local eigenvalues
the first type are also classical leaving us with the ``diamond terms''
alone (Fig. \ref{fig:racksMoment}). Consider a Type I term, which
involves three independent local terms,

\[
\begin{array}{c}
{\scriptstyle \frac{1}{m}\mathbb{E}\textrm{Tr}\left[\left(\mathbb{I}_{d^{2}}\otimes H^{\left(3\right)}\otimes\mathbb{I}_{d^{N-4}}\right)\left(\mathbb{I}\otimes H^{\left(2\right)}\otimes\mathbb{I}_{d^{N-3}}\right)\left(\mathbb{I}_{d^{2}}\otimes H^{\left(3\right)}\otimes\mathbb{I}_{d^{N-4}}\right)\left(\mathbb{I}_{d^{3}}\otimes H^{\left(4\right)}\otimes\mathbb{I}_{d^{N-5}}\right)\right]}\\
=\mu^{2}m_{2}.
\end{array}
\]
This follows immediately from the independence of $H^{\left(4\right)}$
, which allows us to take its expectation value separately giving
a $\mu$ and leaving us with 
\[
\frac{\mu}{m}\mathbb{E}\textrm{Tr}\left[\left(\mathbb{I}_{d^{2}}\otimes H^{\left(3\right)}\otimes\mathbb{I}_{d^{N-4}}\right)^{2}\left(\mathbb{I}\otimes H^{\left(2\right)}\otimes\mathbb{I}_{d^{N-3}}\right)\right]=\mu^{2}m_{2}.
\]

Therefore the only relevant terms, shown by diamonds in Fig. \ref{fig:racksMoment},
are of Type II. As an example of such terms consider (here on repeated
indices are summed over)

\begin{equation}
\begin{array}{c}
{\scriptstyle \frac{1}{m}\mathbb{E}\textrm{Tr}\left[\left(H^{\left(1\right)}\otimes\mathbb{I}_{d^{N-2}}\right)\left(\mathbb{I}\otimes H^{\left(2\right)}\otimes\mathbb{I}_{d^{N-3}}\right)\left(H^{\left(1\right)}\otimes\mathbb{I}_{d^{N-2}}\right)\left(\mathbb{I}\otimes H^{\left(2\right)}\otimes\mathbb{I}_{d^{N-3}}\right)\right]}\\
=\frac{1}{d^{3}}\left\{ \mathbb{E}\left(H_{i_{1}i_{2},j_{1}j_{2}}^{\left(1\right)}H_{i_{1}p_{2},j_{1}k_{2}}^{\left(1\right)}\right)\mathbb{E}\left(H_{j_{2}i_{3},k_{2}k_{3}}^{\left(2\right)}H_{i_{2}i_{3},p_{2}k_{3}}^{\left(2\right)}\right)\right\} ,
\end{array}
\end{equation}
where the indices with subscript $2$ prevent us from treating the
two expectation values independently: $H^{\left(1\right)}$ and $H^{\left(2\right)}$
overlap at the second site. The number of such terms is $2k-1$, where
$k=\frac{N-1}{2}$.
\end{proof}
Therefore, we have found a further reduction of the terms from the
departure theorem, that distinguishes the quantum problem from the
other two. Luckily and interestingly the kurtosis of the quantum case
lies in between the classical and the iso. We emphasize that \textit{the
only inputs to the theory are the geometry of the lattice (e.g., the
number of summands and the inter-connectivity of the local terms)
and the moments} that characterizes the type of the local terms. 

Comment: The most general treatment would consider Type I terms as
well, i.e., there is no assumption of permutation invariance of the
eigenvalues of the local terms. This allows one to treat all types
of local terms. Here we are confining to random local interactions,
where the local eigenvectors are generic or the eigenvalues locally
are permutation invariant in the expectation value sense. 

The goal is to find $p$ by matching fourth moments

\[
1-p=\frac{\mathbb{E}\mbox{Tr}\left(A\Pi^{T}B\Pi\right)^{2}-\mathbb{E}\mbox{Tr}\left(AQ_{q}^{T}BQ_{q}\right)^{2}}{\mathbb{E}\mbox{Tr}\left(A\Pi^{T}B\Pi\right)^{2}-\mathbb{E}\mbox{Tr}\left(AQ^{T}BQ\right)^{2}}
\]
for which we calculated the denominator resulting in Eq. \ref{eq:Iso-Classical_FINAL},
where $\mathbb{E}\left(\left|q_{i,j}\right|^{4}\right)=\frac{\beta+2}{m\left(m\beta+2\right)}$
for $\beta-$Haar $Q$$ $ (Table \ref{tab:HaarExp}). If the numerator
allows a factorization of the moments of the local terms as in Eq.
\ref{eq:Iso-Classical_FINAL}, then the value of $p$ will be independent
of the covariance matrix (i.e., eigenvalues of the local terms). 
\begin{lem*}
\textbf{\textup{(Universality)}} $p\mapsto p\left(N,d,\beta\right)$,
namely, it is independent of the distribution of the local terms.\end{lem*}
\begin{proof}
We use a similar techniques as we did in the isotropic case. The general
form for the numerator of Eq. \ref{eq:1-p} is (denoting Lemma \ref{lem:diamonds}
by L3)

\begin{eqnarray}
\frac{1}{m}\mathbb{E}\textrm{Tr}\left[\left(A\Pi^{T}B\Pi\right)^{2}-\left(AQ_{q}^{T}BQ_{q}\right)^{2}\right] & \overset{\mbox{L3}}{=} & \frac{\left(2k-1\right)}{d^{3}}\mbox{\ensuremath{\mathbb{E}}Tr}\left\{ \left(H^{\left(l\right)}\otimes\mathbb{I}_{d}\right)^{2}\left(\mathbb{I}_{d}\otimes H^{\left(l+1\right)}\right)^{2}\right.\nonumber \\
 &  & -\left.\left[\left(H^{\left(l\right)}\otimes\mathbb{I}_{d}\right)\left(\mathbb{I}_{d}\otimes H^{\left(l+1\right)}\right)\right]^{2}\right\} \nonumber \\
 & = & \frac{\left(2k-1\right)}{d^{3}}\mbox{\ensuremath{\mathbb{E}}}\mbox{Tr}\left\{ \left(Q_{l}^{-1}\Lambda_{l}Q_{l}\otimes\mathbb{I}_{d}\right)^{2}\left(\mathbb{I}_{d}\otimes Q_{l+1}^{-1}\Lambda_{l+1}Q_{l+1}\right)^{2}\right.\nonumber \\
 &  & -\left.\left[\left(Q_{l}^{-1}\Lambda_{l}Q_{l}\otimes\mathbb{I}_{d}\right)\left(\mathbb{I}_{d}\otimes Q_{l+1}^{-1}\Lambda_{l+1}Q_{l+1}\right)\right]^{2}\right\} \label{eq:quantumBook-1}
\end{eqnarray}
where the expectation on the right hand side is taken with respect
to the local terms $H^{\left(l\right)}$ and $H^{\left(l+1\right)}$
. The right hand side is a homogeneous polynomial of order two in
the entries of $\Lambda_{l}$, as well as, in the entries of $\Lambda_{l+1}$;
consequently Eq. \ref{eq:quantumBook-1} necessarily has the form

\[
c_{1}\left(\Lambda^{\mbox{even}},Q_{\mbox{odd}},Q_{\mbox{even}}\right)m_{2}^{\mbox{odd}}+c_{2}\left(H^{\mbox{even}},Q_{\mbox{odd}},Q_{\mbox{even}}\right)m_{1,1}^{\mbox{odd}}
\]
but Eq. \ref{eq:quantumBook-1} must be zero for $\Lambda_{l}=I$,
for which $m_{2}^{\mbox{odd}}=m_{1,1}^{\mbox{odd}}=1$. This implies
that $c_{1}=-c_{2}$. By permutation invariance of the local terms
we can factor out $\left(m_{2}^{\mbox{odd}}-m_{1,1}^{\mbox{odd}}\right)$.
Similarly, the homogeneity and permutation invariance of $H^{\left(l+1\right)}$
implies,

\[
\left(m_{2}^{\mbox{odd}}-m_{1,1}^{\mbox{odd}}\right)\left[D_{1}\left(Q_{\mbox{odd}},Q_{\mbox{even}}\right)m_{2}^{\mbox{even}}+D_{2}\left(Q_{\mbox{odd}},Q_{\mbox{even}}\right)m_{1,1}^{\mbox{even}}\right].
\]
The right hand side should be zero for $\Lambda_{l+1}=I$, whereby
we can factor out $\left(m_{2}^{\mbox{even}}-m_{1,1}^{\mbox{even}}\right)$;
hence the right hand side of Eq. \ref{eq:quantumBook-1} becomes

\begin{equation}
\frac{\left(2k-1\right)}{d^{3}}\left(m_{2}^{\mbox{odd}}-m_{1,1}^{\mbox{odd}}\right)\left(m_{2}^{\mbox{even}}-m_{1,1}^{\mbox{even}}\right)f_{q}\left(Q_{\mbox{odd}},Q_{\mbox{even}}\right)\label{eq:Q-C_final}
\end{equation}
where $f_{q}\left(Q_{\mbox{odd}},Q_{\mbox{even}}\right)$ is a homogeneous
function of order four in the entries of $Q_{\mbox{odd}}$ as well
as $Q_{\mbox{even}}$. To evaluate $f_{q}$, it suffices to let $\Lambda_{l}$
and $\Lambda_{l+1}$ be projectors of rank one where $\Lambda_{l}$
would have only one nonzero entry on the $i^{\mbox{th }}$ position
on its diagonal and $\Lambda_{l+1}$ only one nonzero entry on the
$j^{\mbox{th }}$ position on its diagonal. Further take those nonzero
entries to be ones, giving $m_{1,1}^{A}=m_{1,1}^{B}=0$ and $m_{2}^{A}=m_{2}^{B}=1/n$.
Using this choice of local terms the right hand side of Eq. \ref{eq:quantumBook-1}
now reads 

\begin{eqnarray}
\frac{\left(2k-1\right)}{d^{3}} & \mbox{\ensuremath{\mathbb{E}}}\mbox{Tr} & \left\{ \left(|q_{i}^{\left(l\right)}\rangle\langle q_{i}^{\left(l\right)}|\otimes I_{d}\right)^{2}\left(I_{d}\otimes|q_{j}^{\left(l+1\right)}\rangle\langle q_{j}^{\left(l+1\right)}|\right)^{2}\right.\nonumber \\
 &  & -\left.\left[\left(|q_{i}^{\left(l\right)}\rangle\langle q_{i}^{\left(l\right)}|\otimes I_{d}\right)\left(I_{d}\otimes|q_{j}^{\left(l+1\right)}\rangle\langle q_{j}^{\left(l+1\right)}|\right)\right]^{2}\right\} \label{eq:Qdepart}
\end{eqnarray}
where here the expectation value is taken with respect to random choices
of local eigenvectors. Equating this and Eq. \ref{eq:Q-C_final} 

\begin{eqnarray}
f_{q}\left(Q_{\mbox{odd}},Q_{\mbox{even}}\right) & \mbox{=} & n^{2}\mbox{\ensuremath{\mathbb{E}}}\mbox{Tr}\left\{ \left(|q_{i}^{\left(l\right)}\rangle\langle q_{i}^{\left(l\right)}|\otimes I_{d}\right)^{2}\left(I_{d}\otimes|q_{j}^{\left(l+1\right)}\rangle\langle q_{j}^{\left(l+1\right)}|\right)^{2}\right.\label{eq:Qdepart-1}\\
 &  & -\left.\left[\left(|q_{i}^{\left(l\right)}\rangle\langle q_{i}^{\left(l\right)}|\otimes I_{d}\right)\left(I_{d}\otimes|q_{j}^{\left(l+1\right)}\rangle\langle q_{j}^{\left(l+1\right)}|\right)\right]^{2}\right\} \nonumber 
\end{eqnarray}

To simplify notation let us expand these vectors in the computational
basis $|q_{i}^{\left(l\right)}\rangle=u_{i_{1}i_{2}}|i_{1}\rangle|i_{2}\rangle$
and $|q_{j}^{\left(l+1\right)}\rangle=v_{i_{2}i_{3}}|i_{2}\rangle|i_{3}\rangle.$
The first term on the right hand side of Eq. \ref{eq:Qdepart}, the
classical term, is obtained by assuming commutativity and using the
projector properties,

\begin{eqnarray}
\mbox{Tr}\left[\left(|q_{i}^{\left(l\right)}\rangle\langle q_{i}^{\left(l\right)}|\otimes I_{d}\right)^{2}\left(I_{d}\otimes|q_{j}^{\left(l+1\right)}\rangle\langle q_{j}^{\left(l+1\right)}|\right)^{2}\right] & =\nonumber \\
\mbox{Tr}\left[\left(|q_{i}^{\left(l\right)}\rangle\langle q_{i}^{\left(l\right)}|\otimes I_{d}\right)\left(I_{d}\otimes|q_{j}^{\left(l+1\right)}\rangle\langle q_{j}^{\left(l+1\right)}|\right)\right] & =\nonumber \\
\mbox{Tr}\left[u_{i_{1},i_{2}}\overline{u_{j_{1}j_{2}}}v_{j_{2}i_{3}}\overline{v_{k_{2}k_{3}}}u_{j_{1}k_{2}}|i_{1}i_{2}i_{3}\rangle\langle j_{1}k_{2}k_{3}|\right] & =\nonumber \\
\left[u_{i_{1},i_{2}}\overline{u_{i_{1}j_{2}}}v_{j_{2}i_{3}}\overline{v_{i_{2}i_{3}}}\right]=\left(u^{\dagger}u\right)_{j_{2}i_{2}}\left(vv^{\dagger}\right)_{j_{2}i_{2}} & =\nonumber \\
\mbox{Tr}\left[\left(u^{\dagger}u\right)\left(vv^{\dagger}\right)\right]=\mbox{Tr}\left[uv\left(uv\right)^{\dagger}\right] & =\nonumber \\
\left\Vert uv\right\Vert _{\mbox{F}}^{2} & = & \sum_{i=1}^{d}\sigma_{i}^{2}.\label{eq:CDepart}
\end{eqnarray}
where $\left\Vert \centerdot\right\Vert _{\mbox{F}}$ denotes the
Frobenius norm and $\sigma_{i}$ are the singular values of $uv$.
The second term, the quantum term, is

\begin{eqnarray}
\mbox{Tr}\left[\left(|q_{i}^{\left(l\right)}\rangle\langle q_{i}^{\left(l\right)}|\otimes I_{d}\right)\left(I_{d}\otimes|q_{j}^{\left(l+1\right)}\rangle\langle q_{j}^{\left(l+1\right)}|\right)\right]^{2} & =\label{eq:QdepartFinal}\\
\mbox{Tr}\left[u_{i_{1}i_{2}}\overline{u_{j_{1}j_{2}}}v_{j_{2}i_{3}}\overline{v_{k_{2}k_{3}}}u_{j_{1}k_{2}}\overline{u_{m_{1}m_{2}}}v_{m_{2}k_{3}}\overline{v_{i_{2}i_{3}}}|i_{1}i_{2}i_{3}\rangle\langle p_{1}p_{2}p_{3}|\right] & =\nonumber \\
\left(u^{\dagger}u\right)_{j_{2}k_{2}}\left(vv^{\dagger}\right)_{m_{2}k_{2}}\left(u^{\dagger}u\right)_{m_{2}i_{2}}\left(vv^{\dagger}\right)_{j_{2}i_{2}} & =\nonumber \\
\left(u^{\dagger}uvv^{\dagger}\right)_{j_{2}m_{2}}\left(u^{\dagger}uvv^{\dagger}\right)_{m_{2}j_{2}}=\mbox{Tr}\left\{ \left[uv\left(uv\right)^{\dagger}\right]^{2}\right\}  & =\nonumber \\
\left\Vert uv\left(uv\right)^{\dagger}\right\Vert _{\mbox{F}}^{2} & = & \sum_{i=1}^{d}\sigma_{i}^{4}.\nonumber 
\end{eqnarray}
where we used the symmetry of $\left(uv\left(uv\right)^{\dagger}\right)^{2}=uv\left(uv\right)^{\dagger}\left[uv\left(uv\right)^{\dagger}\right]^{\dagger}$.

Now we can calculate 

\begin{equation}
f_{q}\left(Q_{\mbox{odd}},Q_{\mbox{even}}\right)=n^{2}\mbox{\ensuremath{\mathbb{E}}}\left\{ \left\Vert uv\right\Vert _{\mbox{F}}^{2}-\left\Vert uv\left(uv\right)^{\dagger}\right\Vert _{\mbox{F}}^{2}\right\} \label{eq:d4_introduced}
\end{equation}
giving us the desired result

\begin{eqnarray}
\frac{1}{m}\mathbb{E}\textrm{Tr}\left[\left(A\Pi^{T}B\Pi\right)^{2}-\left(AQ_{q}^{T}BQ_{q}\right)^{2}\right] & = & d\left(2k-1\right)\left(m_{2}^{\mbox{odd}}-m_{1,1}^{\mbox{odd}}\right)\left(m_{2}^{\mbox{even}}-m_{1,1}^{\mbox{even}}\right)\label{eq:Quantum-Classical}\\
 & \times & \mathbb{E}\left(\left\Vert uv\right\Vert _{\mbox{F}}^{2}-\left\Vert uv\left(uv\right)^{\dagger}\right\Vert _{\mbox{F}}^{2}\right),\nonumber 
\end{eqnarray}
from which 

\begin{eqnarray}
1-p & = & \frac{\mbox{ETr}\left(A\Pi^{T}B\Pi\right)^{2}-\mbox{ETr}\left(AQ_{q}^{-1}BQ_{q}\right)^{2}}{\mbox{ETr}\left(A\Pi^{T}B\Pi\right)^{2}-\mbox{ETr}\left(AQ^{-1}BQ\right)^{2}}\label{eq:1-pFINAL}\\
 & = & \frac{d\left(2k-1\right)\mathbb{E}\left(\left\Vert uv\right\Vert _{\mbox{F}}^{2}-\left\Vert uv\left(uv\right)^{\dagger}\right\Vert _{\mbox{F}}^{2}\right)}{\left(\frac{km\left(n-1\right)}{n\left(m-1\right)}\right)^{2}\left\{ 1-m\mathbb{E}\left(q_{ij}^{4}\right)\right\} }.\nonumber 
\end{eqnarray}
The dependence on the covariance matrix has cancelled- a covariance
matrix is one whose element in the $i,j$ position is the covariance
between the $i^{th}$ and $j^{th}$ eigenvalue. This shows that $p$
is independent of eigenvalues of the local terms which proves the
universality lemma.

\noindent \begin{flushleft}
Comment: To get the numerator we used permutation invariance of $A$
and $B$ and local terms, to get the denominator we used permutation
invariance of $Q$. 
\par\end{flushleft}
\end{proof}
Comment: It is interesting that the amount of mixture of the two extremes
needed to capture the quantum spectrum is independent of the actual
types of local terms. It only depends on the physical parameters of
the lattice.

\subsection{\label{sub:The-Slider-Theorem}The Slider Theorem and a Summary}

In this section we make explicit use of $\beta-$Haar properties of
$Q$ and local terms. To prove that there exists a $0\le p\le1$ such
that the combination in Eq. \ref{eq:convex} is convex we need to
evaluate the expected Frobenius norms in Eq. \ref{eq:Quantum-Classical}.
\begin{lem}
$\mathbb{E}\left\Vert uv\right\Vert _{F}^{2}=1/d$ and $\mathbb{E}\left\Vert uv\left(uv\right)^{\dagger}\right\Vert _{F}^{2}=\frac{\beta^{2}\left[3d\left(d-1\right)+1\right]+2\beta\left(3d-1\right)+4}{d\left(\beta d^{2}+2\right)^{2}}$,
when local terms have $\beta-$Haar eigenvectors. \end{lem}
\begin{proof}
It is a fact that $G=u\chi_{\beta d^{2}}$, when $u$ is uniform on
a sphere, $G$ is a $d\times d$ $\beta-$Gaussian matrix whose expected
Frobenius norm has a $\chi-$distribution denoted here by $\chi_{\beta d^{2}}$
(similarly for $v$). Recall that $\mathbb{E}\left(\chi_{h}^{2}\right)=h$
and $\mathbb{E}\left(\chi_{h}^{4}\right)=h\left(h+2\right)$.

\begin{eqnarray}
\mathbb{E}\left\Vert uv\right\Vert _{\mbox{F}}^{2} & \mathbb{E}\left(\chi_{\beta d^{2}}\right)^{2}= & \mathbb{E}\left\Vert \left(G_{1}G_{2}\right)\right\Vert _{\mbox{F}}^{2}\label{eq:SVclass}\\
\Rightarrow\mathbb{E}\left\Vert uv\right\Vert _{\mbox{F}}^{2} & = & \frac{1}{\left(\beta d^{2}\right)^{2}}\mathbb{E}\left\Vert G_{1}G_{2}\right\Vert _{\mbox{F}}^{2}=\frac{d^{2}}{\left(\beta d^{2}\right)^{2}}\mathbb{E}\sum_{k=1}^{d}\left(g_{i,k}^{(1)}g_{kj}^{(2)}\right)^{2}\nonumber \\
 & = & \frac{d^{2}}{\left(\beta d^{2}\right)^{2}}d\left(\beta\right)^{2}=\frac{1}{d}.\nonumber 
\end{eqnarray}
The quantum case, $u^{\dagger}u=\frac{G_{1}^{\dagger}G_{1}}{\left\Vert G_{1}\right\Vert _{\mbox{F}}^{2}}\equiv\frac{W_{1}}{\left\Vert G_{1}\right\Vert _{\mbox{F}}^{2}}$
, similarly $v^{\dagger}v=\frac{G_{2}^{\dagger}G_{2}}{\left\Vert G_{2}\right\Vert _{\mbox{F}}^{2}}\equiv\frac{W_{2}}{\left\Vert G_{2}\right\Vert _{\mbox{F}}^{2}}$,
where $W_{1}$ and $W_{2}$ are Wishart matrices. 

\begin{equation}
\mathbb{E}\left\Vert uv\left(uv\right)^{\dagger}\right\Vert _{\mbox{F}}^{2}=\frac{\mathbb{E}\mbox{Tr}\left(W_{1}W_{2}\right)^{2}}{\mathbb{E}\left(\chi_{d^{2}\beta}^{4}\right)\mathbb{E}\left(\chi_{d^{2}\beta}^{4}\right)}=\frac{\mathbb{E}\mbox{Tr}\left(W_{1}W_{2}\right)^{2}}{\left[d^{2}\beta\left(d^{2}\beta+2\right)\right]^{2}}\label{eq:QuantumFrob}
\end{equation}
hence the complexity of the problem is reduced to finding the expectation
of the trace of a product of Wishart matrices. 

\begin{equation}
\mathbb{E}\mbox{Tr}\left(W_{1}W_{2}\right)^{2}=\mathbb{E}\mbox{Tr}\left(W_{1}W_{2}W_{1}W_{2}\right)=\mathbb{E}\sum_{1\le ijkl\le d}x_{i}x_{i}^{\dagger}y_{j}y_{j}^{\dagger}x_{k}x_{k}^{\dagger}y_{l}y_{l}^{\dagger}\equiv\Pi\left[\begin{array}{cc}
x_{i}^{\dagger}y_{j} & y_{l}^{\dagger}x_{i}\\
y_{j}^{\dagger}x_{k} & x_{k}^{\dagger}y_{l}
\end{array}\right],\label{eq:PROD}
\end{equation}
where $\Pi$ denotes the product of the elements of the matrix. There
are three types of expectations summarized in Table \ref{tab:Expectation-values.}.

\begin{table}
\noindent \begin{centering}
\begin{tabular}{|c|c|c|}
\hline 
Notation & Type & Count\tabularnewline
\hline 
\hline 
$X$ & $\begin{array}{ccc}
i\neq k & \& & j\neq l\end{array}$ & $d^{2}\left(d-1\right)^{2}$\tabularnewline
\hline 
$Y$ & $\begin{array}{ccc}
i=k & \& & j\neq l\\
 & \mbox{or}\\
i\neq k & \& & j=l
\end{array}$ & $2d^{2}\left(d-1\right)$\tabularnewline
\hline 
$Z$ & $\begin{array}{ccc}
i=k & \& & j=l\end{array}$ & $d^{2}$\tabularnewline
\hline 
\end{tabular}
\par\end{centering}

\caption{\label{tab:Expectation-values.}Expectation values.}
\end{table}

In Table \ref{tab:Expectation-values.}

\begin{eqnarray*}
X & \equiv & \mathbb{E}\left[\Pi\left(x_{i}x_{k}\right)\left(y_{i}y_{l}\right)\right]\\
Y & \equiv & \mathbb{E}\left[\left(x_{i}^{\dagger}y_{j}\right)^{2}\left(x_{i}^{\dagger}y_{l}\right)^{2}\right]\\
Z & \equiv & \mathbb{E}\left[\left(x_{i}^{\dagger}y_{i}\right)^{4}\right].
\end{eqnarray*}
We now evaluate these expectation values. We have

\[
X=\Pi\left(\begin{array}{cc}
\chi_{\beta d} & g_{\beta}\\
0 & \chi_{\beta\left(d-1\right)}\\
0 & 0\\
\vdots & \vdots
\end{array}\right)^{\dagger}\left(\begin{array}{cc}
g_{\beta} & g_{\beta}\\
g_{\beta} & g_{\beta}\\
\mbox{DC} & \mbox{DC}\\
\vdots & \vdots
\end{array}\right)
\]
by $QR$ decomposition, where $g_{\beta}$ and $\chi_{h}$ denote
an element with a $\beta-$Gaussian and $\chi_{h}$ distribution respectively;
DC means ``Don't Care''. Consequently

\begin{eqnarray*}
X & = & \Pi\left(\begin{array}{cc}
\chi_{\beta d} & g_{\beta}\\
0 & \chi_{\beta\left(d-1\right)}
\end{array}\right)\left(\begin{array}{cc}
a & b\\
c & d
\end{array}\right)\\
 & =\Pi & \left[\begin{array}{cc}
a\chi_{\beta d} & g_{\beta}a+\chi_{\beta\left(d-1\right)}c\\
b\chi_{\beta d} & g_{\beta}b+\chi_{\beta\left(d-1\right)}d
\end{array}\right]=\Pi\left[\begin{array}{cc}
a\chi_{\beta d} & g_{\beta}a\\
b\chi_{\beta d} & g_{\beta}b
\end{array}\right]\\
 & = & \chi_{\beta d}^{2}a^{2}b^{2}g_{\beta}^{2}=\beta^{4}d.
\end{eqnarray*}
where we denoted the four independent Gaussian entries by $a,b,c,d$
to not confuse them as one number. From Eq. \ref{eq:PROD} we have

\begin{eqnarray*}
Y & = & \mathbb{E}\left[\left(x_{i}^{\dagger}y_{j}\right)^{2}\left(x_{i}^{\dagger}y_{l}\right)^{2}\right]=\mathbb{E}\left(\chi_{d\beta}g_{\beta}^{\left(1\right)}\right)^{2}\left(\chi_{d\beta}g_{\beta}^{\left(2\right)}\right)^{2}=\beta d\left(\beta d+2\right)\beta^{2}\\
Z & = & \mathbb{E}\left(x^{\dagger}y\right)^{4}=\mathbb{E}\left(\chi_{\beta d}^{4}\right)\mathbb{E}\left(\chi_{\beta}^{4}\right)=\beta d\left(\beta d+2\right)\beta\left(\beta+2\right).
\end{eqnarray*}
Eq. \ref{eq:QuantumFrob} now reads

\begin{equation}
\mathbb{E}\left\Vert uv\left(uv\right)^{\dagger}\right\Vert _{F}^{2}=\frac{\beta^{2}\left[3d\left(d-1\right)+1\right]+2\beta\left(3d-1\right)+4}{d\left(\beta d^{2}+2\right)^{2}}.\label{eq:SVq}
\end{equation}
\end{proof}
\begin{thm*}
\textbf{\textup{(The Slider Theorem)\label{thm:The-Slider-Theorem}}}\textbf{
}The quantum kurtosis lies in between the classical and the iso kurtoses,
$\gamma_{2}^{iso}\leq\gamma_{2}^{q}\leq\gamma_{2}^{c}$. Therefore
there exists a $0\leq p\leq1$ such that $\gamma_{2}^{q}=p\gamma_{2}^{c}+\left(1-p\right)\gamma_{2}^{iso}$.
Further, $\lim_{N\rightarrow\infty}p=1$.\end{thm*}
\begin{proof}
We have $\left\{ 1-\frac{1}{m}\sum_{ij=1}^{m}q_{ij}^{4}\right\} \geq0$,
since $\sum_{ij}q_{ij}^{4}\leq\sum_{ij}q_{ij}^{2}=m.$ The last inequality
follows from $q_{ij}^{2}\le1$ . Therefore, Eq. \ref{eq:denom} is

\begin{eqnarray*}
\gamma_{2}^{iso}-\gamma_{2}^{c} & = & \frac{2}{\sigma^{4}}\left(m_{2}^{\mbox{odd}}-m_{11}^{\mbox{odd}}\right)\left(m_{2}^{\mbox{even}}-m_{11}^{\mbox{even}}\right)\times\\
 &  & \left(\frac{km\left(n-1\right)}{n\left(m-1\right)}\right)^{2}\left\{ m\mathbb{E}\left(q_{11}^{4}\right)-1\right\} \le0.
\end{eqnarray*}

From Eqs. \ref{eq:QdepartFinal} and \ref{eq:CDepart} and using the
fact that the singular values $\sigma_{i}\le1$ we have

\[
\left\Vert uv\left(uv\right)^{\dagger}\right\Vert _{\mbox{F}}^{2}=\sum_{i=1}^{d}\sigma_{i}^{4}\le\sum_{i=1}^{d}\sigma_{i}^{2}=\left\Vert uv\right\Vert _{\mbox{F}}^{2}
\]
which proves $\gamma_{2}^{q}-\gamma_{2}^{c}\leq0$. In order to establish
$\gamma_{2}^{iso}\leq\gamma_{2}^{q}\leq\gamma_{2}^{c}$, we need to
show that $\gamma_{2}^{c}-\gamma_{2}^{q}\le\gamma_{2}^{c}-\gamma_{2}^{iso}$.
Eq. \ref{eq:1-pFINAL} after substituting $m\mathbb{E}\left(q_{ij}^{4}\right)=\frac{\beta+2}{\left(m\beta+2\right)}$
from Table \ref{tab:HaarExp} and Eqs. \ref{eq:SVclass}, \ref{eq:SVq}
reads

\begin{equation}
1-p=\left(1-d^{-2k-1}\right)\left[1-\left(\frac{k-1}{k}\right)^{2}\right]\left\{ \left(1-\frac{1-d^{-2k+1}}{1+\beta d^{2}/2}\right)\left(\frac{d}{d+1}\right)^{2}\left[\frac{\beta\left(d^{3}+d^{2}-2d+1\right)+4d-2}{\left(d-1\right)\left(\beta d^{2}+2\right)}\right]\right\} \label{eq:1-pSlider}
\end{equation}
We want to show that $0\le1-p\le1$ for any integer $k\ge1,\mbox{ }d\ge2$
and $\mbox{ }\beta\ge1$. All the factors are manifestly $\ge0$,
therefore $1-p\ge0$. The first two factors are clearly $\le1$ so
we need to prove that the term in the braces is too. Further, $k=1$
provides an upper bound as $\left(1-\frac{1-d^{-2k+1}}{1+\beta d^{2}/2}\right)\le\left(1-\frac{1-d^{-3}}{1+\beta d^{2}/2}\right)$.
We rewrite the term in the braces

\begin{equation}
\frac{d\left(\beta d^{3}+2\right)\left[\beta\left(d^{3}+d^{2}-2d+1\right)+4d-2\right]}{\left(\beta d^{2}+2\right)^{2}\left(d+1\right)^{2}\left(d-1\right)},\label{eq:num-denom}
\end{equation}
but we can subtract the denominator from the numerator to get

\[
\left(\beta d+2\right)\left[\beta\left(d^{4}-2d^{3}\right)+2\left(d^{3}-d^{2}-1\right)\right]\ge0\quad\forall\; d\ge2.
\]

This proves that ($\ref{eq:num-denom}$) is less than one. Therefore,
the term in the braces is less than one and hence $0\le p\le1$. Let
us note the following limits of interest (recall $N-1=2k$)

\begin{eqnarray*}
\lim_{d\rightarrow\infty}\left(1-p\right) & = & \frac{2k-1}{k^{2}}\overset{k=1}{=}1\\
\lim_{N\rightarrow\infty}\left(1-p\right) & \sim & \frac{1}{N}\rightarrow0
\end{eqnarray*}
the first limit tells us that if we consider having two local terms
and take the local dimension to infinity we have essentially free
probability theory as expected. The second limit shows that in the
thermodynamical limit (i.e., $N\rightarrow\infty$) the convex combination
slowly approaches the classical end. In the limit where $\beta\rightarrow\infty$
the $\beta$ dependence in $\left(1-p\right)$ cancels out. This is
a reconfirmation of the fact that in free probability theory, for
$\beta\rightarrow\infty$, the result should be independent of $\beta$.
We see that the bounds are tight. 

\begin{figure}[H]
\begin{centering}
\includegraphics[scale=0.35]{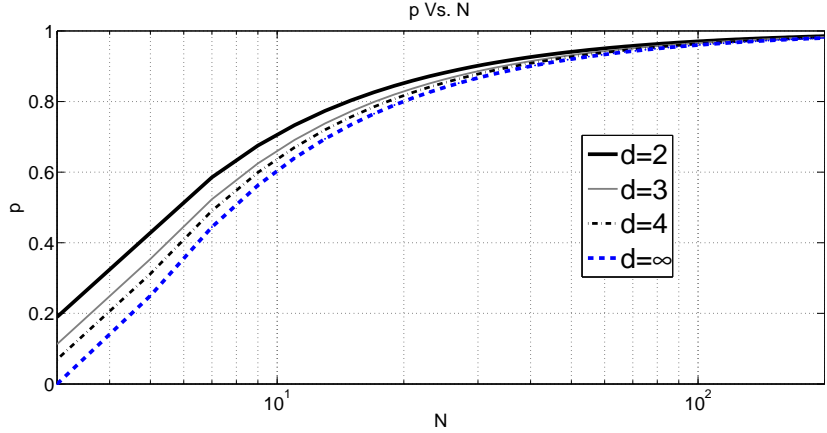}
\par\end{centering}

\centering{}\caption{\label{fig:pVsN}An example: $\beta=1$: the quantum problem for all
$d$ lies in between the iso $(p=0)$ and the classical $(p=1)$. }
\end{figure}

\end{proof}
\begin{flushleft}
Comment: Entanglement shows itself starting at the fourth moment;
further, in the expansion of the fourth moments only the terms that
involve \textit{a pair} of local terms \textit{sharing a site} differ.
Note that when the QMBS possesses a translational symmetry, there
is an additional complication introduced by the dependence of the
local terms. Though, in this case, the non-iid nature of the local
terms complicates the matter theoretically, we have not seen a practical
limitation of IE in our numerical experiments. 
\par\end{flushleft}

\begin{flushleft}
Comment: One could from the beginning use free approximation instead
of isotropic ($m\rightarrow\infty$), in which case the proofs are
simplified.
\par\end{flushleft}

We now \textit{\uline{summarize}} the main thesis of this work.
We are interested in the eigenvalue distribution of 

\[
H\equiv H_{\mbox{odd}\vphantom{\mbox{even}}}+H_{\mbox{even}\mbox{\ensuremath{\vphantom{odd}}}}=\sum_{l=1,3,5,\cdots}\mathbb{I}\otimes H_{l,l+1}\otimes\mathbb{I}+\sum_{l=2,4,6,\cdots}\mathbb{I}\otimes H_{l,l+1}\otimes\mathbb{I},
\]
 which in a basis in that $H_{\mbox{odd}\vphantom{\mbox{even}}}$
is diagonal reads $H=A+Q_{q}^{-1}BQ_{q}$. Since this problem has
little hope in being solved exactly we consider along with it two
known approximations:

\begin{eqnarray*}
H_{c} & = & A+\Pi^{-1}B\Pi\\
H & = & A+Q_{q}^{-1}BQ_{q}\\
H_{iso} & = & A+Q^{-1}BQ.
\end{eqnarray*}

We proved that the first three moments of the three foregoing equations
are equal. We then calculated their fourth moments as encoded by their
kurtoses ($\gamma_{2}$'s) analytically and proved that there exists
a $0\le p\le$1 such that 

\[
\gamma_{2}^{q}=p\gamma_{2}^{c}+\left(1-p\right)\gamma_{2}^{iso}.
\]
It turned out that the only terms in the expansion of the fourth moments
that were relevant were

\begin{equation}
1-p=\frac{\mathbb{E}\mbox{Tr}\left\{ \left(A\Pi^{-1}B\Pi\right)^{2}-\left(AQ_{q}^{-1}BQ_{q}\right)^{2}\right\} }{\mathbb{E}\mbox{Tr}\left\{ \left(A\Pi^{-1}B\Pi\right)^{2}-\left(AQ^{-1}BQ\right)^{2}\right\} }.\label{eq:1-P_Effective}
\end{equation}
Through direct calculation we found that the numerator $\mathbb{E}\mbox{Tr}\left\{ \left(A\Pi^{-1}B\Pi\right)^{2}-\left(AQ_{q}^{-1}BQ_{q}\right)^{2}\right\} $
evaluates to be

\[
d\left(2k-1\right)\left(m_{2}^{\mbox{odd}}-m_{1,1}^{\mbox{odd}}\right)\left(m_{2}^{\mbox{even}}-m_{1,1}^{\mbox{even}}\right)\mathbb{E}\left(\left\Vert uv\right\Vert _{\mbox{F}}^{2}-\left\Vert uv\left(uv\right)^{\dagger}\right\Vert _{\mbox{F}}^{2}\right),
\]
and the denominator $\mathbb{E}\mbox{Tr}\left\{ \left(A\Pi^{-1}B\Pi\right)^{2}-\left(AQ^{-1}BQ\right)^{2}\right\} $

\[
\left(m_{2}^{\mbox{odd}}-m_{1,1}^{\mbox{odd}}\right)\left(m_{2}^{\mbox{even}}-m_{1,1}^{\mbox{even}}\right)\left(\frac{km\left(n-1\right)}{n\left(m-1\right)}\right)^{2}\left\{ 1-m\mathbb{E}\left(q_{ij}^{4}\right)\right\} .
\]

Therefore $1-p$ does not depend on the local distribution and can
generally be expressed as

\framebox{\begin{minipage}[t]{1\columnwidth}%
\[
1-p=\frac{d\left(2k-1\right)\mathbb{E}\left(\left\Vert uv\right\Vert _{\mbox{F}}^{2}-\left\Vert uv\left(uv\right)^{\dagger}\right\Vert _{\mbox{F}}^{2}\right)}{\left(\frac{km\left(n-1\right)}{n\left(m-1\right)}\right)^{2}\left\{ 1-m\mathbb{E}\left(q_{ij}^{4}\right)\right\} }.
\]
\end{minipage}}

\bigskip{}

If we further assume that the local eigenvectors are $\beta-\mbox{Haar}$
distributed we get 

\begin{eqnarray*}
1-p & = & \left(1-d^{-2k-1}\right)\left[1-\left(\frac{k-1}{k}\right)^{2}\right]\left(1-\frac{1-d^{-2k+1}}{1+\beta d^{2}/2}\right)\left(\frac{d}{d+1}\right)^{2}\\
 & \times & \left[\frac{\beta\left(d^{3}+d^{2}-2d+1\right)+4d-2}{\left(d-1\right)\left(\beta d^{2}+2\right)}\right].
\end{eqnarray*}

\medskip{}

Next we asserted that this $p$ can be used to approximate the distribution

\[
d\nu^{q}\approx d\nu^{IE}=pd\nu^{c}+\left(1-p\right)d\nu^{iso}.
\]
We argued that the spectra obtained using Isotropic Entanglement (IE)
are accurate well beyond four moments. 

For illustration, we apply IE theory in full detail to a chain with
Wishart matrices as local terms. Other types of local terms (e.g.
GOE, random $\pm1$ eigenvalues) can be treated similarly; therefore
in Section \ref{sec:Other-Examples} we show the plots comparing IE
with exact diagonalization for these cases.

\section{\label{sec:First-Example:-Wishart}A Detailed Example: Wishart Matrices
as Local Terms }

As an example take a chain with odd number of sites and for the local
terms in Eq. \ref{eq:Hamiltonian} pick $H^{\left(l\right)}=W^{T}W$,
where $W$ is a rank $r$ matrix whose elements are picked randomly
from a Gaussian distribution ($\beta=1$); these matrices $W^{T}W$
are known as \textit{Wishart matrices. }Clearly the maximum possible
rank is $r=d^{2}$ for each of the local terms.

Any cumulant is equal to the corresponding cumulant of one local term,
denoted by $\kappa$, times the number of summands in Eq. \ref{eq:Hamiltonian}.
In particular, the fourth cumulant of $H$ is $\kappa_{4}^{\left(N-1\right)}=\left(N-1\right)\kappa_{4}$.
Below we drop the superscripts when the quantity pertains to the whole
chain. Next we recall the definitions in terms of cumulants of the
mean $(\mu)$, the variance $(\sigma^{2})$, the skewness $(\gamma_{1})$,
and the kurtosis $(\gamma_{2})$ 

\begin{equation}
\begin{array}{cccccccc}
\mu\equiv\kappa_{1} & \quad & \sigma^{2}\equiv\kappa_{2} & \quad & \gamma_{1}\equiv\frac{\kappa_{3}}{\sigma^{3}} & \quad & \gamma_{2}\equiv\frac{\kappa_{4}}{\sigma^{4}}=\frac{m_{4}}{\sigma^{4}}-3 & .\end{array}\label{eq:momCumul}
\end{equation}

\subsection{\label{sub:Wishart-Classical-Case:-Calculation}Evaluation of $p=\frac{\gamma_{2}^{q}-\gamma_{2}^{iso}}{\gamma_{2}^{c}-\gamma_{2}^{iso}}$}

The moments of the local terms are obtained from MOPS \cite{mops}; 

\begin{equation}
\begin{array}{c}
m_{1}=\beta r\\
m_{2}=\beta r\left[\beta\left(r+n-1\right)+2\right]\\
m_{3}=\beta r\left\{ \beta^{2}\left[n^{2}+\left(r-1\right)\left(3n+r-2\right)\right]+6\beta\left(n+r-1\right)+8\right\} \\
m_{4}=\beta r\left\{ 48+\beta^{3}\left[n^{3}+6n^{2}(r-1)+n\left(6r-11\right)\left(r-1\right)-6\left(r^{2}+1\right)+r^{3}+11r\right]\right.\\
\left.+2\beta^{2}\left[6\left(n^{2}+r^{2}\right)+17\left(n(r-1)-r\right)+11\right]+44\beta\left(n+r-1\right)\right\} \\
m_{1,1}=\beta^{2}r\left(r-1\right)
\end{array}
\end{equation}
which for real matrices $\beta=1$ yields

\begin{equation}
\begin{array}{c}
m_{1}=r\\
m_{2}=r\left(r+n+1\right)\\
m_{3}=r\left(n^{2}+3n+3rn+3r+r^{2}+4\right)\\
m_{4}=r\left(6n^{2}+21n+6rn^{2}+17rn+21r+6nr^{2}+6r^{2}+n^{3}+r^{3}+20\right)\\
m_{1,1}=r\left(r-1\right).
\end{array}
\end{equation}

The mean, variance, skewness, and kurtosis are obtained from the foregoing
relations, through the cumulants Eq. \ref{eq:momCumul}. We drop the
superscripts when the quantity pertains to the whole chain. Therefore,
using Eq. \ref{eq:momCumul}, we have

\begin{equation}
\begin{array}{ccc}
\mu\equiv\left(N-1\right)r & \quad & \sigma^{2}\equiv r\left(N-1\right)\left(n+1\right)\\
\gamma_{1}\equiv\frac{n^{2}+3n+4}{\left(n+1\right)^{3/2}\sqrt{r\left(N-1\right)}} & \quad & \gamma_{2}^{\left(c\right)}\equiv\frac{n^{2}\left(n+6\right)-rn\left(n+1\right)+21n+2r+20}{r\left(N-1\right)\left(n+1\right)^{2}}.
\end{array}\label{eq:AllClassical}
\end{equation}

From Eq. \ref{eq:m_2} we readily obtain 

\begin{equation}
\frac{1}{m}\mathbb{E}\textrm{Tr}\left(A\Pi^{T}B\Pi\right)^{2}=r^{2}k^{2}\left(rk+n+1\right)^{2}.\label{eq:Classical}
\end{equation}

By The Matching Three Moments theorem we immediately have the mean,
the variance and the skewness for the isotropic case 

\[
\begin{array}{ccc}
\mu=\left(N-1\right)r &  & \sigma^{2}=r\left(N-1\right)\left(n+1\right)\\
 & \gamma_{1}=\frac{n^{2}+3n+4}{\left(n+1\right)^{3/2}\sqrt{r\left(N-1\right)}}.
\end{array}
\]

Note that the denominator in Eq. \ref{eq:convex} becomes, 

\begin{equation}
\gamma_{2}^{c}-\gamma_{2}^{iso}=\frac{\kappa_{4}^{(c)}-\kappa_{4}^{(iso)}}{\sigma^{4}}=\frac{2}{m}\frac{\mathbb{E}\left\{ \textrm{Tr}\left[\left(A\Pi^{T}B\Pi\right)^{2}-\left(AQ^{T}BQ\right)^{2}\right]\right\} }{r^{2}\left(N-1\right)^{2}\left(n+1\right)^{2}}.\label{eq:denomEx}
\end{equation}

In the case of Wishart matrices, $m_{1}^{\mbox{odd}}=m_{1}^{\mbox{even}}=r$,
and $m_{2}^{\mbox{odd}}=m_{2}^{\mbox{even}}=r\left(r+n+1\right)$,
$m_{11}^{\mbox{odd}}=m_{11}^{\mbox{even}}=r\left(r-1\right)$ given
by Eqs. \ref{eq:m_2} and \ref{eq:elemGeneral} respectively. Therefore
we can substitute these into Eq. \ref{eq:Iso-Classical_FINAL} 

\begin{eqnarray}
\frac{1}{m}\mathbb{E}\left\{ \textrm{Tr}\left[\left(A\Pi^{T}B\Pi\right)^{2}-\left(AQ^{T}BQ\right)^{2}\right]\right\}  & = & \left(m_{2}^{(A)}-m_{1,1}^{(A)}\right)\left(m_{2}^{(B)}-m_{1,1}^{(B)}\right)\left\{ 1-m\mathbb{E}\left(q_{ij}^{4}\right)\right\} \nonumber \\
 & = & \frac{\beta\left(m-1\right)}{\left(m\beta+2\right)}\left(\frac{km\left(n-1\right)}{n\left(m-1\right)}\right)^{2}\left(m_{2}-m_{1,1}\right)^{2}\nonumber \\
 & = & \frac{\beta k^{2}m^{2}\left(n-1\right)^{2}}{\left(m\beta+2\right)\left(m-1\right)n^{2}}\left(m_{2}-m_{1,1}\right)^{2}\label{eq:WishartDenom}
\end{eqnarray}
 One can also calculate each of the terms separately and obtain the
same results (see Appendix for the alternative).

From Eq. \ref{eq:Quantum-Classical} we have 

\begin{eqnarray}
\frac{1}{m}\mathbb{E}\left[\mbox{Tr}\left(A\Pi^{T}B\Pi\right)^{2}-\textrm{Tr}\left(AQ_{q}^{T}BQ_{q}\right)^{2}\right] & = & d\left(2k-1\right)\left(m_{2}-m_{1,1}\right)^{2}\nonumber \\
 & \times & \mathbb{E}\left(\left\Vert uv\right\Vert _{\mbox{F}}^{2}-\left\Vert uv\left(uv\right)^{\dagger}\right\Vert _{\mbox{F}}^{2}\right)\nonumber \\
 & = & \left(2k-1\right)\left(m_{2}-m_{1,1}\right)^{2}\nonumber \\
 & \times & \left\{ \frac{1+d\left(d^{3}+d-3\right)}{\left(d^{2}+2\right)^{2}}\right\} .\label{eq:WishartNum}
\end{eqnarray}
We can divide Eq. \ref{eq:WishartNum} by Eq. \ref{eq:WishartDenom}
to evaluate the parameter $p$.

\subsection{\label{sub:Wishart-Summary}Summary of Our Findings and Further Numerical
Results}

We summarize the results along with numerical experiments in Tables
\ref{tab:summaryTheory} and \ref{tab:SummaryNumerical} to show the
equivalence of the first three moments and the departure of the three
cases in their fourth moment. As said above,

\begin{equation}
Q_{c}=\mathbb{I}_{d^{N}}\qquad Q_{iso}\equiv Q\;\textrm{Haar}\; d^{N}\times d^{N}\qquad Q_{q}=\left(Q_{q}^{(A)}\right)^{T}Q_{q}^{(B)}
\end{equation}

\begin{flushleft}
where, $\left(Q_{q}^{(A)}\right)^{T}Q_{q}^{(B)}$ is given by Eq.
\ref{eq:OddEven}. In addition, from Eq. \ref{eq:AllClassical} we
can define $\Delta$ to be the part of the kurtosis that is equal
among the three cases
\par\end{flushleft}

\begin{center}
$\Delta=$ $\gamma_{2}^{c}$ - $\frac{2m_{2}^{A}m_{2}^{B}}{\sigma^{4}}=$
$\gamma_{2}^{c}-\frac{1}{2}\frac{\left(rk+n+1\right)^{2}}{\left(n+1\right)^{2}}$.
\par\end{center}

Using $\Delta$ we can obtain the full kurtosis for the iso and quantum
case, and therefore (see Table \ref{tab:summaryTheory} for a theoretical
summary): 

\begin{equation}
p=\frac{\gamma_{2}^{q}-\gamma_{2}^{iso}}{\gamma_{2}^{c}-\gamma_{2}^{iso}}.
\end{equation}

\begin{table}
\begin{centering}
\begin{tabular}{|c||c||c||c|}
\hline 
$\beta=1$ Wishart & Iso & Quantum & Classical \tabularnewline
\hline 
\hline 
Mean $\mu$ & \multicolumn{3}{c|}{$r\left(N-1\right)$}\tabularnewline
\hline 
Variance $\sigma^{2}$ & \multicolumn{3}{c|}{$r\left(N-1\right)\left(d^{2}+1\right)$}\tabularnewline
\hline 
Skewness $\gamma_{1}$ & \multicolumn{3}{c|}{$\frac{d^{4}+3d^{2}+4}{\sqrt{r\left(N-1\right)\left(d^{2}+1\right)^{3}}}$}\tabularnewline
\hline 
$\frac{1}{m}\mathbb{E}\left[\textrm{Tr}\left(AQ_{\bullet}^{T}BQ_{\bullet}\right)^{2}\right]$ & $m_{2}^{A}m_{2}^{B}-$ Eq. \ref{eq:WishartDenom} & $m_{2}^{A}m_{2}^{B}-$ Eq. \ref{eq:WishartNum} & $r^{2}k^{2}\left(rk+n+1\right)^{2}$\tabularnewline
\hline 
Kurtosis $\gamma_{2}^{\left(\bullet\right)}$ & $\frac{2}{m\sigma^{4}}\mathbb{E}\left[\textrm{Tr}\left(AQ^{T}BQ\right)^{2}\right]$$+\Delta$ & $\frac{2}{m\sigma^{4}}\mathbb{E}\left[\textrm{Tr}\left(AQ_{q}^{T}BQ_{q}\right)^{2}\right]$$+\Delta$ & Eq. \ref{eq:AllClassical}\tabularnewline
\hline 
\end{tabular}
\par\end{centering}

\centering{}\caption{\label{tab:summaryTheory}Summary of the results when the local terms
are Wishart matrices. The fourth moment is where the three cases differ.}
\end{table}

\begin{table}
\begin{centering}
\begin{longtable}{|c||c||c||c|c||c||c|c|c|c|}
\hline 
\noalign{\vskip\doublerulesep}
\multicolumn{10}{|c|}{Experiments based on $500000$ trials}\tabularnewline[1sp]
\hline 
\noalign{\vskip\doublerulesep}
\multicolumn{4}{|c|}{$N=3$ } & \multicolumn{3}{c|}{Theoretical value} & \multicolumn{3}{c|}{Numerical Experiment }\tabularnewline[1sp]
\hline 
\noalign{\vskip\doublerulesep}
\multicolumn{4}{|c|}{} & Iso & Quantum & Classical & \multicolumn{1}{c||}{Iso} & \multicolumn{1}{c||}{Quantum} & Classical\tabularnewline[1sp]
\hline 
\hline 
\noalign{\vskip\doublerulesep}
\multicolumn{4}{|c|}{Mean $\mu$} & \multicolumn{3}{c|}{$8$} & $8.007$ & $8.007$ & $7.999$\tabularnewline[1sp]
\hline 
\noalign{\vskip\doublerulesep}
\multicolumn{4}{|c|}{Variance $\sigma^{2}$} & \multicolumn{3}{c|}{$40$} & $40.041$ & $40.031$ & $39.976$\tabularnewline[1sp]
\hline 
\noalign{\vskip\doublerulesep}
\multicolumn{4}{|c|}{Skewness $\gamma_{1}$} & \multicolumn{3}{c|}{$\frac{8}{25}\sqrt{10}=1.01192$} & $1.009$ & $1.009$ & $1.011$\tabularnewline[1sp]
\hline 
\noalign{\vskip\doublerulesep}
\multicolumn{4}{|c|}{Kurtosis $\gamma_{2}$} & \multicolumn{1}{c|}{$\frac{516}{875}=0.590$} & \multicolumn{1}{c|}{$\frac{33}{50}=0.660$} & $\frac{24}{25}=0.960$  & $0.575$ & $0.645$ & $0.953$\tabularnewline[1sp]
\hline 
\noalign{\vskip\doublerulesep}
\multicolumn{10}{|c|}{Experiments based on $500000$ trials}\tabularnewline[1sp]
\hline 
\noalign{\vskip\doublerulesep}
\multicolumn{4}{|c|}{$N=5$ } & \multicolumn{3}{c|}{Theoretical value} & \multicolumn{3}{c|}{Numerical Experiment }\tabularnewline[1sp]
\hline 
\noalign{\vskip\doublerulesep}
\multicolumn{4}{|c|}{} & Iso & Quantum & Classical & \multicolumn{1}{c||}{Iso} & \multicolumn{1}{c||}{Quantum} & Classical\tabularnewline[1sp]
\hline 
\hline 
\noalign{\vskip\doublerulesep}
\multicolumn{4}{|c|}{Mean $\mu$} & \multicolumn{3}{c|}{$16$} & $15.999$ & $15.999$ & $16.004$\tabularnewline[1sp]
\hline 
\noalign{\vskip\doublerulesep}
\multicolumn{4}{|c|}{Variance $\sigma^{2}$} & \multicolumn{3}{c|}{$80$} & $79.993$ & $80.005$ & $80.066$\tabularnewline[1sp]
\hline 
\noalign{\vskip\doublerulesep}
\multicolumn{4}{|c|}{Skewness $\gamma_{1}$} & \multicolumn{3}{c|}{$\frac{8}{25}\sqrt{5}=0.716$} & $0.715$ & $0.715$ & $0.717$\tabularnewline[1sp]
\hline 
\noalign{\vskip\doublerulesep}
\multicolumn{4}{|c|}{Kurtosis $\gamma_{2}$} & \multicolumn{1}{c|}{$\frac{228}{2635}=0.087$} & \multicolumn{1}{c|}{$\frac{51}{200}=0.255$} & $\frac{12}{25}=0.48$  & $0.085$ & $0.255$ & $0.485$\tabularnewline[1sp]
\hline 
\noalign{\vskip\doublerulesep}
\multicolumn{10}{|c|}{Experiments based on $300000$ trials}\tabularnewline[1sp]
\hline 
\noalign{\vskip\doublerulesep}
\multicolumn{4}{|c|}{$N=7$ } & \multicolumn{3}{c|}{Theoretical value} & \multicolumn{3}{c|}{Numerical Experiment }\tabularnewline[1sp]
\hline 
\noalign{\vskip\doublerulesep}
\multicolumn{4}{|c|}{} & Iso & Quantum & Classical & \multicolumn{1}{c||}{Iso} & \multicolumn{1}{c||}{Quantum} & Classical\tabularnewline[1sp]
\hline 
\hline 
\noalign{\vskip\doublerulesep}
\multicolumn{4}{|c|}{Mean $\mu$} & \multicolumn{3}{c|}{$24$} & $23.000$ & $23.000$ & $24.095$\tabularnewline[1sp]
\hline 
\noalign{\vskip\doublerulesep}
\multicolumn{4}{|c|}{Variance $\sigma^{2}$} & \multicolumn{3}{c|}{$120$} & $120.008$ & $120.015$ & $120.573$\tabularnewline[1sp]
\hline 
\noalign{\vskip\doublerulesep}
\multicolumn{4}{|c|}{Skewness $\gamma_{1}$} & \multicolumn{3}{c|}{$\frac{8}{75}\sqrt{30}=0.584$} & $0.585$ & $0.585$ & $0.588$\tabularnewline[1sp]
\hline 
\noalign{\vskip\doublerulesep}
\multicolumn{4}{|c|}{Kurtosis $\gamma_{2}$} & \multicolumn{1}{c|}{$-\frac{16904}{206375}=-0.082$} & \multicolumn{1}{c|}{$\frac{23}{150}=0.153$} & $\frac{8}{25}=0.320$  & $-0.079$ & $0.156$ & $0.331$\tabularnewline[1sp]
\newpage
\hline 
\noalign{\vskip\doublerulesep}
\multicolumn{10}{|c|}{Experiments based on $40000$ trials}\tabularnewline[1sp]
\hline 
\noalign{\vskip\doublerulesep}
\multicolumn{4}{|c|}{$N=9$ } & \multicolumn{3}{c|}{Theoretical value} & \multicolumn{3}{c|}{Numerical Experiment }\tabularnewline[1sp]
\hline 
\noalign{\vskip\doublerulesep}
\multicolumn{4}{|c|}{} & Iso & Quantum & Classical & \multicolumn{1}{c||}{Iso} & \multicolumn{1}{c||}{Quantum} & Classical\tabularnewline[1sp]
\hline 
\hline 
\noalign{\vskip\doublerulesep}
\multicolumn{4}{|c|}{Mean $\mu$} & \multicolumn{3}{c|}{$32$} & $32.027$ & $32.027$ & $31.777$\tabularnewline[1sp]
\hline 
\noalign{\vskip\doublerulesep}
\multicolumn{4}{|c|}{Variance $\sigma^{2}$} & \multicolumn{3}{c|}{$160$} & $160.074$ & $160.049$ & $157.480$\tabularnewline[1sp]
\hline 
\noalign{\vskip\doublerulesep}
\multicolumn{4}{|c|}{Skewness $\gamma_{1}$} & \multicolumn{3}{c|}{$\frac{4}{25}\sqrt{10}=0.506$} & $0.505$ & $0.506$ & $0.500$\tabularnewline[1sp]
\hline 
\noalign{\vskip\doublerulesep}
\multicolumn{4}{|c|}{Kurtosis $\gamma_{2}$} & \multicolumn{1}{c|}{$-\frac{539142}{3283175}=-0.164$} & \multicolumn{1}{c|}{$\frac{87}{800}=0.109$} & $\frac{6}{25}=0.240$  & $-0.165$ & $0.109$ & $0.213$\tabularnewline[1sp]
\hline 
\noalign{\vskip\doublerulesep}
\multicolumn{10}{|c|}{Experiments based on $2000$ trials}\tabularnewline[1sp]
\hline 
\noalign{\vskip\doublerulesep}
\multicolumn{4}{|c|}{$N=11$} & \multicolumn{3}{c|}{Theoretical value} & \multicolumn{3}{c|}{Numerical Experiment }\tabularnewline[1sp]
\hline 
\noalign{\vskip\doublerulesep}
\multicolumn{4}{|c|}{} & \multicolumn{1}{c||}{Iso} & Quantum & Classical & \multicolumn{1}{c||}{Iso} & \multicolumn{1}{c||}{Quantum} & Classical\tabularnewline[1sp]
\hline 
\noalign{\vskip\doublerulesep}
\multicolumn{4}{|c|}{Mean $\mu$} & \multicolumn{3}{c|}{$40$} & $39.973$ & $39.973$ & $39.974$\tabularnewline[1sp]
\hline 
\noalign{\vskip\doublerulesep}
\multicolumn{4}{|c|}{Variance $\sigma^{2}$} & \multicolumn{3}{c|}{$200$} & $200.822$ & $200.876$ & $197.350$\tabularnewline[1sp]
\hline 
\noalign{\vskip\doublerulesep}
\multicolumn{4}{|c|}{Skewness $\gamma_{1}$} & \multicolumn{3}{c|}{$\frac{8}{25}\sqrt{2}=0.452548$} & $0.4618$ & $0.4538$ & $0.407$\tabularnewline[1sp]
\hline 
\noalign{\vskip\doublerulesep}
\multicolumn{4}{|c|}{Kurtosis $\gamma_{2}$} & \multicolumn{1}{c|}{$-\frac{11162424}{52454375}=-0.213$} & \multicolumn{1}{c|}{$\frac{21}{250}=0.084$} & $\frac{24}{125}=0.192$  & $-0.189$ & $0.093$ & $0.102$\tabularnewline[1sp]
\hline 
\end{longtable}\medskip{}

\par\end{centering}

\caption{\label{tab:SummaryNumerical}The mean, variance and skewness of classical,
iso and quantum results match. However, the fourth moments (kurtoses)
differ. Here we are showing results for $d=2,\; r=4$ with an accuracy
of three decimal points.}
\end{table}

\pagebreak{}

The numerical convergence of the kurtoses to the theoretical values
were rather slow. To make sure the results are consistent we did a
large run with $500$ million trials for $N=5$, $d=2$, $r=3$ and
$\beta=1$ and obtained four digits of accuracy

\[
\begin{array}{c}
\qquad\gamma_{2}^{c}-\gamma_{2}^{iso}=0.39340\qquad\textrm{Numerical\;\ experiment}\\
\gamma_{2}^{c}-\gamma_{2}^{iso}=0.39347\qquad\textrm{Theoretical\;\ value.}
\end{array}
\]

Convergence is faster if one calculates $p$ based on the departing
terms alone (Eq. \ref{eq:1-P_Effective}). In this case, for full
rank Wishart matrices with $N=5$ and $d=2$

\begin{center}
\begin{tabular}{|c|c|}
\hline 
$\beta=1$, trials: $5$ Million & $1-p$\tabularnewline
\hline 
\hline 
Numerical Experiment & $0.57189$\tabularnewline
\hline 
Theoretical Value & $0.57183$\tabularnewline
\hline 
\end{tabular}\hspace{1cc}%
\begin{tabular}{|c|c|}
\hline 
$\beta=2$, trials: $10$ Million & $1-p$\tabularnewline
\hline 
\hline 
Numerical Experiment & $0.63912$\tabularnewline
\hline 
Theoretical Value & $0.63938$\tabularnewline
\hline 
\end{tabular}
\par\end{center}

\[
\]

Below we compare our theory against exact diagonalization for various
number of sites $N$, local ranks $r$, and site dimensionality $d$
(Figures \ref{fig:N=00003D3}-\ref{fig:N=00003D11}).

\begin{figure}[H]
\begin{centering}
\includegraphics[scale=0.4]{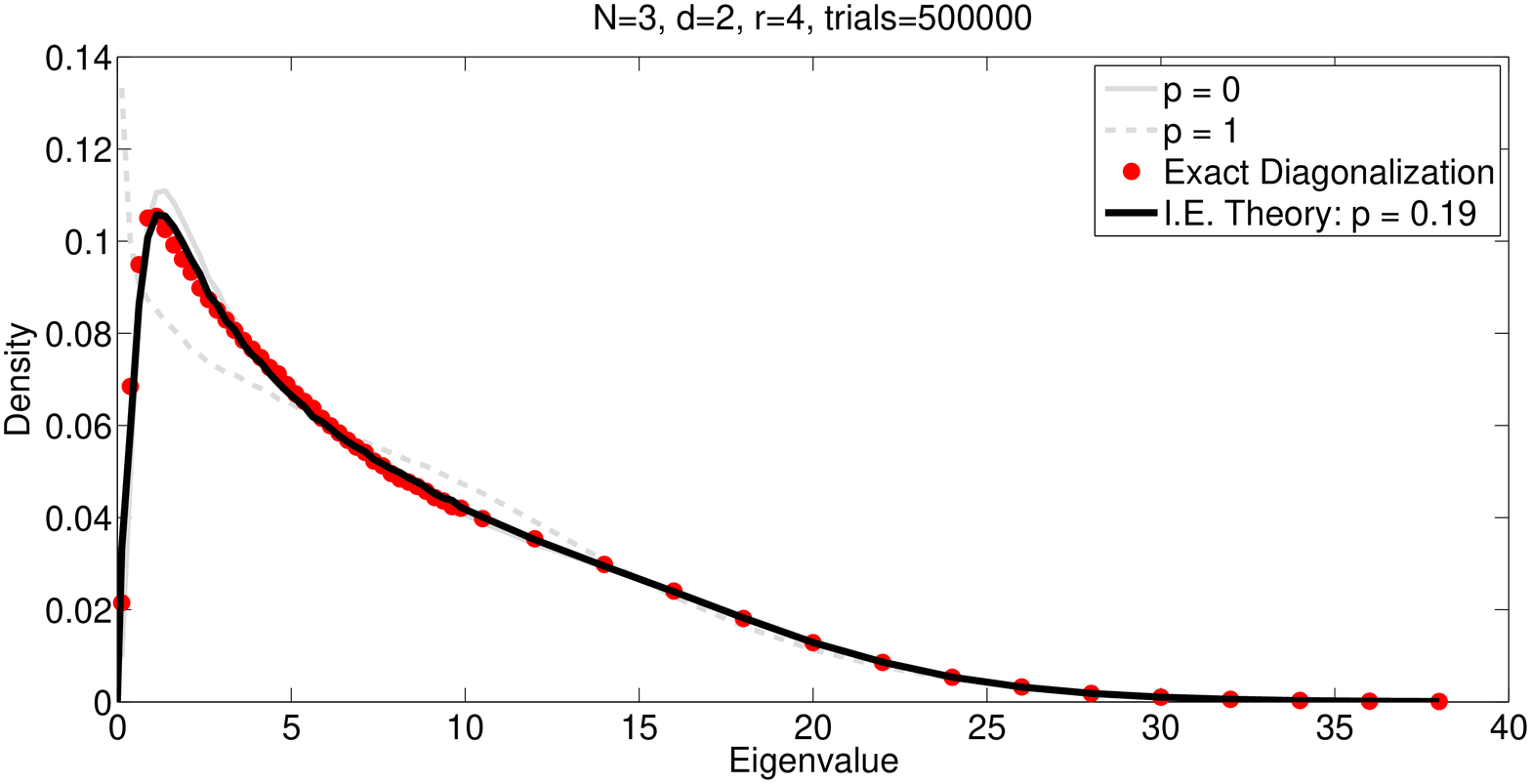}\vspace{0.2in}

\par\end{centering}

\begin{centering}
\includegraphics[scale=0.4]{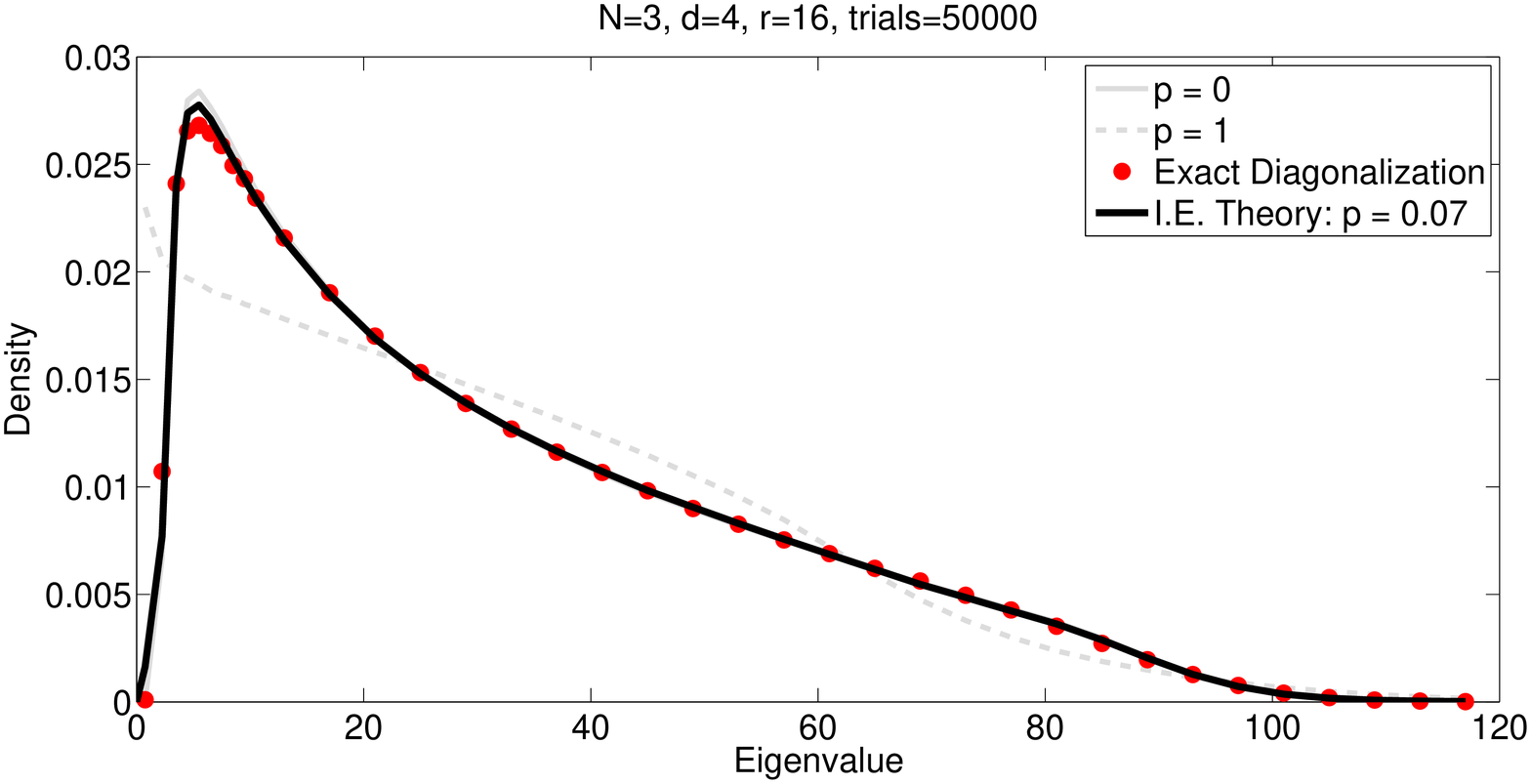}
\par\end{centering}

\begin{centering}
\vspace{0.2in}

\par\end{centering}

\begin{centering}
\includegraphics[scale=0.4]{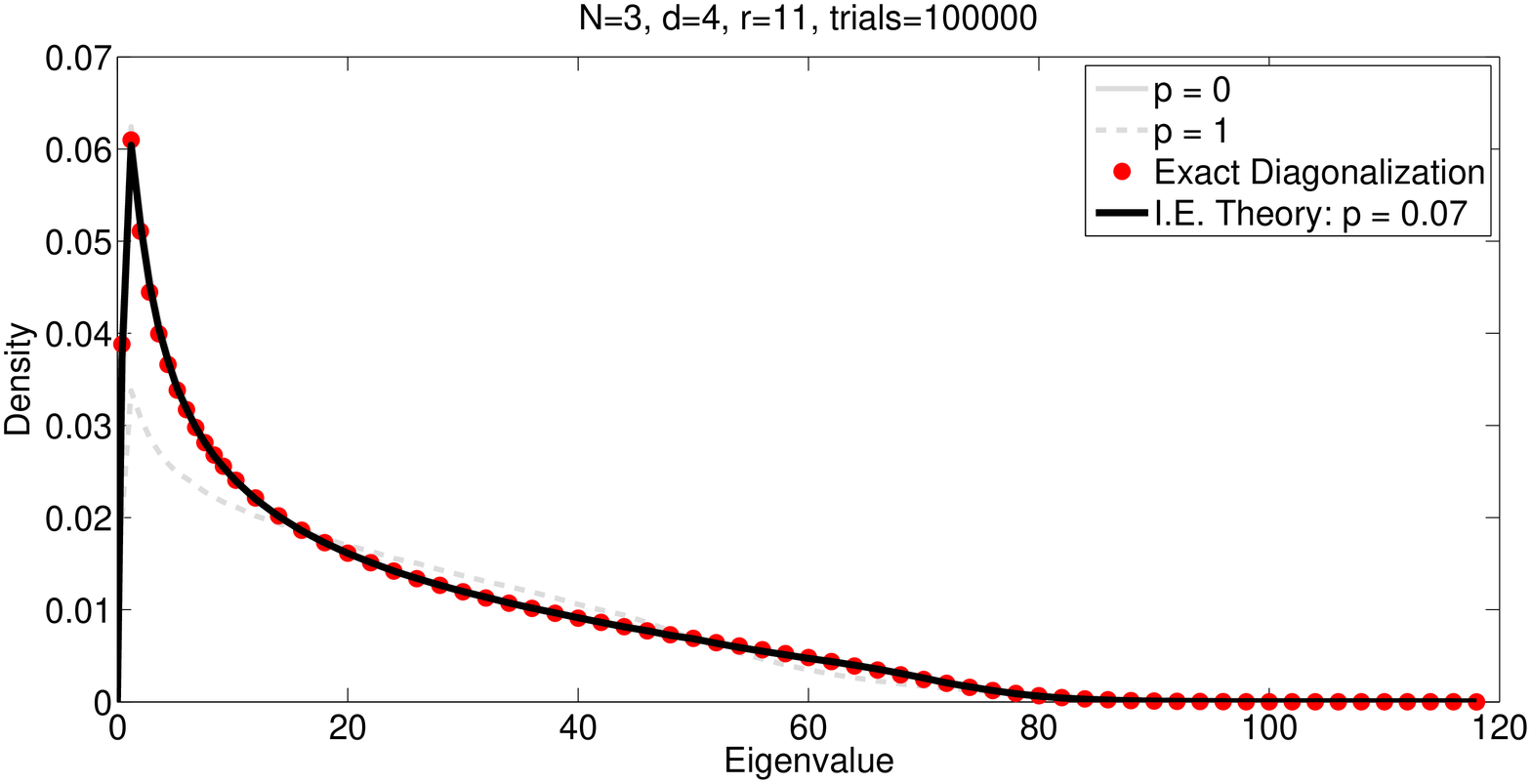}
\par\end{centering}

\centering{}\caption{\label{fig:N=00003D3}$N=3$ examples. Note that the last two plots
have the same $p$ despite having different ranks $r$. This is a
consequence of the Universality Lemma since they have the same $N$
and $d$. }
\end{figure}

\begin{figure}[H]
\begin{centering}
\includegraphics[scale=0.4]{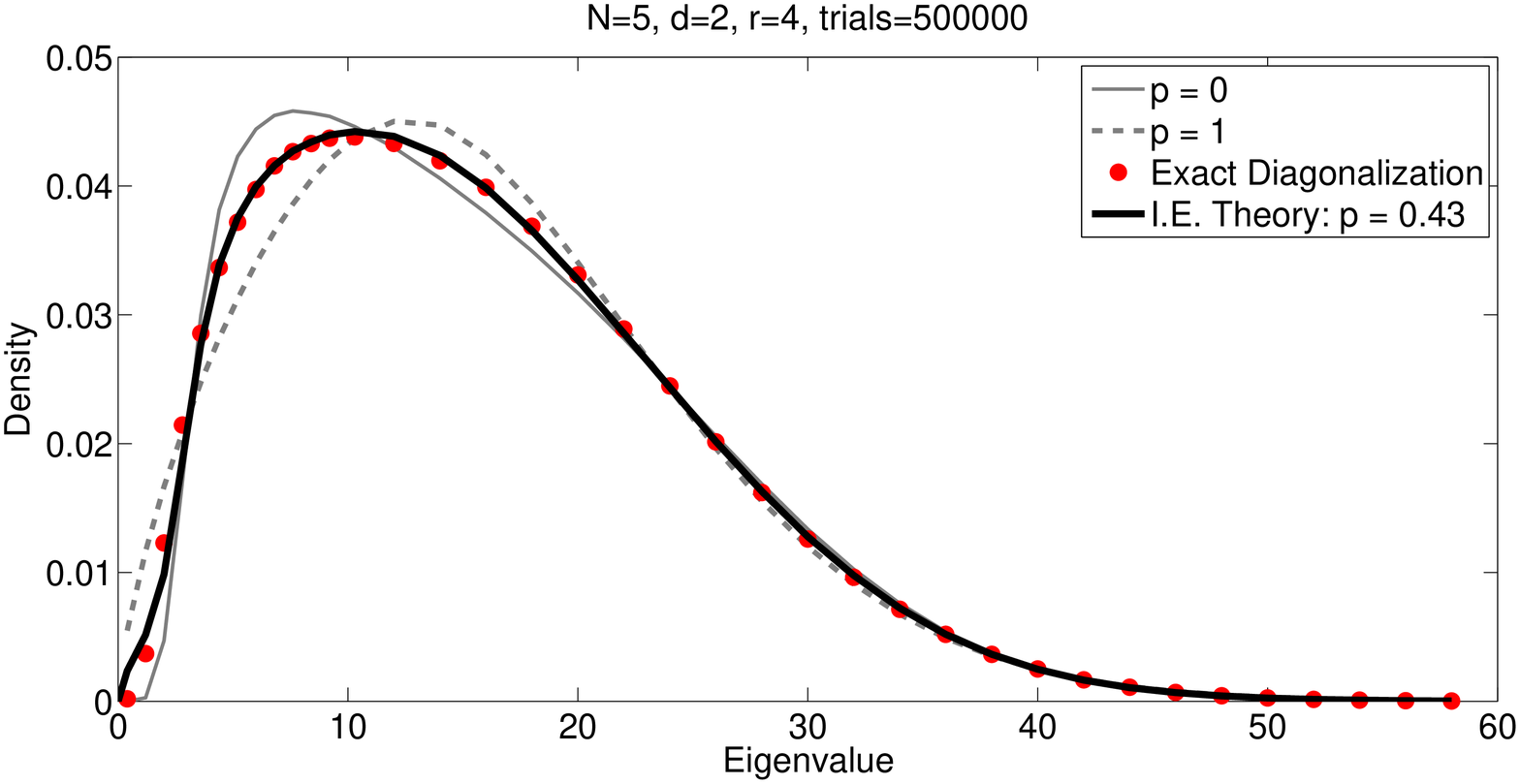}
\par\end{centering}

\begin{centering}
\vspace{0.2in}

\par\end{centering}

\begin{centering}
\includegraphics[scale=0.4]{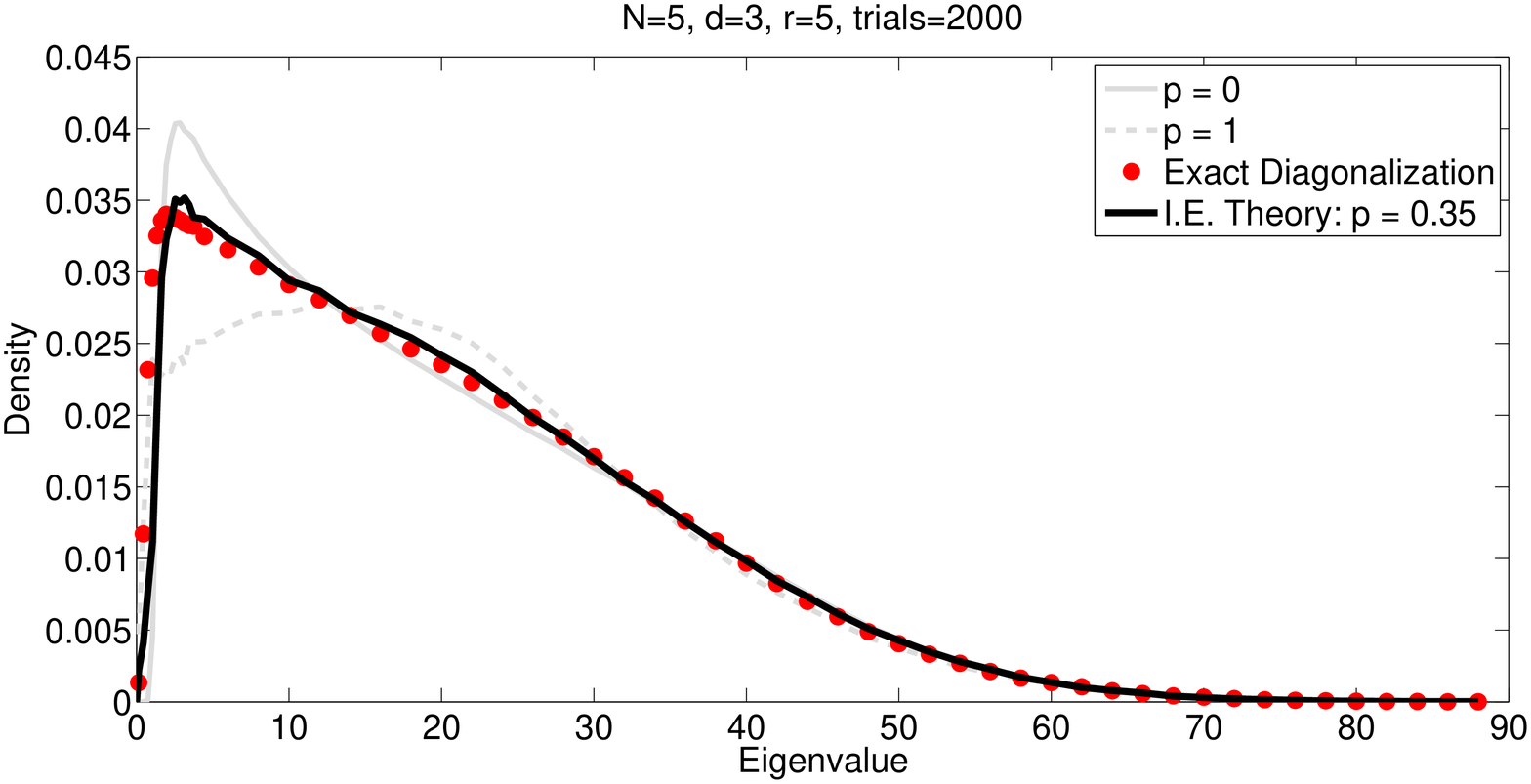}
\par\end{centering}

\centering{}\caption{$N=5$}
\end{figure}

\begin{figure}[H]
\begin{centering}
\includegraphics[scale=0.4]{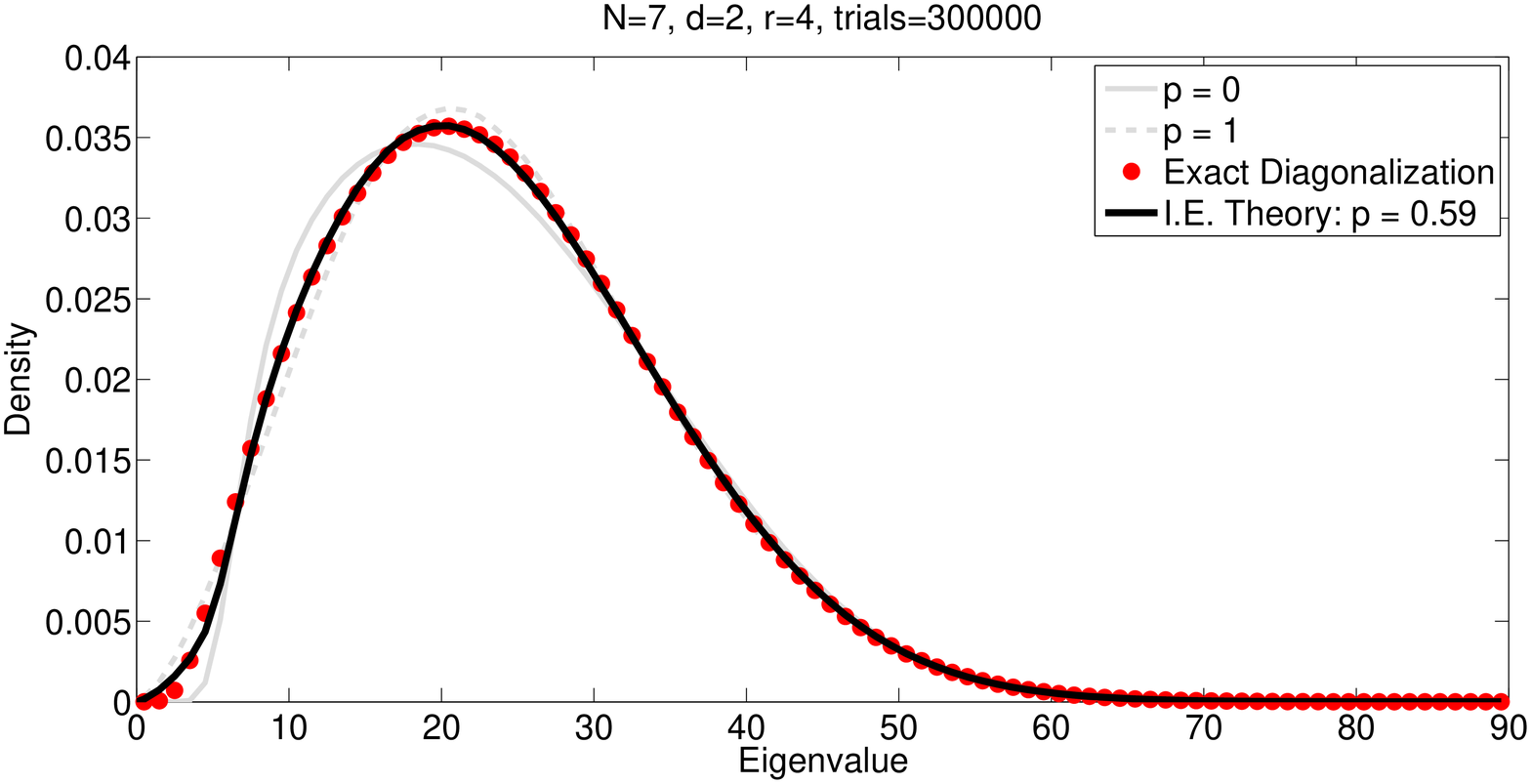}
\par\end{centering}

\centering{}\caption{$N=7$}
\end{figure}

\begin{figure}[H]
\begin{centering}
\includegraphics[scale=0.4]{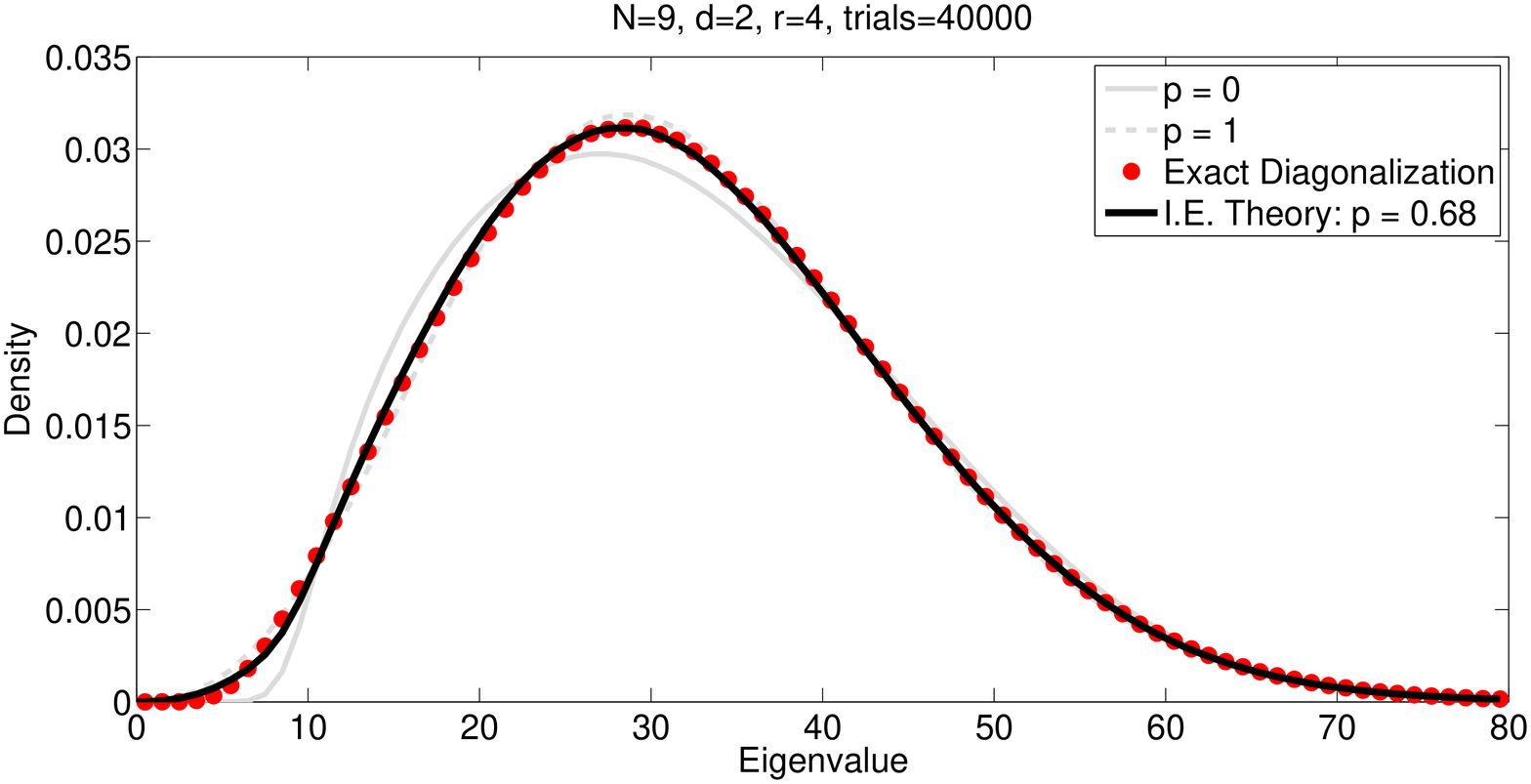}
\par\end{centering}

\begin{centering}
\includegraphics[scale=0.4]{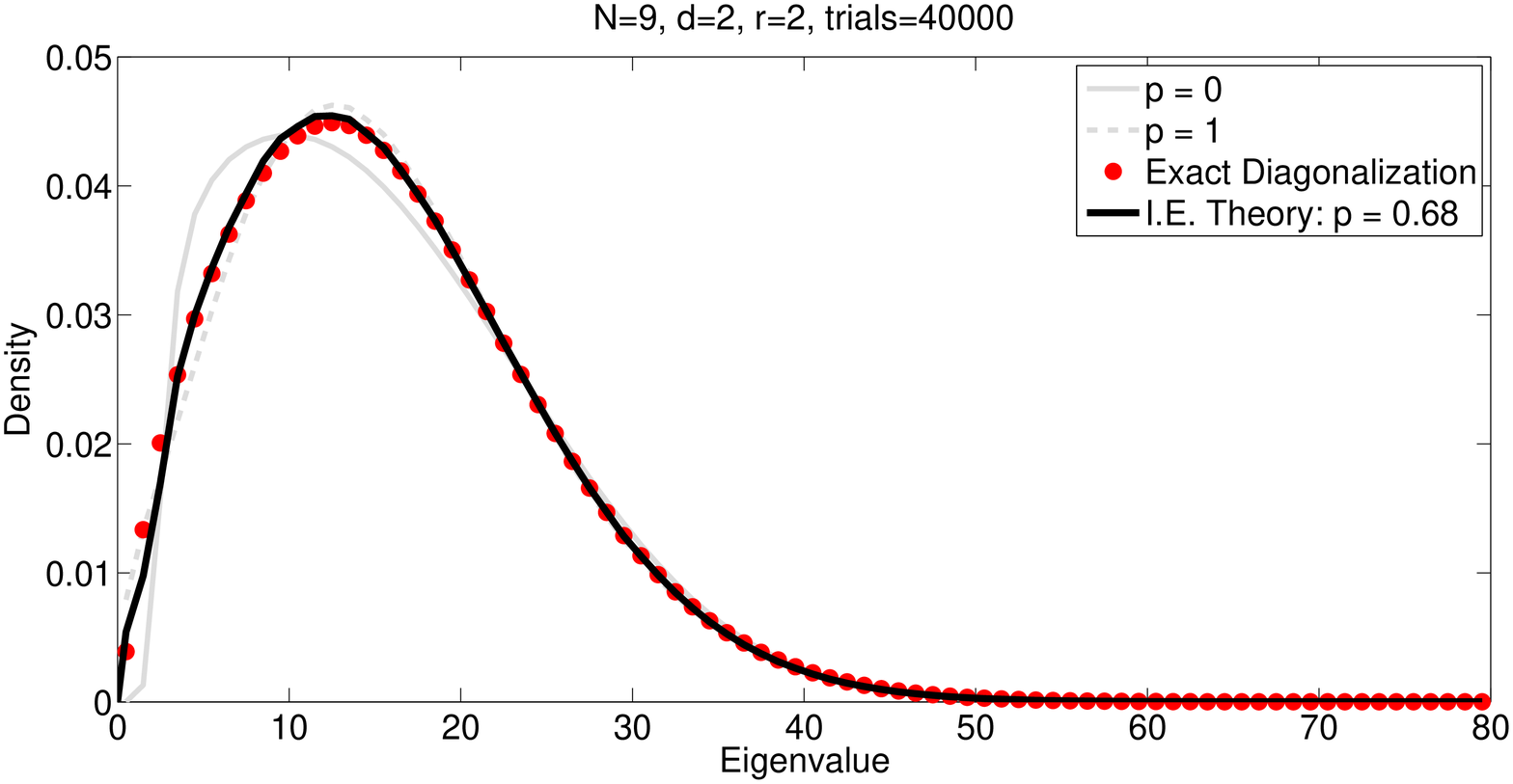}
\par\end{centering}

\centering{}\caption{$N=9$. Note that the two plots have the same $p$ despite having
different local ranks.}
\end{figure}

\begin{figure}[H]
\centering{}\includegraphics[scale=0.4]{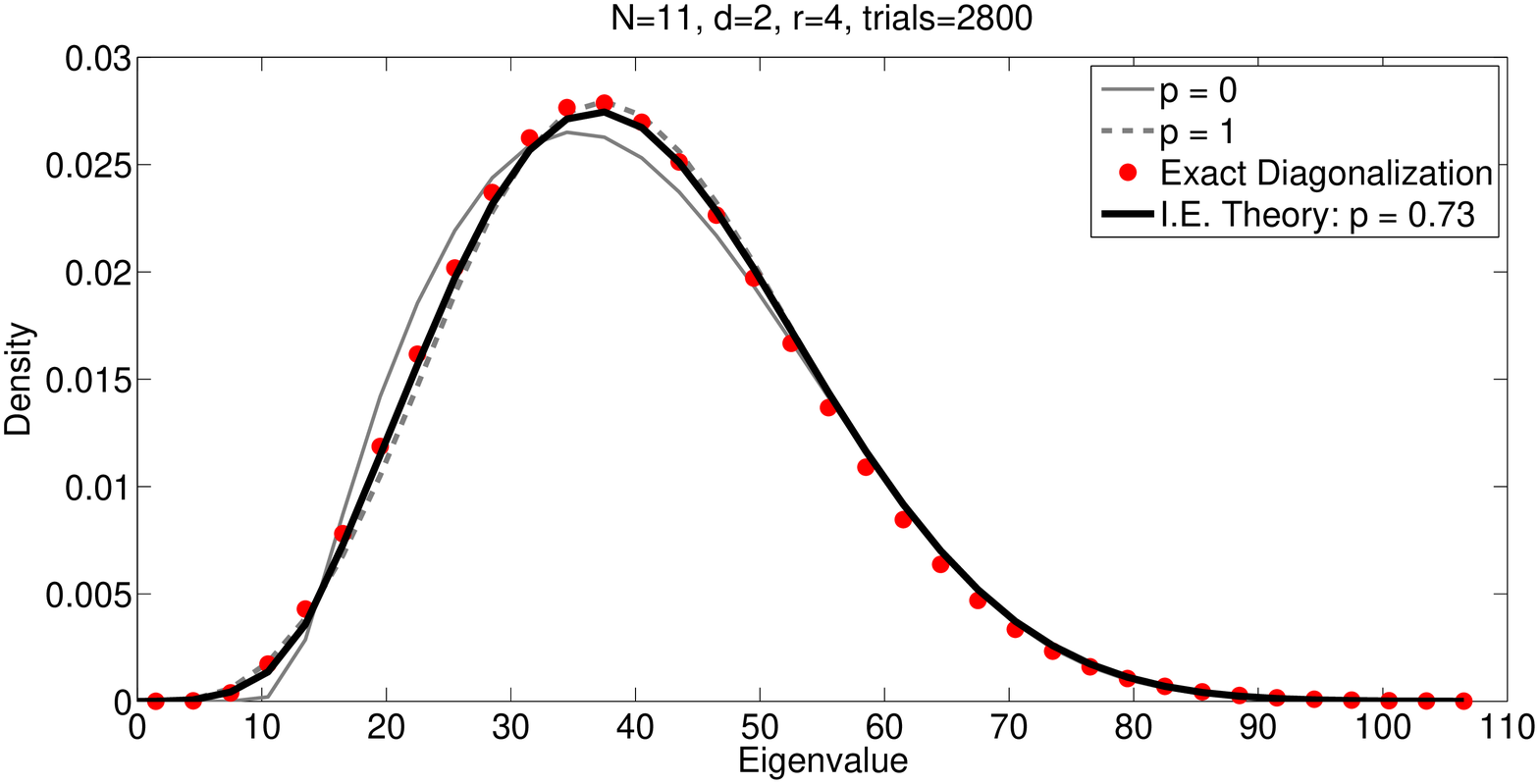}\caption{\label{fig:N=00003D11}$N=11$}
\end{figure}

\section{\label{sec:Other-Examples}Other Examples of Local Terms}

Because of the Universality lemma, $p$ is independent of the type
of local distribution. Furthermore, as discussed above, the application
of the theory for other types of local terms is entirely similar to
the Wishart case. Therefore, we only show the results in this section.
As a second example consider GOE's as local terms, i.e., $H_{l,l+1}=\frac{G^{T}+G}{2}$,
where $G$ is a full rank matrix whose elements are real Gaussian
random numbers. 

\begin{figure}[H]
\begin{centering}
\includegraphics[scale=0.4]{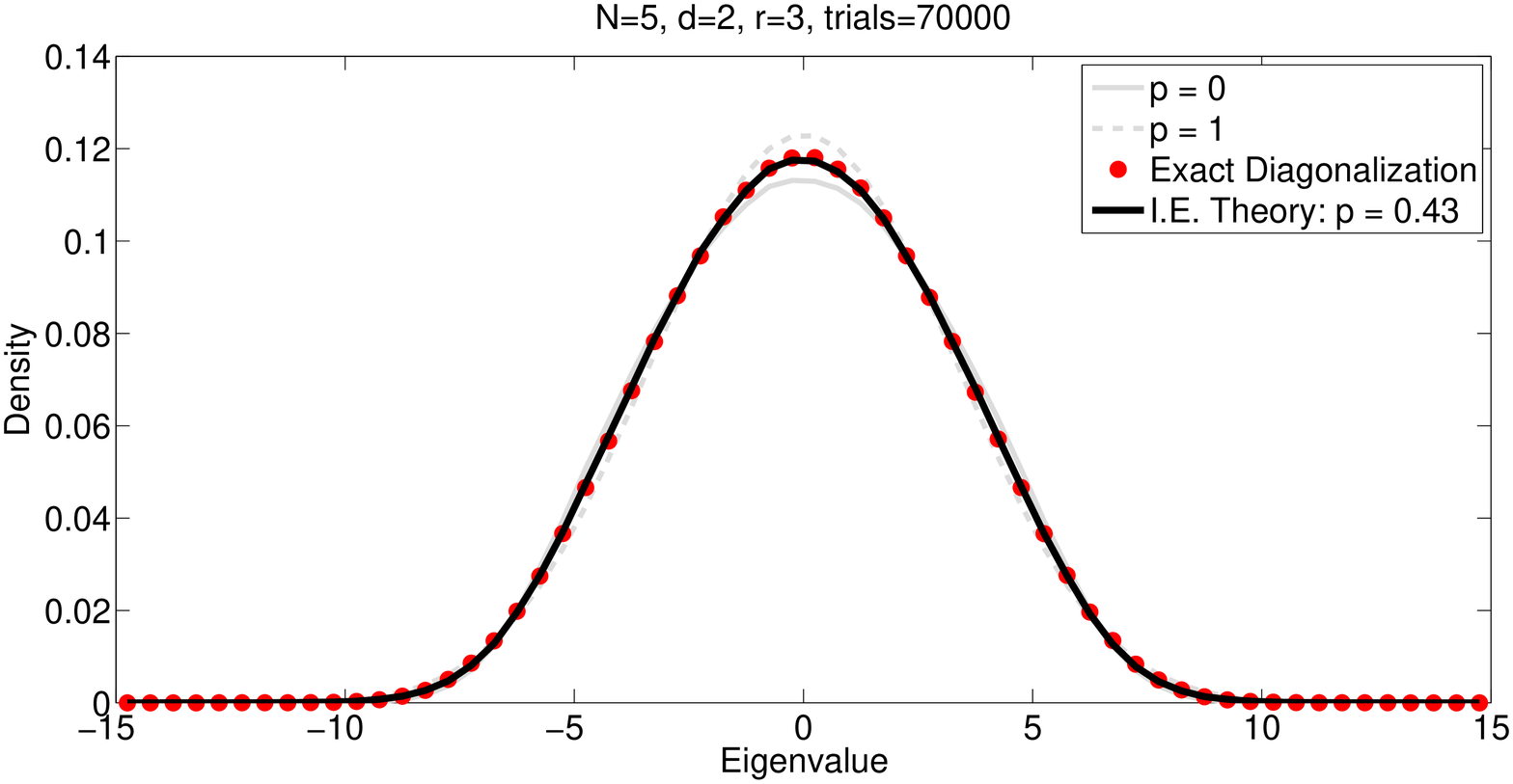}
\par\end{centering}

\begin{centering}
\includegraphics[scale=0.4]{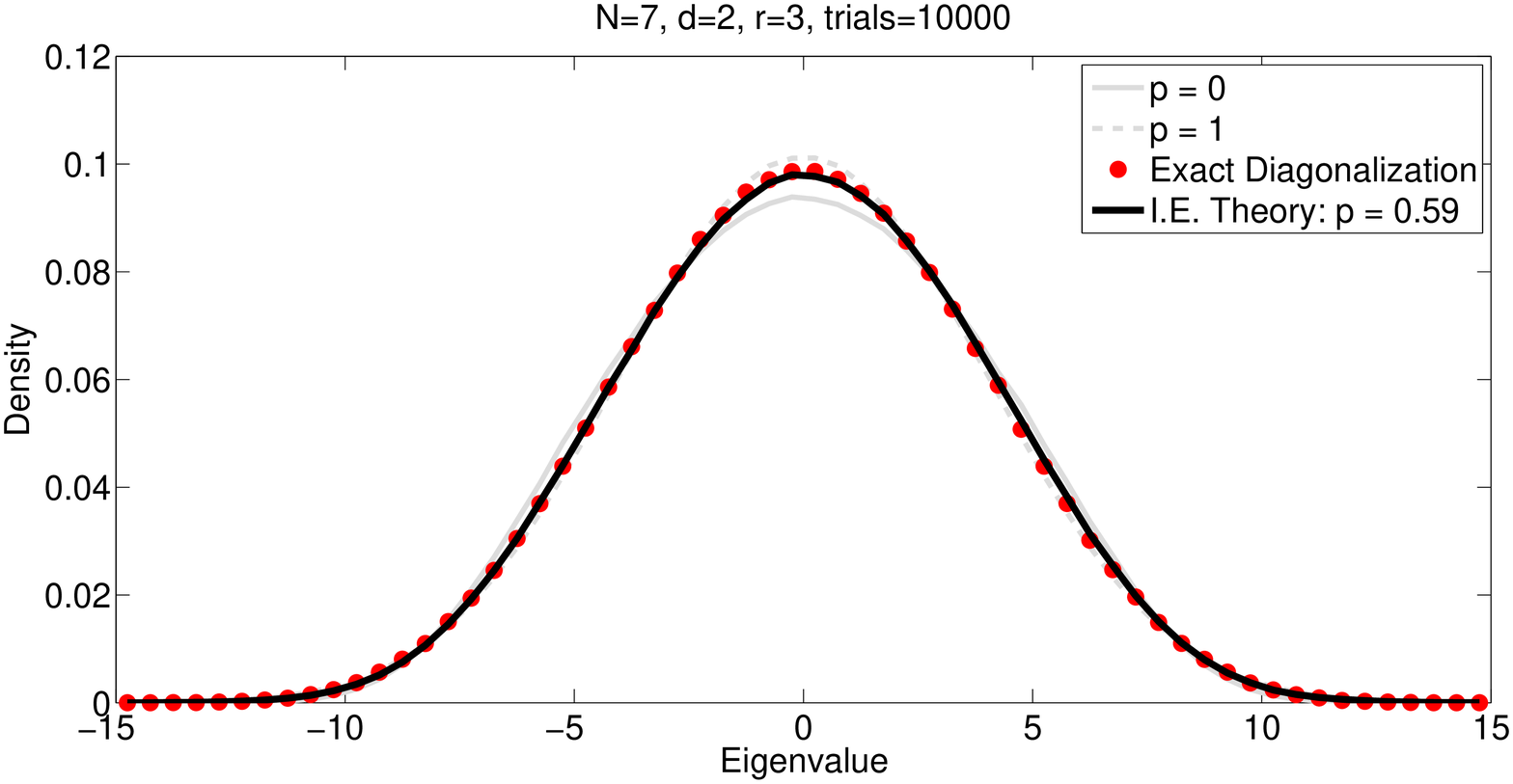}
\par\end{centering}

\begin{centering}
\includegraphics[scale=0.4]{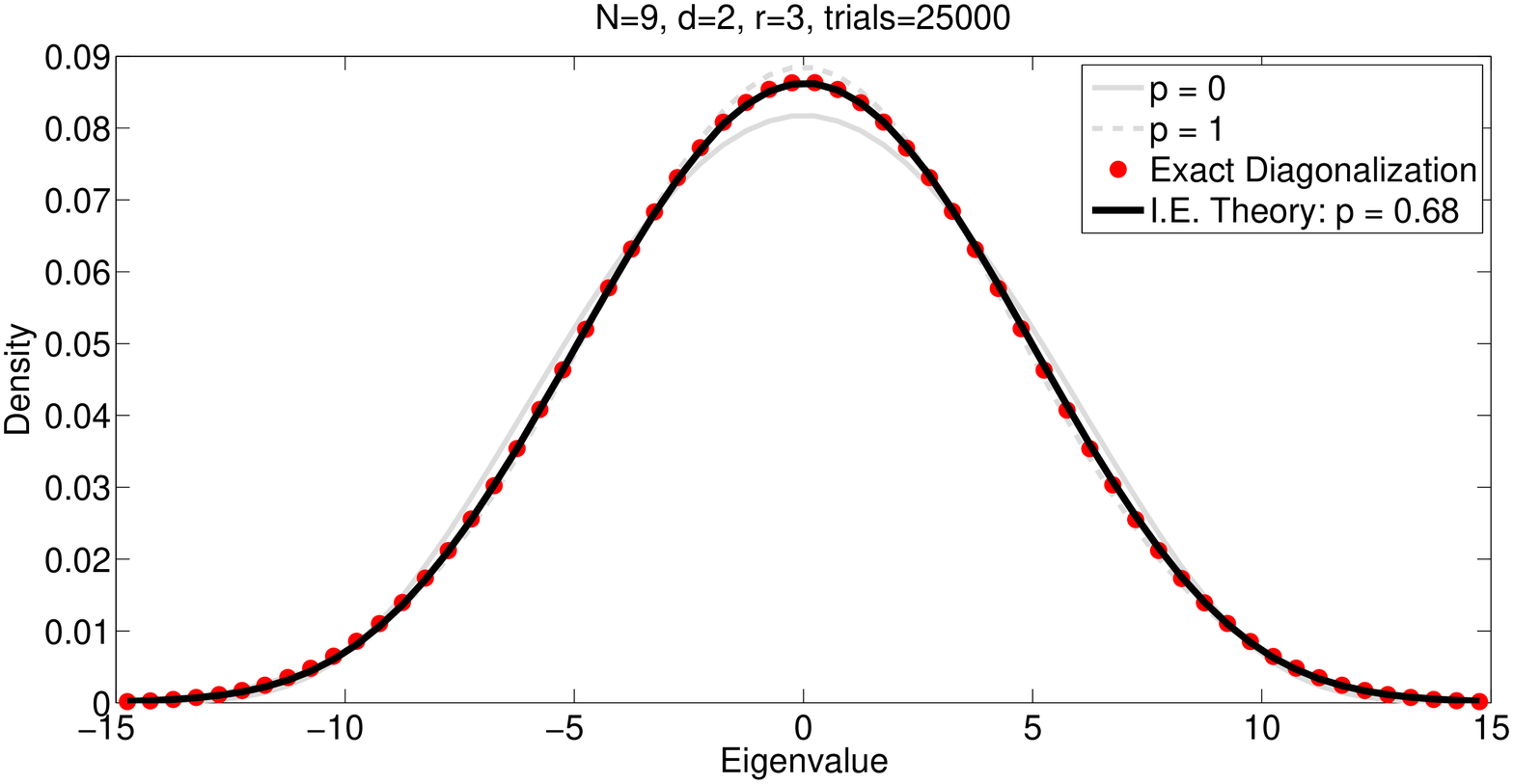}
\par\end{centering}

\caption{\label{fig:GOE's}GOE's as local terms}
\end{figure}

Lastly take the local terms to have Haar eigenvectors but with random
eigenvalues $\pm1$, i.e., $H_{l,l+1}=Q_{l}^{T}\Lambda_{l}Q_{l}$,
where $\Lambda_{l}$ is a diagonal matrix whose elements are binary
random variables $\pm1$ (Figure \ref{fig:binomial}).

\begin{figure}
\begin{centering}
\includegraphics[scale=0.4]{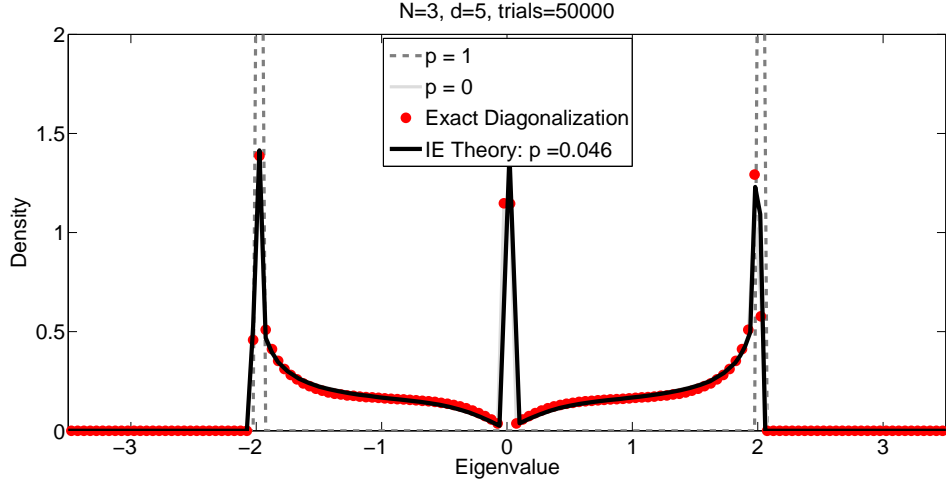}
\par\end{centering}

\caption{\label{fig:binomial}Local terms have a random binomial distribution.}
\end{figure}

In this case the classical treatment of the local terms leads to a
binomial distribution. As expected $p=1$ in Figure \ref{fig:binomial}
has three atoms at $-2,0,2$ corresponding to the randomized sum of
the eigenvalues from the two local terms. The exact diagonalization,
however, shows that the quantum chain has a much richer structure
closer to iso; i.e, $p=0$. This is captured quite well by IE with
$p=0.046$.

\section{\label{sub:Wishart-Beyond-Nearest-Neighbors}Beyond Nearest Neighbors
Interaction: $L>2$}

If one fixes all the parameters in the problem and compares $L>2$
with nearest neighbor interactions, then one expects the former to
act more isotropic as the number of random parameters in Eq. \ref{eq:Hamiltonian}
are more. When the number of random parameters introduced by the local
terms, i.e., $\left(N-L+1\right)d^{L}$ and $d^{N}$ are comparable,
we find that we can approximate the spectrum with a high accuracy
by taking the summands to be all isotropic\cite{ramisPI} (See Figures
\ref{fig:L=00003D3}-\ref{fig:L=00003D5}). 

\begin{figure}[H]
\begin{centering}
\includegraphics[scale=0.4]{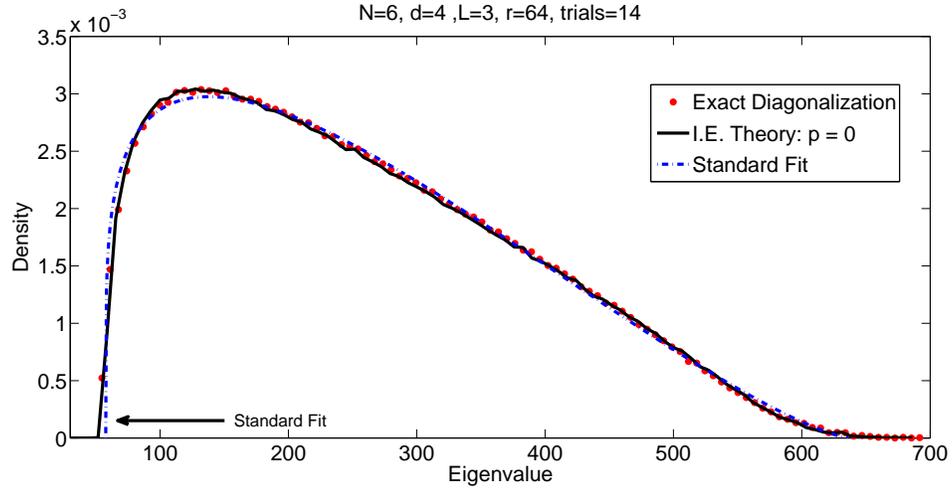}
\par\end{centering}

\caption{\label{fig:L=00003D3}IE method approximates the quantum spectrum
by $H^{IE}=\sum_{l=1}^{4}Q_{l}^{T}H_{l,\cdots,l+2}Q_{l}$}
\end{figure}

\begin{figure}[H]
\begin{centering}
\includegraphics[scale=0.4]{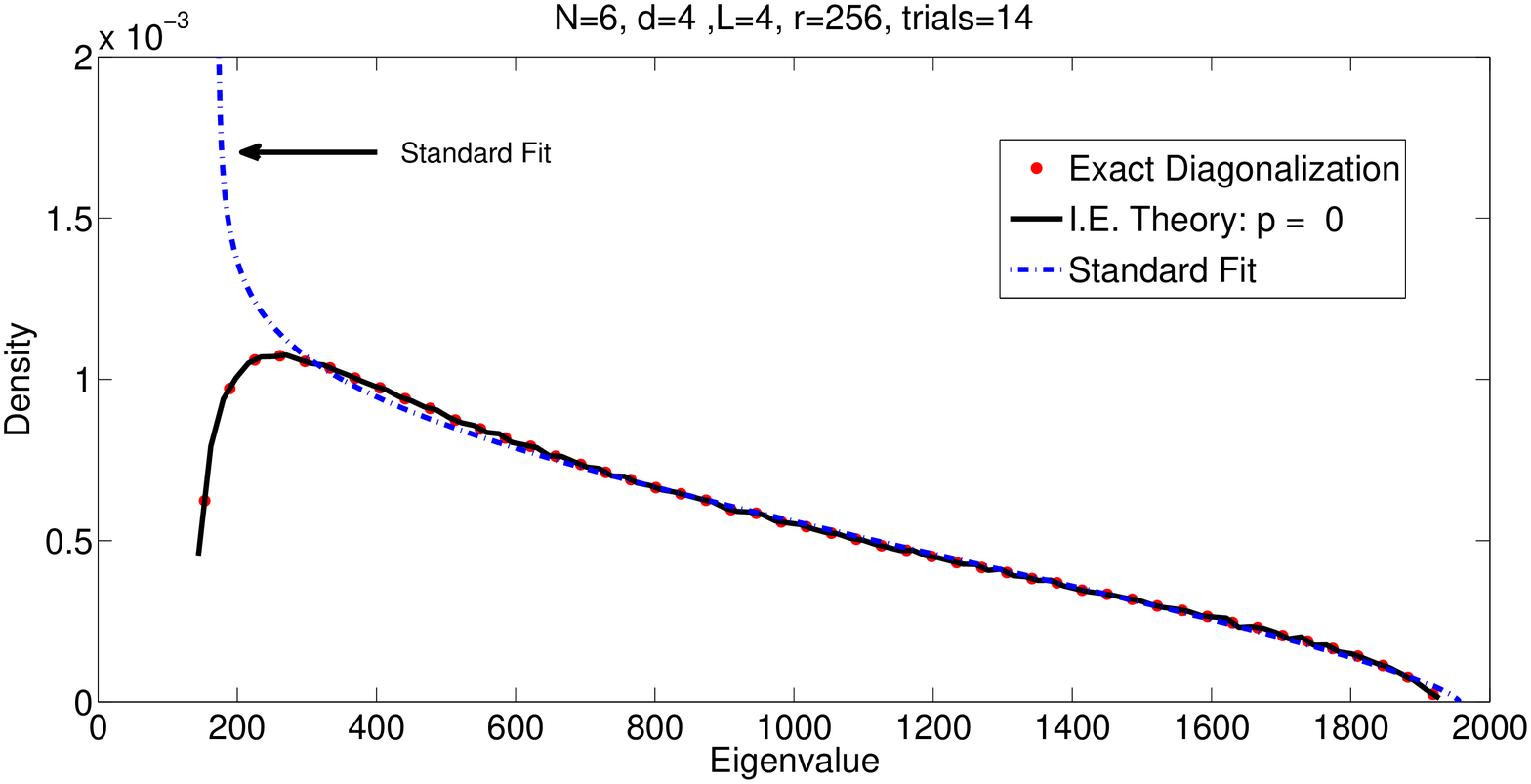}
\par\end{centering}

\caption{\label{fig:L=00003D4}IE method approximates the quantum spectrum
by $H^{IE}=\sum_{l=1}^{3}Q_{l}^{T}H_{l,\cdots,l+3}Q_{l}$. }
\end{figure}

\begin{figure}[H]
\begin{centering}
\includegraphics[scale=0.4]{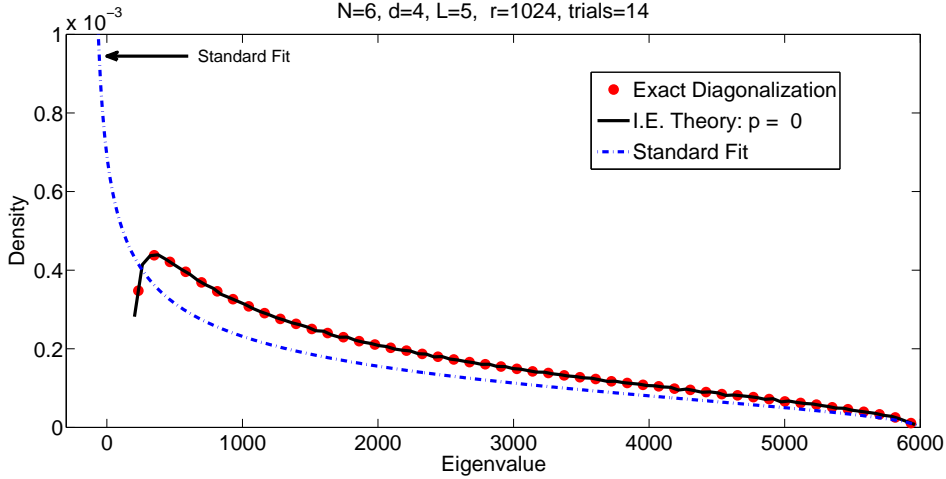}
\par\end{centering}

\caption{\label{fig:L=00003D5}IE method approximates the quantum spectrum
by $H^{IE}=\sum_{l=1}^{2}Q_{l}^{T}H_{l,\cdots,l+4}Q_{l}$}
\end{figure}

Most distributions built solely from the first four moments, would
give smooth curves. Roughly speaking, the mean indicates the center
of the distribution, variance its width, skewness its bending away
from the center and kurtosis how tall and skinny versus how short
and fat the distribution is. Therefore, it is hard to imagine that
the kinks, cusps and local extrema of the quantum problem (as seen
in some of our examples and in particular Figure in \ref{fig:binomial})
could be captured by fitting only the first four moments of the QMBS
Hamiltonian to a known distribution. It is remarkable that a one parameter
(i.e., $p$) interpolation between the isotropic and classical suffices
in capturing the richness of the spectra of QMBS.

\section{Conjectures and Open Problems}

In this paper we have offered a method that successfully captures
the density of states of QMBS with generic local interaction with
an accuracy higher than one expects solely from the first four moments.
We would like to direct the reader's attention to open problems that
we believe are within reach.
\begin{enumerate}
\item We conjecture that the higher moments may be analyzed for their significance.
For example, one can show that the fraction of departing terms in
the expansion of the higher moments (e.g. analogous to bold faced
and underlined terms in Eqs. \ref{eq:fourthMoments},\ref{eq:fifthmoments}
but for higher moments) is asymptotically upper bounded by $1/N^{3}$.
In Section\ref{sub:More-Than-Four} we conjectured that their expectation
values would not change the moments significantly. It would be of
interest to know if

\begin{eqnarray*}
\mathbb{E}\textrm{Tr}\left\{ \ldots Q^{-1}B^{\ge1}QA^{\ge1}Q^{-1}B^{\ge1}Q\ldots\right\}  & \le\\
\mathbb{E}\textrm{Tr}\left\{ \ldots Q_{q}^{-1}B^{\ge1}Q_{q}A^{\ge1}Q_{q}^{-1}B^{\ge1}Q_{q}\ldots\right\}  & \le\\
\mathbb{E}\textrm{Tr}\left\{ \ldots\Pi^{-1}B^{\ge1}\Pi A^{\ge1}\Pi^{-1}B^{\ge1}\Pi\ldots\right\}  & .
\end{eqnarray*}
For example, one wonders if

\[
\mathbb{E}\textrm{Tr}\left\{ \left(AQ^{-1}BQ\right)^{k}\right\} \le\mathbb{E}\textrm{Tr}\left\{ \left(AQ_{q}^{-1}BQ_{q}\right)^{k}\right\} \le\mathbb{E}\textrm{Tr}\left\{ \left(A\Pi^{-1}B\Pi\right)^{k}\right\} 
\]
for $k>2$; we have proved that the inequality becomes an equality
for $k=1$ (Departure Theorem) and holds for $k=2$ (Slider Theorem). 

\item Though we focus on decomposition for spin chains, we believe that
the main theorems may generalize to higher dimensional graphs. Further
rigorous and numerical work in higher dimensions would be of interest. 
\item At the end of Section \ref{sub:Isotropic-theory} we proposed that
more general local terms might be treated by explicitly including
the extra terms (Type I terms). 
\item Application of this method to slightly disordered systems would be
of interest in condensed matter physics. In this case, the assumption
of fully random local terms needs to be relaxed. 
\item In our numerical work, we see that the method gives accurate answers
in the presence of an external field. It would be nice to formally
extend the results and calculate thermodynamical quantities.
\item We derived our results for general $\beta$ but numerically tested
$\beta=1,2$. We acknowledge that general $\beta$ remains an abstraction. 
\item Readers may wonder whether it is better to consider ``iso'' or the
infinite limit which is ``free''. We have not fully investigated
these choices, and it is indeed possible that one or the other is
better suited for various purposes.
\item A grander goal would be to apply the ideas of this paper to very general
sums of matrices. 
\end{enumerate}
\pagebreak{}
\begin{acknowledgments}
We thank Peter W. Shor, Jeffrey Goldstone, Patrick Lee, Peter Young,
Gil Strang, Mehran Kardar, Salman Beigi, Xiao-Gang Wen and Ofer Zeitouni
for discussions. RM would like to thank the National Science Foundation
for support through grant CCF-0829421. We thank the National Science
Foundation for support through grants DMS 1035400 and DMS 1016125.\end{acknowledgments}

\section{Appendix}

To help the reader with the random quantities that appear in this
paper, we provide explanations of the exact nature of the random variables
that are being averaged. A common assumption is that we either assume
a uniformly randomly chosen eigenvalue from a random matrix or we
assume a collection of eigenvalues that may be randomly ordered, Random
ordering can be imposed or a direct result of the eigenvector matrix
having the right property. Calculating each of the terms separately
and then subtracting gives the same results.

\begin{equation}
\frac{1}{m}\mathbb{E}\left[\textrm{Tr}\left(AQ^{T}BQ\right)^{2}\right]=\frac{1}{m}\mathbb{E}\left\{ \sum_{1\leq i_{1},i_{2},j_{1},j_{2}\leq m}a_{i_{1}}a_{i_{2}}b_{j_{1}}b_{j_{2}}\left(q_{i_{1}j_{1}}q_{i_{1}j_{2}}q_{i_{2}j_{1}}q_{i_{2}j_{2}}\right)\right\} ,\label{eq:isoStart}
\end{equation}
where $a_{i}$ and $b_{j}$ are elements of matrices $A$ and$B$
respectively. The right hand side of Eq. \ref{eq:isoStart} can have
terms with two collisions (i.e., $i_{1}=i_{2}$ and $j_{1}=j_{2}$),
one collision (i.e. $i_{1}\neq i_{2}$ exclusive-or $j_{1}\neq j_{2}$),
or no collisions (i.e., $i_{1}\neq i_{2}$ and $j_{1}\neq j_{2}$).
Our goal now is to group terms based on the number of collisions.
The pre-factors for two, one and no collisions along with the counts
are summarized in Table \ref{tab:HaarExp}. Using the latter we can
sum the three types of contributions, to get the expectation

\begin{equation}
\begin{array}{c}
\frac{1}{m}\mathbb{E}\left[\textrm{Tr}\left(AQ^{T}BQ\right)^{2}\right]=\frac{\left(\beta+2\right)}{\left(m\beta+2\right)}\mathbb{E}\left(a^{2}\right)\mathbb{E}\left(b^{2}\right)+\\
\frac{\beta\left(m-1\right)}{\left(m\beta+2\right)}\left[\mathbb{E}\left(b^{2}\right)\mathbb{E}\left(a_{1}a_{2}\right)+\mathbb{E}\left(a^{2}\right)\mathbb{E}\left(b_{1}b_{2}\right)\right]-\frac{\beta\left(m-1\right)}{\left(m\beta+2\right)}\mathbb{E}\left(a_{1}a_{2}\right)\mathbb{E}\left(b_{1}b_{2}\right).
\end{array}\label{eq:Isotropic}
\end{equation}

\begin{flushleft}
If we take the local terms to be from the same distribution we can
further express the foregoing equation
\par\end{flushleft}

\begin{equation}
\frac{1}{m}\mathbb{E}\left[\textrm{Tr}\left(AQ^{T}BQ\right)^{2}\right]=\frac{1}{\left(m\beta+2\right)}\left[\left(\beta+2\right)m_{2}^{2}+\beta\left(m-1\right)\mathbb{E}\left(a_{1}a_{2}\right)\left\{ 2m_{2}-\mathbb{E}\left(a_{1}a_{2}\right)\right\} \right].\label{eq:isotropicGeneral}
\end{equation}

The quantity of interest is the difference of the classical and the
iso (see Eq. \ref{eq:denom}),

\begin{equation}
\begin{array}{c}
\frac{1}{m}\mbox{\ensuremath{\mathbb{E}}Tr}\left(AB\right)^{2}-\frac{1}{m}\mbox{\ensuremath{\mathbb{E}}Tr}\left(AQ^{T}BQ\right)^{2}=\\
\frac{\beta\left(m-1\right)}{\left(m\beta+2\right)}\left\{ \mathbb{E}\left(a^{2}\right)\mathbb{E}\left(b^{2}\right)-\mathbb{E}\left(b^{2}\right)\mathbb{E}\left(a_{1}a_{2}\right)-\mathbb{E}\left(a^{2}\right)\mathbb{E}\left(b_{1}b_{2}\right)+\mathbb{E}\left(a_{1}a_{2}\right)\mathbb{E}\left(b_{1}b_{2}\right)\right\} =\\
\frac{\beta\left(m-1\right)}{\left(m\beta+2\right)}\left\{ m_{2}^{(A)}m_{2}^{(B)}-m_{2}^{(B)}m_{1,1}^{(A)}-m_{2}^{(A)}m_{1,1}^{(B)}+m_{1,1}^{(A)}m_{1,1}^{(B)}\right\} =\\
\frac{\beta\left(m-1\right)}{\left(m\beta+2\right)}\left(m_{2}^{(A)}-m_{1,1}^{(A)}\right)\left(m_{2}^{(B)}-m_{1,1}^{(B)}\right)=
\end{array}\label{eq:denom_General}
\end{equation}

\begin{flushleft}
If we assume that the local terms have the same distribution: $m_{2}\equiv m_{2}^{(A)}=m_{2}^{\left(B\right)}$,
$m_{11}\equiv m_{1,1}^{(A)}=m_{1,1}^{\left(B\right)}$ as in Eq. \ref{eq:isotropicGeneral},
the foregoing equation simplifies to
\par\end{flushleft}

\[
\frac{1}{m}\mbox{\ensuremath{\mathbb{E}}Tr}\left(AB\right)^{2}-\frac{1}{m}\mbox{\ensuremath{\mathbb{E}}Tr}\left(AQ^{T}BQ\right)^{2}=\frac{\beta\left(m-1\right)}{\left(m\beta+2\right)}\left(m_{2}-m_{1,1}\right)^{2}.
\]

In the example of Wishart matrices as local terms we have

\begin{equation}
m_{1,1}\equiv\begin{array}{c}
m_{2}\equiv\mathbb{E}\left(a^{2}\right)=rk\left(rk+n+1\right)\\
\mathbb{E}\left(a_{1}a_{2}\right)=k\left(k-1\right)r^{2}+\frac{kr}{m-1}\left\{ \left(tn^{k-1}-1\right)(n+r+1)+tn^{k-1}\left(n-1\right)(r-1)\right\} \\
=k\left(k-1\right)r^{2}+\frac{kr}{m-1}\left\{ tn^{k-1}(nr+2)-n-r-1\right\} .
\end{array}\label{eq:elem}
\end{equation}

\end{document}